\documentclass[runningheads]{llncs}
\usepackage{amsmath,amssymb,graphicx,url,hyperref,listings,xcolor,xspace,enumerate,enumitem}  
\usepackage{lstcoq}
\usepackage[inference]{semantic}
\usepackage{mathpartir,mathtools}
\usepackage{textcomp,upgreek}

\definecolor{ltblue}{rgb}{0,0.4,0.4}
\definecolor{dkblue}{rgb}{0,0.1,0.6}
\definecolor{dkgreen}{rgb}{0,0.35,0}
\definecolor{dkviolet}{rgb}{0.3,0,0.5}
\definecolor{dkred}{rgb}{0.5,0,0}

\newcommand\ocamlnewstyle{\lstset{%
backgroundcolor = \color{white},%
language=caml,%
flexiblecolumns=false,%
commentstyle=\color{red},%
basicstyle=\ttfamily\tiny,%
escapeinside={/*@}{@*/},
numberstyle=\tiny,%
stepnumber=1,%
numbersep=15pt,%
extendedchars=true
breaklines=true,
showstringspaces=false,
mathescape=true,
numbers=left,
morekeywords={begin,end,match,try,assert,raise},
moredelim=**[is][\color{dkblue}]{£}{£},
tabsize=3,
sensitive=true,
breaklines=false,
 basicstyle=\scriptsize,%
captionpos=b,
basewidth={2em, 0.5em},
columns=[l]flexible,
identifierstyle={\ttfamily\color{black}},
keywordstyle=[1]{\ttfamily\bfseries\color{dkgreen}},
stringstyle=\ttfamily,
commentstyle={\ttfamily\itshape\color{dkblue}},
literate=
    {^}{${}^\wedge $}1
    {[]}{$[\,]$}1
    {->}{{$\rightarrow\;\;$}}1
    {++}{{\code{++}}}1
    {"}{{$"$}}1
}
}

\lstnewenvironment{ocaml}[1][]{
\ocamlnewstyle
	\lstset{#1}
}{}

\lstnewenvironment{ocamlno}[1][]{
\ocamlnewstyle
	\lstset{#1, numbers=none,basicstyle=\scriptsize}
}{}

\lstnewenvironment{ocamlsc}[1][]{
\ocamlnewstyle
	\lstset{#1, basicstyle=\scriptsize}
}{}

\lstnewenvironment{ocamlfn}[1][]{
\ocamlnewstyle
	\lstset{#1, basicstyle=\footnotesize}
}{}

\newcommand\ocamlfile[2][]{{
	\ocamlnewstyle
	\lstinputlisting[#1]{#2}}
}

\newcommand{\tr}{\triangleright}
\newcommand{\multiset}[1]{\{\hspace{-0.1em}| #1 |\hspace{-0.1em}\}}
\newcommand{\eqdef}{\;\stackrel{\text{\scriptsize def}}{=}\;}
\newcommand{\MPST}{{MPST}\xspace} 
\newcommand{\keyword}[1]{\textsf{\upshape\small #1}\xspace}
\newcommand{\partition}[1]{\textsf{\upshape\small minPartition\hspace{-.1em}}_{#1}}
\newcommand{\eNd}{\textsf{\upshape\scriptsize End\xspace}}
\newcommand\pequiv{\mathrel{\rlap{\raisebox{\fontdimen22\textfont2}{$=$}}
  \raisebox{-0.5\fontdimen22\textfont2}{$ = $}}}
\newcommand\wf{\textrm{WF}}
\newcommand\wb{\keyword{wb}}
\newcommand{\less}[1]{\backslash_{#1}}
\newcommand{\dom}{\operatorname{dom}}
\newcommand{\dual}{\overline}
\newcommand{\typeconst}[1]{\keyword{#1}}
\newcommand{\nat}{\typeconst{nat}}
\newcommand{\bool}{\typeconst{bool}}
\newcommand{\intk}{\keyword{int}}
\newcommand{\stringk}{\keyword{str}}
\newcommand{\unit}{\keyword{unit}}
\newcommand{\End}{\typeconst{end}}
\newcommand{\INACT}{\mathbf 0}
\newcommand{\PAR}{\hspace{.125em}\parallel\hspace{.125em}}
\newcommand{\PARI}{\hspace{.125em}\parallel_{i\in I}\hspace{.125em}}

\newcommand{\PIN}[3]{#1?#2(#3).}
\newcommand{\POUT}[3]{#1!#2\langle#3\rangle.}
\newcommand{\AIN}[3]{#1?#2(#3)}
\newcommand{\AOUT}[3]{#1!#2\langle#3\rangle}
\newcommand{\ASYN}[3]{#1@#2\bowtie#3}
\newcommand{\sync}[3]{#1@\pp #2\hspace{-.3em}\bowtie\hspace{-.3em}\pp #3} 
\newcommand{\ifk}{\keyword{if}}
\newcommand{\thenk}{\keyword{then}}
\newcommand{\elsek}{\keyword{else}}
\newcommand{\truek}{\keyword{true}}
\newcommand{\falsek}{\keyword{false}}
\newcommand{\lts}[1]
{ \setbox0=\hbox{$\ {\scriptstyle#1}\ $}
  \setbox1=\hbox{$\rightarrow$}
  \loop\setbox1=\hbox{$-$\kern-0.3em\unhbox1}\ifdim\wd1<\wd0\repeat
  \hbox{$\ \mathop{\box1}\limits^{#1}\ $}
}
\newcommand{\alts}[1]
{ \setbox0=\hbox{$\ {\scriptstyle#1}\ $}
  \setbox1=\hbox{$\rightarrow_a$}
  \loop\setbox1=\hbox{$-$\kern-0.3em\unhbox1}\ifdim\wd1<\wd0\repeat
  \hbox{$\ \mathop{\box1}\limits^{#1}\ $}
}
\newcommand{\dlts}[1]
{ \setbox0=\hbox{$\ {\scriptstyle#1}\ $}
  \setbox1=\hbox{$\rightarrow_d$}
  \loop\setbox1=\hbox{$-$\kern-0.3em\unhbox1}\ifdim\wd1<\wd0\repeat
  \hbox{$\ \mathop{\box1}\limits^{#1}\ $}
}
\newcommand{\nlts}[1]
{ \setbox0=\hbox{$\ {\scriptstyle#1}\ $}
  \setbox1=\hbox{$\nrightarrow$}
  \loop\setbox1=\hbox{$-$\kern-0.3em\unhbox1}\ifdim\wd1<\wd0\repeat
  \hbox{$\ \mathop{\box1}\limits^{#1}\ $}
}
\newcommand{\Lts}{\Longrightarrow}
\newcommand\pp[1]{{\color{blue}\tt #1}}
\newcommand{\ppp}{\pp p}
\newcommand{\ppq}{\pp q}
\newcommand{\ppr}{\pp r}
\newcommand{\pps}{\pp s}
\newcommand{\ppc}{\pp c}
\newcommand{\ppa}{\pp a}
\newcommand{\login}{\keyword{login}}
\newcommand{\pwd}{\keyword{pwd}}
\newcommand{\ssh}{\keyword{ssh}}
\newcommand{\cancel}{\keyword{cancel}}
\newcommand{\auth}{\keyword{auth}}
\newcommand{\quit}{\keyword{quit}}
\newcommand{\fail}{\keyword{fail}}
\newcommand{\alogin}{{\tt login}}

\newcommand{\assh}{{\tt ssh}}
\newcommand{\acancel}{{\tt cancel}}

\newcommand{\myparagraph}[1]{\medskip\noindent\textbf{\emph{#1}.} }

\begin{document}  

\title{Iso-Recursive Multiparty Sessions and their Automated Verification}
\subtitle{Technical Report}
\author{Marco Giunti\orcidID{0000-0002-7582-0308} \and Nobuko Yoshida\orcidID{0000-0002-3925-8557}}
\institute{University of Oxford, UK}
\authorrunning{M. Giunti \and N. Yoshida}

\maketitle 

\begin{abstract}
Most works on session types take an equi-recursive approach
and do not distinguish among a recursive type and its unfolding.
This becomes more important in recent type systems  which do not require global types, also known as
\emph{generalised multiparty session types} (GMST).
In GMST, in order to establish properties as deadlock-freedom, 
the environments which type processes are assumed to satisfy extensional 
properties holding in all infinite sequences.
This is a problem because:
(1) the mechanisation of GMST and equi-recursion in proof assistants 
is utterly complex and eventually requires
co-induction; and 
(2) the implementation of GMST in type checkers relies on model checkers 
for environment verification, and thus the program analysis is not self-contained.

In this paper, 
we overcome these limitations
by providing an \emph{iso-recursive typing system} 
that computes the \emph{behavioural properties of environments}.
The type system relies on a terminating function named \emph{compliance}
that computes all  final redexes of an environment, and determines when these redexes do not contain 
mismatches or deadlocks: \emph{compliant environments cannot go wrong}.
The function is defined theoretically by introducing the novel notions of 
deterministic LTS of environments and of environment closure, 
and can be implemented in mainstream programming languages and compilers.
 We showcase an implementation in OCaml by using exception handling to tackle the inherent non-determinism 
 of synchronisation of branching and selection types. 
We assess that the implementation provides the desired properties,
namely absence of mismatches and of deadlocks in environments,
by resorting to automated 
deductive verification  performed in tools of the OCaml ecosystem relying on Why3.
\end{abstract}

\tableofcontents
\listoffigures
\newpage

\section{Introduction}
\label{sec:intro}
\emph{Session types}~\cite{Honda93,TakeuchiHK94,HondaVK98} are an 
effective method to control the behaviour of software components that run in message-passing 
distributed systems.
\emph{Multiparty session types} (\MPST)~\cite{HondaYC08,HondaYC16} enhance session types by 
providing support for sessions involving multiple participants, thus representing more 
expressive scenarios.
Various theories of \MPST have been deployed in 
programming languages~\cite{YoshidaPL24} allowing
verification of industrial code at compile or run-time~\cite{betty17}.

In most works on session types, recursive types follow an \emph{equi-recursive} 
view~\cite{Pierce02} 
and represent infinite trees that are manipulated {co-inductively}. 
This representation does not have a direct 
counterpart in non-lazy programming languages, which typically resort to 
\emph{iso-recursive} types~\cite{theoryOfObjects96,Pierce02}
 that are 
manipulated {inductively}. 
Moreover, lazy evaluation of predicates on equi-recursive 
trees might not terminate, and is thus not effective for static program analysis.
In practice, \MPST are embedded in non-lazy languages by encoding equi-recursive types;
for instance,~\cite{ImaiNYY20} defines infinite sequence of types as polymorphic lenses~\cite{FosterGMPS07}
 by using OCaml generalised algebraic data types. 

Our proposal to overcome this problem consists in introducing a theory of 
\emph{iso-recursive multiparty session types}
relying on a type system that \emph{computes the deadlock-freedom} of type environments.

Lately, there have been several advances in \MPST 
that can establish deadlock-freedom 
without using global types,
e.g.~\cite{LangeNTY18,ScalasYB19,ScalasY19,BarwellSY022,ampst23,PetersY24,00020DG21,Brun2023}: 
this bottom-up approach is known as \emph{generalised multiparty
session types (GMST)}. 
However, the price to pay in GMST is that environments must satisfy extensional 
predicates requiring that  
a certain property holds for all infinite sequences.
This is a non-integrated feature in GMST, which resort to external tools as model checkers
to assess these predicates.
Moreover, mechanising equi-recursive GMST in proof assistants  is quite complex,
and eventually relies on co-induction~\cite{EkiciY24}.
Specifically, formulations based on GMST
are difficult to implement in programming languages because of the
interplay among 
equi-recursive types and the verification of the semantic properties of environments.
Another possibility is to proceed top-down by using global types 
and ensure deadlock freedom without verifying environments, 
while the analysis' expressiveness is affected by  
projectability~\cite{UdomsrirungruangY25}. 

In this paper, we propose a formal system to compute the deadlock-freedom of type environments
\emph{compositionally} at the typing of parallel processes, and we provide an implementation that is 
automatically verified  by using 
automated deductive tools of the OCaml ecosystem~\cite{PereiraR20,Pereira24,ChargueraudFLP19}
relying on Why3~\cite{FilliatreP13}.

\subsection{Equi-recursive vs Iso-recursive Types: SSH/OAuth2 Example}\label{ex:auth}
We illustrate our methodology using a recursive variant of the  OAuth
2.0 protocol (cf.~\cite{ScalasY19}) which provides support for 
\emph{ssh}~\cite{devlog}.
Let us indicate \emph{send to} and \emph{receive from} participant \pp p 
as the \emph{types} $\pp p!l(S).T$ and $\pp p?l(S).T$, respectively, 
where $l$ is a \emph{label} indicating the nature of the communication, $S$ 
is the \emph{sort} of the payload, and $T$ is the type of the continuation.
Selection among (branching on) different output (input) types is done by
means of the binary operator $+$.
Recursion is provided by the construct $\mu X.T$, which binds the type variable $X$ in~ $T$.
Termination is represented by type $\End$.
\emph{Sorts} describe the types of string, boolean, and unit values.
The session types of the service (\pps), of the client (\ppc), and of the authorisation server 
(\ppa) are: 
\begin{align*}
T_\pps &\eqdef \mu X.(\ppc!\login(\unit).\ppa?\auth(\bool).X +
        \ppc!\cancel(\unit).\End)   
\\ 
T_\ppc &\eqdef \mu X.(\pps?\login(\unit).
(\ppa!\pwd(\stringk).X + \ppa!\ssh(\unit).X)\, +
\\
&\hspace{4.3em}        
\pps?\cancel(\unit).\ppa!\quit(\unit).\End) 
\\ 
T_\ppa &\eqdef \mu X.R_\ppa
\\
R_\ppa &\eqdef\ppc?\pwd(\stringk).\pps!\auth(\bool).X +
\ppc?\ssh(\unit).\pps!\auth(\bool).X +
\ppc?\quit(\unit).\End  
\end{align*}    
The protocol says that the service (\pps) sends to the client (\ppc)
either a request to \verb|login|, or \verb|cancel|; in
the first case, \ppc\ continues by sending the password (\verb|pwd|, carrying a string),
or by sending \verb|ssh|, to \ppa,  who
in turn sends authentication to \pps\ (\verb|auth|, with a boolean, telling whether the client is authorised), 
and the session
restarts; in the second case, \ppc\ sends \verb|quit| to \ppa, and the session ends.

\textbf{A  problem} of equi-recursive GMST, 
e.g.~\cite{ScalasY19},
is that types are defined co-inductively (cf.~\cite{EkiciY24}).
Recursive types can be infinitely folded and unfolded:
for instance, we have the following \emph{equi-recursive equations}:
\begin{align*}
  T_\ppa &= T^*_\ppa = T^{**}_\ppa = \cdots
  \\
  T^*_\ppa &\eqdef 
\ppc?\pwd(\stringk).\pps!\auth(\bool).T_\ppa +
\ppc?\ssh(\unit).\pps!\auth(\bool).T_\ppa + 
\ppc?\quit(\unit).\End
        \\
T^{**}_\ppa &\eqdef 
\ppc?\pwd(\stringk).\pps!\auth(\bool).T^*_\ppa +
\ppc?\ssh(\unit).\pps!\auth(\bool).T^*_\ppa + 
        \ppc?\quit(\unit).\End 
\end{align*}

This is particularly relevant when establishing the properties of the typing system,
e.g. safety~\cite[Definition 4.1]{ScalasY19},
which are based on a notion of \emph{transition of session environments}.
To illustrate, the idea is to interpret \emph{types as processes}, cf.~\cite{Milner89}, 
and consider transitions of  \emph{session environments} mapping participants~\pp p to types $T$.
The environment $\Gamma_1\eqdef\pps \colon \ppc!\login(\unit).\ppa?\auth(\bool).T_\pps +
\ppc!\cancel(\unit).\End$ can fire an output action $\pp \ppc!\login(\unit)$ and reach 
$\pps \colon \ppa?\auth(\bool).$ $T_\pps$, or can fire an output action $\pp \ppc!\cancel(\unit)$
and reach $\pps \colon\End$.
The environment 
$\Gamma_2\eqdef\ppc \colon \pps?\login(\unit).T'_\ppc + \pps?\cancel(\unit).\ppa!\quit(\unit).\End$ 
can fire an input action $\pps?\login(\unit)$ and reach  
$\ppc \colon T'_\ppc$, where $T'_\ppc\eqdef \ppa!\pwd(\stringk).T_\ppc + \ppa!\ssh(\unit).T_\ppc$,
or can fire an input action $\pps?\cancel(\unit)$ and reach $\ppc\colon \ppa!\quit(\unit).\End$.
The environment $\Gamma_1,\Gamma_2$ can fire two synchronisation actions:
(1) $\sync \login s c$, which indicates a synchronisation on the label \login by
the input participant \pps\ and the output participant \ppc, and reach the environment 
$\pps \colon \ppa?\auth(\bool).T_\pps, \ppc \colon T'_\ppc$;
and (2) $\sync\cancel s c $, and reach  the environment
$\pps \colon\End, \ppc\colon \ppa!\quit(\unit).\End$.

In particular, the rule for recursive types in~\cite[Definition 2.8]{ScalasY19} 
states that a recursive type $\mu X.T$ inherits the transitions from its unfolding, 
that is the type $T\{\mu X.T/X\}$. 
For instance, the rule can be instantiated with type $T_\ppa$ as\\[1mm]
\centerline{
$
\inference[\textsc{[$\Gamma$-$\mu$]}]{
\Gamma,\ppa \colon T^*_\ppa\lts\alpha \Gamma' 
}{
\Gamma,\ppa \colon T_\ppa\lts\alpha \Gamma'
}
$
}\\[1mm]
and allows for inferring the following transitions:\\[1mm]
\centerline{
$\begin{array}{l}
\Gamma,\ppa \colon T_\ppa
\lts{\ppc?\assh({\tt unit})}
\Gamma,\ppa \colon \pps!\auth(\bool).T_\ppa\\
\Gamma,\ppa \colon T_\ppa
\lts{\ppc?\assh({\tt unit})}
\Gamma,\ppa \colon \pps!\auth(\bool).T^*_\ppa \ \cdots
\end{array}
$}\\[1mm]
We note that this elegant approach is appropriate for the theory, 
but less suited for mechanising GMST in theorem provers, and for automated verification.

More specifically, the approach introduced in~\cite{ScalasY19} and followed in many subsequent 
papers on GMST, e.g.~\cite{LangeNTY18,ScalasY19,BarwellSY022,ampst23,PetersY24,00020DG21,Brun2023}, 
requires to type check sessions with environments having certain \emph{extensional properties}.
Crucially, such properties  must be established \emph{before typing} by analysing all possible infinite transitions
of session environments.
To illustrate, the paper~\cite{ScalasY19}
provides a companion artefact by using~mCRL2~\cite{GrooteM2014} featuring
$\mu$-calculus formulae that represent the safety and deadlock freedom properties of environments~\cite[Figure 5]{ScalasY19},
which are defined by least and greatest fixed points. 

\myparagraph{Our solution}{In this paper}, we 
\emph{compute the semantic properties} of 
the session environment in the rule for \emph{type checking the session composition},
hence achieving \emph{decidable type checking}.
This is possible because our types are iso-recursive and have a finite structure.

In our setting, the types $T_\ppa$, $T^*_\ppa$, and $T^{**}_\ppa$ are all different, 
but isomorphic.
A recursive type $\mu X.T$ can only be used to type-check a recursive process, or a type variable.
To type check an input or output process, we need to unfold $\mu X.T$
by applying the substitution $\mu X.T /X$ to type $T$, denoted as 
$T\{\mu X.T /X\}$.
That is, we have the following \emph{iso-recursive equations}:\\[1mm]
\centerline{
$
\begin{array}{lllll}
T^*_\ppa &= R_\ppa\{T_\ppa / X\} 
&\quad
T^{**}_\ppa &= R_\ppa\{T^*_\ppa / X\}=
R_\ppa\{(R_\ppa\{T_\ppa / X\}) / X\} 
& \cdots
\end{array}
$}\\[1mm]
For instance, consider the authorisation process $Q_\ppa$ below,
whose syntax mirrors the one of its type $T^*_\ppa$, but that 
the payload of input and output are (bound) variables and expressions, respectively,
and that there is a process variable~$\chi$:
\begin{align*}
  Q_\ppa&\eqdef  \PIN \ppc\pwd x \POUT \pps \auth\falsek \chi +
  \PIN \ppc\ssh x \POUT \pps\auth\truek \chi +
  \PIN \ppc\quit x \INACT
\end{align*}
In order to type check the recursive process $\mu \chi.Q_\ppa$ we use the typing judgement:
\begin{mathpar}
\inference[\textsc{T-Rec}]{
\chi\colon T_\ppa\vdash Q_\ppa\colon T^*_\ppa
}{
\emptyset\vdash  \mu \chi.Q_\ppa\colon T_\ppa
}
\end{mathpar}  
Now consider the process obtained by substituting $\chi$ in $Q_\ppa$ with $\mu \chi.Q_\ppa$, 
that is process $P^*_\ppa\eqdef Q_\ppa \{\mu \chi.Q_\ppa/ \chi\}$, 
and the parallel execution of $\ppa\lhd P^*_\ppa$ 
with $\pps\lhd P_\pps$ and $\ppc\lhd P_\ppc$,
where  $P_\pps, P_\ppc$ are recursive processes implementing 
the service $\pps$ typed by $T_\pps$,
and the client $\pp c$ typed by $T_\ppc$, respectively.
This    session should be accepted by the type system, since at runtime it behaves correctly,
independently of the fact that the authorisation service $P^*_\ppa$  
has been \mbox{``unrolled'' once.}

That is, we want to infer the following judgement by using the
 rule for session composition of the typing system for sessions, denoted $\Vdash$:
\begin{mathpar}
\inference[\textsc{T-Ses}]{
\emptyset\vdash P_\pps\colon T_\pps\quad
\emptyset\vdash P_\ppc\colon T_\ppc\quad 
\emptyset\vdash P^*_\ppa\colon T^*_\ppa\\
\Delta = \pps \colon T_\pps, \pp c\colon T_\ppc, \ppa\colon T^*_\ppa \quad 
\keyword{comp}(\Delta)
}
{
\emptyset\Vdash 
\pps\lhd P_\pps \PAR 
\ppc\lhd P_\ppc \PAR
\ppa\lhd P^*_\ppa
\rhd
\Delta
}
\end{mathpar}
The predicate $\keyword{comp}(\Delta)$ establishes \emph{compliance}
by using the  \emph{computable \mbox{function}} $\keyword{comp}$.
The goal is to calculate all possible \emph{final environments} that
are reachable 
from~$\Delta$, and verify that they are not errors.
Intuitively, an environment is final when is stuck, or when it has 
already been encountered, reaching a fixed point.

Since we are interested in mechanising compliance, the calculation should be
achieved by relying on the novel notion of \emph{deterministic transition}, denoted $\dlts{}$, such that
$\Delta \dlts{\alpha_1} \Delta_1$ and 
$\Delta \dlts{\alpha_2} \Delta_2$ imply $\alpha_1=\alpha_2$ and $\Delta_1=\Delta_2$.
The key point is that a deterministic transition system can be encoded as a computable function
that  can be  deployed in type checkers and compilers.
Moreover, the properties of the function can be verified with  automated deductive verification tools 
as Why3~\cite{FilliatreP13}.
In particular, we propose  the idea of \emph{closure of an environment}~$\Delta$:   
the function receives $\Delta$ in input  and returns in output a finite set of final environments 
reachable from $\Delta$ by multiple applications of~$\dlts{}$.

The compliance function decides when in all final environments reached by  transitions starting from $\Delta$,
there is not a \emph{communication mismatch} or a \emph{deadlock}.
A communication mismatch arises when a participant 
\pp p has a single I/O type receiving from/sending to participant \pp q,
and \pp q has a single I/O type receiving from/sending to  participant \pp p,
and one of the following cases arise: 
(i) both~\pp p and \pp q are sending or receiving;
(ii) the intersection among the labels used by \pp p and \pp q is empty;
(iii) \pp p and \pp q agree on a label but disagree on the label's sort.
A deadlock arises when the environment $\Delta$ cannot fire any transition and there is at least a participant
\pp p s.t. $\Delta(\pp p)\ne\End$.

\medskip
To see an example of environment rejected by the compliance function \keyword{comp}, consider
$\Delta''$ below.
An authorisation server  typed by $T''_\ppa$ only allows two subsequent attempts for \verb|ssh| authentication:
after that, it ends.
Conversely, a client typed by $T_\ppc$ performs an infinite number of requests of \verb|ssh| authentication:
for this very reason, a system typed by $\Delta''$ can deadlock and must be rejected.
\begin{align*}
T'_\ppa &\eqdef \mu X. (\ppc?\pwd(\stringk).\pps!\auth(\bool).X +
\ppc?\ssh(\unit).\pps!\auth(\bool).\End +
\\
&\hspace{4.2em}
\ppc?\quit(\unit).\End ) 
\\
T''_\ppa &\eqdef \ppc?\pwd(\stringk).\pps!\auth(\bool).T_\ppa +
\ppc?\ssh(\unit).\pps!\auth(\bool).T'_\ppa +
\ppc?\quit(\unit).\End
\\
\Delta''&\eqdef   
\pps\colon T_\pps, \pp c\colon T_\ppc, \ppa\colon T''_\ppa         
\end{align*}
The closure of $\Delta''$ does return a set of environments  containing the
deadlocked environment $\Delta_{{\tt lock}}$, which depicts the scenario discussed above.
Since  $\Delta_{{\tt lock}}\in\keyword{closure}(\Delta'')$  and 
$\Delta_{{\tt lock}}$ is a deadlock, we have
$\neg\,\keyword{comp}(\Delta'')$:\\[1mm]
\centerline
{$\Delta_{{\tt lock}}\eqdef
 \pps\colon \ppa?\auth(\bool).T_\pps,
 \pp c\colon \ppa!\pwd(\stringk).T_\ppc + \ppa!\ssh(\unit).T_\ppc,
 \ppa\colon \End
 $
}

\myparagraph{Outline}
\textbf{\S~\ref{sec:syntax}} introduces the syntax and semantics of multiparty
sessions.
\textbf{\S~\ref{sec:alg-lts}} presents the non-deterministic labelled transition semantics of
session environments (cf.~\textbf{\S~\ref{sec:lts}}),
and its deterministic counterpart (cf.~\textbf{\S~\ref{sec:dlts}}):
the former is used to define deadlocks and to prove subject reduction;
the latter is used in~\textbf{\S~\ref{sec:closure}} 
to define closure and in turn to mechanise compliance.
\textbf{\S~\ref{sec:typing}} introduces  
the typing system. 
We first analyse the typing rules for processes.
Second, we analyse the rule for typing sessions, which relies on a computable function
calculating compliance that is defined  
in~\textbf{\S~\ref{sec:compliance}}.
Last, in~\textbf{\S~\ref{sec:sr-sketch}} we provide the proof of subject reduction and we state a progress result.
\textbf{\S~\ref{sec:implementation}} is devoted to the automated deductive 
verification of compliance.
We start in \textbf{\S~\ref{sec:code}} 
by  outlining few details of the implementation of the 
closure of deterministic transitions.
\textbf{\S~\ref{sec:verification}} verifies the behavioural specification of the 
implementation in automated deductive verification tools of the OCaml ecosystem
relying on Why3.
\textbf{\S~\ref{sec:discussion}} concludes by presenting related work
and next directions.
The full proofs and omitted definitions can be found
in appendix
and the accompanying artefact can be found at \url{https://doi.org/10.5281/zenodo.14621028}. 

\section{Multiparty Sessions}
\label{sec:syntax}
The syntax of types and processes is in Definition~\ref{fig:syntax}.
We consider \emph{iso-recursive} types of the form $\mu X.T$ where $\mu X.T$ 
and its unfolding are not equal, but isomorphic. 
We stress that types have a \emph{finite representation} rather than abstract infinite trees
(cf.~equi-recursive types).

\begin{definition}[Syntax of types and processes]\label{fig:syntax}
  \begin{align*}
    S &:= \nat \mid \intk \mid \stringk \mid \bool \mid \unit&&\emph{Sorts}
    \\
    T &:= 
    \ppr!l(S).T \mid
    \ppr?l(S).T \mid 
    T + T \mid
    \End\mid
    \mu X.T \mid
    X 
    &&\emph{Types}
    \\
    P &:= 
    \POUT \ppr l e P \mid 
    \PIN \ppr l x P\mid
    P + P\mid
    \mu \chi .P\mid
    \chi
    \mid
    \ifk\; e \; \thenk\; P\; \elsek\; Q\mid \INACT
    &&\emph{Processes}
    \\
    {\cal M} &:= \ppp\lhd P \mid\hspace{.2em} 
    \PARI{\ppp_i\lhd P_i}
    &&\emph{Sessions}
  \end{align*}
  \end{definition}
  
  We require all terms to be contractive, i.e. 
$\mu X_1.\mu X_2. \dots \mu X_n.X_1$ is not allowed as a sub-term for any 
$n\geq 1$~\cite[p.~300]{Pierce02}, which can be alternatively stated as
type variables occur guarded 
(by input or output prefixes)~\cite{DemangeonH11}.\footnote{
Formally, contractiveness is mechanised in Coq~\cite{coq} (cf.~App.~\ref{sec:coq}) 
by relying on 
the reflexive-transitive closure of the transition system of types 
introduced \mbox{in~\S~\ref{sec:alg-lts}.}
}

We use $\ppp, \ppq, \ppr$ to range over  
\emph{participants},
$l$ to range over \emph{labels}, and $i,j$ to range over indexes (natural numbers).
$X, Y$ range over \emph{type variables}, 
$e, e'$ range over \emph{expressions}, $v, w$ range over \emph{values},
$x,y$ range over \emph{variables}, and $\chi$ range over \emph{process variables}.
\emph{Sessions} $\cal M$ belong to the set $\mathbb M$.
A single session or \emph{thread} is a process $P$ indexed by a
participant, denoted $\ppp\rhd P$. A \emph{multiparty session} is a composition of all threads, 
denoted $\PARI\ppp_i\lhd P_i$ or
$\ppp_1\lhd P_1 \PAR \cdots \PAR \ppp_n\lhd P_n$.

The constructor $\mu$ is a \emph{binder} in types and processes, respectively: 
we let $X$ be bound in $\mu X.T$ 
and \emph{free} in $T$; similarly, $\chi$ is bound in $\mu \chi.P$ and free in $P$.
The remaining binder for processes is input: variable $x$ is bound in 
$\PIN \ppr l x P$ and free in~$P$.
\emph{Closed} terms are those without free variables.

We assume the \emph{substitution} of free occurrences of a type variable $X$ in a type $T_1$ with a closed 
type  $T_2$, written $T_1\{T_2/X\}$.
We assume the substitution of free occurrences of a process variable $\chi$ in process $P_1$ 
with a closed process $P_2$, written $P_1\{P_2/\chi\}$, and the substitution of free occurrences of 
variable $x$ in process $P$ with a value $v$, written $P\{v/x\}$.
A type $R$ is $\mu$-\emph{guarded} (\emph{guarded}, for short)  if it is a sub-term of $T$ 
in the definition $\mu X.T$. 

The symbol $=$ is reserved for Leibniz equality.
\begin{definition}[Session notation]\label{def:notation}\\
$\begin{array}{rll}
\oplus_{i\in I} \ppr!l_i(S_i).T_i  &\eqdef 
    \ppr!l_1(S_1).T_1 + \cdots + \ppr!l_n(S_n).T_n &I = (1, \dots,n), n \geq 1
    \\
    \&_{i\in I} \ppr?l_i(S_i).T_i &\eqdef 
    \ppr?l_1(S_1).T_1 + \cdots + \ppr?l_n(S_n).T_n &I = (1, \dots,n), n \geq 1
  \end{array}$
  \end{definition}  
The next step towards the definition of the typing system
is to identify \textit{well-formed} types that correctly abstract multiparty sessions.
The definition is in ~App.~\ref{sec:wf}.
We collect the labels of types in multi-sets, and the polarities and the participants of types in sets.
Intuitively, a sum type $T_1 + T_2$ is \emph{well-behaved} when it has not duplicated labels,
$T_1$ and $T_2$ have the same unique polarity, and the same unique participant. 
These assumptions eliminate ill-types of the form e.g. $\ppp!l(S_1).T_1 
+ \ppp!l(S_2).T_2$ or of the form e.g. 
$\ppp_1?l_1(S_1).T_1 + \ppp_2?l_2(S_2).T_2$ with $\ppp_1\ne \pp \ppp_2$, 
as well as mixed choice types, e.g.~$\ppp!l_1(S_1).T_1 + \ppp?l_2(S_2).T_2$.
A type $T$ is well-formed, denoted $\wf(T)$,  when it is 
well-behaved, contractive, and closed.

\myparagraph{Operational semantics of multiparty sessions}\label{sec:lts-process}
We assume an \emph{evaluation} function $\downarrow$ transforming  expressions $e$ into boolean, 
integer and unit values 
$v$, written $e\downarrow v$. 
The operational semantics of multiparty sessions are defined modulo a \emph{structural congruence} relation over 
sessions $\cal M$, denoted $\pequiv\subseteq \mathbb M \times\mathbb M$.
We let $\pequiv$ be the least reflexive relation that satisfies the axiom\\[1mm]
\centerline{
$
\parallel_{i\in I}\ppp_i\lhd P_i   \pequiv \hspace{.2em}
\parallel_{j \in J} \ppp_j\lhd P_j 
\qquad(\keyword{permutation}(I, J)) 
$}\\[1mm]
The \emph{labelled transition rules} are defined in Figure~\ref{fig:session-semantics};
we just present the left rules.
A \emph{computation} is a sequence of $\alpha$-transitions, $\alpha\in\{\tau,\sync l p q\}$,
 or \emph{reductions}  
${\cal M}_1\lts\alpha{\cal M}_2\lts\alpha\cdots$.
We are mainly interested in analysing computations of well-typed sessions
(cf.~\S~\ref{sec:typing}).

\begin{figure}[t]
  \begin{mathpar}
  \inference[\textsc{R-Inp}]
  {
  }{
  \ppp\lhd \PIN \ppq l x P
  \lts{\AIN \ppq l v } 
  \ppp\lhd P\{v/x\}
  }
  \and
  \inference[\textsc{R-Out}]
  {
  e\downarrow v
  }{
  \ppp\lhd \POUT \ppq l e P 
  \lts{\AOUT \ppq l v} 
  \ppp\lhd P
  }
  \and
  \inference[\textsc{R-Sum-L}]{
  \ppr\lhd P\lts \alpha   \ppr\lhd P'
  }{
  \ppr\lhd P + Q \lts \alpha  \ppr\lhd P'
  }
  \and
  \inference[\textsc{R-Com}]{
    \ppp\lhd P\lts{\AIN \ppq l v} \ppp\lhd P'\qquad
    \ppq\lhd Q\lts{\AOUT \ppp l v} \ppq\lhd Q'
  }{
    \ppp\lhd P\PAR \ppq\lhd Q\PARI \ppr_i\lhd R_i
    \lts{\ASYN l \ppp \ppq}
    \ppp\lhd P'\PAR \ppq\lhd Q'\PARI \ppr_i\lhd R_i 
  }
  \and
  \inference[\textsc{R-Rec}]{
  }{
  \ppr\lhd \mu\chi. P\PARI \ppr_i\lhd R_i
  \lts\tau
  \ppr\lhd P\{\mu\chi.P/\chi\}\PARI \ppr_i\lhd R_i
  }
  \and
  \inference[\textsc{R-IfT}]
  {
    e\downarrow \truek 
  }{
  \ppr\lhd \ifk\; e \; \thenk\; P\; \elsek\; Q\PARI \ppr_i\lhd R_i
  \lts\tau
  \ppr\lhd P\PARI \ppr_i\lhd R_i
  }
 \and
 \inference[\textsc{R-IfF}]
 {
   e\downarrow \falsek
 }{
 \ppr\lhd \ifk\; e \; \thenk\; P\; \elsek\; Q\PARI \ppr_i\lhd R_i
 \lts\tau
 \ppr\lhd Q\PARI \ppr_i\lhd R_i
 }
  \and
  \inference[\textsc{R-Str}]
  {
   {\cal M}'_1\pequiv {\cal M}_1 \qquad
   {\cal M}_1 \lts\alpha {\cal M}_2\qquad
   {\cal M}_2 \pequiv {\cal M}'_2
  }{
  {\cal M}'_1
  \lts\alpha
  {\cal M}'_2
  }
  \end{mathpar}  
  \caption{Labelled transition rules for multiparty sessions}
  \label{fig:session-semantics}
  \end{figure}

Rule \textsc{R-Inp} says that a participant \ppp\ waiting for a value from \ppq\ on
the label~$l$ can do a transition labelled by $\AIN \ppq l v$ and instantiate
the formal parameter~$x$ with the value~$v$ in the continuation $P$, noted as $P\{v/x\}$.
Rule \textsc{R-Out} allows a participant \ppp\ sending to \ppq\  
on label~$l$ an expression~$e$ that can be evaluated as $v$  
to do a transition labelled by $\AOUT \ppq l v$ and continue as $P$.
Non-deterministic reductions are allowed by means of rule \textsc{R-Sum-L},
which says that a participant~\ppr\ non-deterministically choosing among process 
$P$ and $Q$, denoted $P + Q$, can do a transition labelled by $\alpha$ and reach 
$\ppr \lhd P'$ whenever $\ppr \lhd P$ can fire the same transition and reach the same redex.

Communication among two participants \ppp\ and  \ppq \
is performed by means of rule \textsc{R-Com}.
Whenever $\ppp\lhd P$ can do a transition labelled by the input action $\AIN \ppq l v$ and reach the
 redex
$\ppp\lhd P'$, and $\ppq\lhd Q$ can do a transition labelled by the output action 
$\AOUT \ppp l v$ and reach the redex $\ppq\lhd Q'$, we can infer a transition labelled with 
$\sync l p q$ 
from
the composition of $\ppp\lhd P$ and $\ppq\lhd Q$ and a session 
$\PARI\ppr_i\lhd R_i$  to  
the composition 
of $\ppp\lhd P'$ and $\ppq\lhd Q'$ and $\PARI\ppr_i\lhd R_i$.
Rule \textsc{R-Rec} allows a participant \ppr\ recursively defined as $\mu \chi.P$  
and running in parallel with a session $\PARI\ppr_i\lhd R_i$, 
to do an internal transition $\tau$ and unfold the body $P$ while instantiating 
the occurrences of $\chi$ in $P$ with $\mu \chi.P$, thus reaching the redex
$\ppr\lhd P\{\mu \chi.P/\chi\}\PARI \ppr_i\lhd R_i$.
Rule \textsc{R-IfT} (\textsc{R-IfF}) says that a participant 
\ppr\ with the body $\ifk\; e \; \thenk\; P\; \elsek\; Q$ and 
running in parallel with a session $\PARI\ppr_i\lhd R_i$, 
can do a $\tau$-transition and reach the redex 
$\ppr\lhd P\PARI \ppr_i\lhd R_i$
($\ppr\lhd Q\PARI \ppr_i\lhd R_i$) whenever the expression $e$ evaluates to \truek (\falsek).
Rule \textsc{R-Str} 
rearranges processes with structural congruence.

\begin{example}\label{ex:auth-lts}
Consider the authorisation protocol in~\S~\ref{ex:auth} and\\[1mm]
$\begin{array}{rl}
  Q_\ppa&\eqdef  \PIN \ppc\pwd x \POUT \pps \auth\falsek \chi +
 \PIN \ppc \ssh x \POUT \pps\auth\truek \chi +
 \PIN \ppc\quit x \INACT
 \\
 P_\pps &\eqdef \mu \chi. (\POUT \ppc\login{} \PIN\ppa\auth x \chi +
        \POUT \ppc\cancel{} \INACT)  
\\ 
P_\ppc &\eqdef \mu \chi.
(\PIN \pps\login x
(\POUT \ppa\pwd {\text{``fido''}} \chi + 
\POUT \ppa\ssh{} \chi)\, +
\PIN \pps\cancel x \POUT \ppa\quit {} \INACT )
\\
{\cal M}&\eqdef \pps\lhd P_\pps\PAR \pp c\lhd P_\ppc\PAR
\ppa\lhd P^*_\ppa
\end{array}$\\[1mm] 
where process $P^*_\ppa \eqdef Q_\ppa \{\mu \chi.Q_\ppa/ \chi\}$ implements the (unfolding of the) 
authorisation server \ppa, 
and processes  $P_\pps$ and $P_\ppc$ implement   the service \pps\ 
and the client~\ppc, respectively.
We analyse transitions of the session introduced in~\S~\ref{ex:auth} and composing the
service  \pps, the client \pp c, and the server \pp a,  
here referred as~$\cal M$.

We want to analyse a communication of the server \pps\ with the client \ppc\ depicting a 
\verb|login| transaction.
A first application of rule \textsc{R-Rec} unfolds the service~\pps:\\[1mm]
\centerline{
$\begin{array}{rl}
{\cal M}\lts\tau
\pps\lhd P^*_\pps 
 \PAR \pp c\lhd P_\ppc\PAR
        \ppa\lhd P^*_\ppa\eqdef {\cal M}_1
\end{array}
$}
where
$P^*_\pps\eqdef 
\POUT \ppc\login{}
\PIN \ppa\auth x P_\pps +
\POUT \ppc\cancel{} \INACT
$.

The next step consists in unfolding the client \pp c.
Since the client thread does not occur in the left, we need to 
first apply \textsc{R-Rec} and then apply structural congruence in rule 
\textsc{R-Str}:
\[
\begin{array}{rl}
\inference[\textsc{R-Str}]{
  \inference[\textsc{R-Rec}]{}{
   \pp c \lhd P_\ppc\PAR \pps\lhd P^*_\pps \PAR
         \ppa\lhd P^*_\ppa\lts\tau
         \pp c\lhd P^*_\ppc\PAR \pps\lhd P^*_\pps \PAR
         \ppa\lhd P^*_\ppa
  }  
}{
{\cal M}_1\lts\tau   
\pps\lhd P^*_\pps\PAR
\ppc\lhd P^*_\ppc\PAR
\ppa\lhd  P^*_\ppa\eqdef {\cal M}_2
}
\end{array}
\]
where 
$P^*_\ppc \eqdef 
\PIN\pps\login x
(\POUT \ppa\pwd{\text{``fido''}} P_\ppc + 
\POUT \ppa\ssh{} P_\ppc)\,  +
\PIN \pps\cancel x 
\POUT\ppa\quit{} \INACT
$. 
Now we apply rule \textsc{R-Com} to infer a communication among the service \pps\ 
and the client \pp c on the label \verb|login|, followed
by \textsc{R-Str}:

\begin{mathpar}
\inference[\textsc{R-Str}]{
  \inference[\textsc{R-Com}]{
    (A) \qquad (B)
  }
  {
    \hspace{-.3em}\ppc\lhd P^*_\ppc\hspace{-.3em}\PAR\hspace{-.3em}  \pps\lhd P^*_\pps
    \hspace{-.3em}\PAR\hspace{-.3em}
    \ppa\lhd  P^*_\ppa
    \hspace{-.3em}
    \lts{\alogin@\ppc \bowtie \pps}
    \hspace{-.3em}
    \pp c\lhd P'_\ppc\hspace{-.3em}\PAR\hspace{-.3em}
    \pps\lhd P'_\pps\hspace{-.3em}\PAR\hspace{-.3em}
    \ppa\lhd  P^*_\ppa\hspace{-.3em}
  }
  }{
  {\cal M}_2\lts{\alogin@\pp c \bowtie \pps}   {\cal M}_3
  }
  \end{mathpar}
where ${\cal M}_3\eqdef \pps\lhd \PIN \ppa\auth x  P_\pps\PAR
\pp c\lhd P'_\ppc\PAR
\ppa\lhd  P^*_\ppa$,
$P'_\pps\eqdef \PIN \ppa\auth x P_\pps$,
$P'_\ppc\eqdef \POUT \ppa\pwd{\text{``fido''}} P_\ppc + 
\POUT \ppa\ssh{} P_\ppc$,
and
\begin{mathpar}
  \inference[(A) \textsc{R-Sum-L}]{
    \inference[\textsc{R-Inp}]{      
    }{\pps?\login(x).P'_\ppc
    \lts{\AIN \pps\alogin {}}
    P'_\ppc
    }
  }{
    \pp c \lhd P^*_\ppc
    \lts{\AIN \pps\alogin{}} 
    \pp c\lhd P'_\ppc
  }
  \and
  \inference[(B) \textsc{R-Sum-L}]{
    \inference[\textsc{R-Out}]{      
      }{\pps \lhd  \POUT \ppc\login{} P'_\pps 
      \lts{\AOUT \ppc\alogin{}}
      \pps\lhd P'_\pps
      }
    }{
      \pps \lhd P^*_\pps
      \lts{\AOUT \ppc\alogin{}} 
      \pps\lhd P'_\pps
    }  
\end{mathpar} 
As you can see, in session ${\cal M}_3$ the client \ppc\ 
is ready to communicate the \emph{password}, or to send a \verb|ssh| request,  
to the authorisation server \ppa.
\qed
\end{example}

\section{Session Environment Reduction, Algorithmically}
\label{sec:alg-lts}
A central notion of multiparty session types is the interaction among parties.
We model this abstraction by depicting the behaviour of 
\emph{session environments}~$\Delta$ assigning types $T$ 
to participants $\ppp$.

Our aim is to define a \emph{function} that decides at \emph{compile-time} when it is safe
to type-check a group of participants running in parallel and 
willing to communicate with each other.
This is reminiscent of the notion of type duality 
in binary session  types (e.g.~\cite{GayH05,GiuntiV16}), 
but encompasses multiple participants. 
We will use the function in the typing system introduced in~\S~\ref{sec:typing}.

\begin{definition}[Labelled transition system]\label{def:lts}
A labelled transition system (LTS) is a tuple 
$(\tilde A, {\cal S}_1,{\cal A}, {\cal S}_2,\rightarrow)$, noted as 
$\tilde A\triangleright\sigma_1\lts\alpha  \sigma_2$, whenever 
$\sigma_1\in{\cal S}_1$ and $\sigma_2\in{\cal S}_2$ and $\alpha\in\cal A$,
where $\tilde A$ is a (possibly empty) tuple of parameters,
${\cal S}_i$ are set of states, $i=1,2$, 
$\cal A$ is a set of actions,
and $\rightarrow$ is a transition relation s.t. 
$\rightarrow\subseteq \tilde A\times {\cal S}_1\times{\cal A}\times{\cal S}_2$.
A~transition relation $\rightarrow$ is a partial function whenever 
$\tilde A\triangleright\sigma_1\lts{\alpha'}  \sigma'_2$ and
$\tilde A\triangleright\sigma_1\lts{\alpha''}  \sigma''_2$ imply
$\alpha'=\alpha''$ and $\sigma'_2=\sigma''_2$.
A LTS is deterministic whenever its transition relation is a partial function.
\end{definition}  

\subsection{Non-deterministic Transition System}\label{sec:lts}
We first define a non-deterministic LTS of session environments,
and then in~\S~\ref{sec:dlts} we outline its transformation to a deterministic LTS. 
Non-deterministic transitions are used in the notion of deadlock 
(cf.~Definition~\ref{def:deadlock}), and in the proof of subject reduction 
(cf.~\S~\ref{sec:sr-sketch}).
In the non-deterministic setting, the parameters $\widetilde A$ are empty and 
${\cal S}_1={\cal S}_2$.

We start by defining a non-deterministic LTS of 
types. 
Since we will also use the transition system to match the actions of processes, it is practical
to use the same labels of the LTS of Figure~\ref{fig:session-semantics}. 
The left rules for types are in Figure~\ref{fig:lts-types}.
The rules are designed for well-formed types (cf.~\S~\ref{sec:syntax}),
as we discuss below (cf.~rules~\textsc{E-Sel-L}, fits\textsc{E-Bra-L}).
Rule \textsc{E-Out} says that a type doing an output to the participant 
\ppr\ on label~$l$ with payload $S$ and continuing as $T$
can fire the action $\AOUT \ppr l v$ and reach  the redex~$T$
whenever $v$ is a value of sort $S$.
Dually, rule \textsc{E-In} allows an input type from 
\ppr\ on label $l$ with payload $S$ and continuing as $T$ 
to do  an action $\AIN \ppr l v$ and reach  the redex $T$, 
if $v$ has sort $S$.
Rule \textsc{E-Sel-L} allows a sum type $T_1 + T_2$ to do an output action
$\AOUT \ppr l v$ and reach the redex $T'$ whenever $T_1$ can fire this action and reach $T'$.
Dually, rule \textsc{E-Bra-L} allows a sum type $T_1 + T_2$ to 
do an input action
$\AIN \ppr l v$ and reach the redex $T'$ if 
$T_1$ can fire this action and reach $T'$.
Note that input and output are the only actions that a sum type can fire.
This is because types as e.g. $T_1 + \mu X.T$ or $T_1 + (\mu X.T + T_2)$  are not well-formed.
\begin{figure}[t]
  \emph{Transition rules for types: \framebox{$T\lts\alpha  T$}}  
\begin{mathpar}
{
\inference[\textsc{E-Out}]{
\emptyset\vdash v\colon S  
}{
\ppr!l(S).T \lts{\AOUT \ppr l v} 
T
}
\
\inference[\textsc{E-In}]{
  \emptyset\vdash v\colon S
}{
  \ppr?l(S).T  \lts{\AIN \ppr l v} T
}
\
\inference[\textsc{E-Sel-L}]{
T_1\lts{\AOUT \ppr l v} T'
}{
T_1 + T_2\lts{\AOUT \ppr l v}
T'
}
}
\and
\inference[\textsc{E-Bra-L}]{
  T_1\lts{\AIN \ppr l v} T'
}{
  T_1 + T_2\lts{\AIN \ppr l v} T'
}
\and
\inference[\textsc{E-Rec}]{
}{
\mu X.T \lts{\tau} T\{\mu X.T/X\}
}
\end{mathpar}
\emph{Transition rules for session environments: \framebox{$\Delta \lts\alpha\Delta$}} 
\begin{mathpar}
\inference[\textsc{Se-Rec}]{
  T\lts\tau T'  
}{
\Delta,\ppp:T\lts\tau
\Delta,\ppp:T'
}
\quad 
\inference[\textsc{Se-Com}]{
  T_\ppp\lts{\AIN \ppq l v} T'_\ppp
  \and
  T_\ppq\lts{\AOUT \ppp l v}  T'_\ppq
}{
\Delta, \ppp: T_\ppp,  \ppq: T_\ppq \lts{\ASYN l \ppp \ppq} 
\Delta, \ppp: T'_\ppp,  \ppq: T'_\ppq
}
\end{mathpar}
\emph{Transition rule for configurations: 
\framebox{$D\diamond \Delta \lts\alpha D\diamond\Delta$}} 
\begin{mathpar}
\inference[\textsc{Se-Top}]{
  \Delta\in D\and 
  \Delta\lts\alpha \Delta'  
}{
D\diamond\Delta\lts\alpha
D\less\Delta\diamond\Delta'
}
\end{mathpar}  
\caption{Labelled transition system of session environments}
\label{fig:lts-types}
\end{figure}  

The non-deterministic transition rules for session environments follow in Figure~\ref{fig:lts-types}, 
and are the counterpart of the non-deterministic rules of the form 
$\Gamma\lts\alpha\Gamma$ 
used in GMST (cf.~\cite{ScalasY19}) to analyse the 
\emph{safety} and \emph{deadlock freedom} of multiparty protocols. 
We consider a top-level rule of the form
$
D\diamond \Delta \lts\alpha D\diamond \Delta
$,
where we refer to $D\diamond \Delta$ as a \emph{configuration}, 
and use $C$ to range over it.
$D$ is a set of type environments representing a \emph{decreasing set}
which is a \emph{subset of a fixed point}: a step can be taken only if $\Delta$ is in 
the decreasing set $D$.
The idea is the following:
since we are interested in computing all possible redexes of session environments,
 we avoid to further analyse the same environment 
twice by removing the visited environments from the (possibly infinite) set of all possible environments.  

Rule \textsc{Se-Top} applies to configurations and checks that an environment $\Delta$  is in the decreasing set $D$,
and $\Delta$ can move to $\Delta'$ with label $\alpha$:
in such case the configuration $D\diamond \Delta$ moves
to the configuration $D\less\Delta\diamond \Delta'$,
where $D\less\Delta$ notes the decreasing set $D$ less the environment~$\Delta$.

Rule \textsc{Se-Rec} applies to session environments and says that  
$\Delta, \ppp   \colon\mu X.T$ can do an internal 
action $\tau$ and reach the environment 
$\Delta, \ppp \colon T\{\mu X.T/X\}$, thus unfolding the type
of the participant \ppp. 
Rule  \textsc{Se-Com} applies to session environments and depicts a communication: 
when a participant \ppp\ has 
a type $T_\ppp$ that can fire an input action $\AIN \ppq l v$ and move to $T'_\ppp$, 
and a participant 
\ppq\ has a type $T_\ppq$ that can fire an output action $\AOUT \ppp l v$ and move to $T'_\ppq$,
then $\Delta,\ppp\colon T_\ppp,\ppq\colon T_\ppq$ can fire a synchronisation action 
$\sync l p q$
and move to   $\Delta,\ppp\colon T'_\ppp,\ppq\colon T'_\ppq$. 
\begin{example}\label{ex:auth-ered}
 Consider the protocol introduced in~\S~\ref{ex:auth} and take
  $\Delta \eqdef \pps \colon T_\pps, \pp c\colon T_\ppc$, 
  $\ppa\colon T^*_\ppa$. 
Consider a fixed point $D$ such that $\Delta\in D$.
A first application of \textsc{E-Rec}, \textsc{Se-Rec}, \textsc{Se-Top} allows for unfolding the type of 
the service \pps, where we let 
  $T^*_\pps\eqdef \ppc!\login(\unit).\ppa?\auth(\bool).T_\pps 
+ \ppc!\cancel(\unit).\End$:\\[1mm]
\centerline{
$
D\diamond \Delta \lts\tau D\less\Delta\diamond 
\Delta, \pps\colon T^*_\pps, \pp c\colon T_\ppc, \ppa\colon T^*_\ppa\eqdef\Delta_1
$
}\\[1mm]
To continue and unfold the type of the client \pp c, we need to verify that $\Delta_1\in D\less\Delta$:
this follows indeed from the property of a fixed point, that is to be closed under transition, and from the 
fact $\Delta_1\ne\Delta$, which holds because types are iso-recursive, and in turn
$T^*_\pps\ne T_\pps$.
We proceed as above and infer the following transition, where
$T^*_\ppc\eqdef
\pps?\login(\unit).
(\ppa!\pwd(\stringk).T_\ppc + \ppa!\ssh(\unit).T_\ppc) +
\pps?\cancel(\unit).\ppa!\quit(\unit).\End$:\\[1mm]
\centerline{
$D\less\Delta\diamond \Delta_1 \lts\tau D\less{\Delta,\Delta_1}\diamond 
 \Delta, \pps\colon T^*_\pps, \pp c\colon T^*_\ppc, \ppa\colon T^*_\ppa\eqdef\Delta_2 
$}\\[1mm]
  Two non-deterministic transitions are available from $\Delta_2$, and involve the 
  synchronisation of \pps\ and \ppc: 
  one over the label \login and the other over the label \cancel.
  The interaction below corresponds to the label \login and is obtained by applying 
  \textsc{E-Out}, \textsc{E-Sel-L},
  \textsc{E-In}, \textsc{E-Bra-L},
  \textsc{Se-Com}, \textsc{Se-Top},
  where
  $T'_\ppc\eqdef
\ppa!\pwd(\stringk).T_\ppc + \ppa!\ssh(\unit).T_\ppc $ and
$D_2\eqdef D\less{\Delta,\Delta_1}$
and $D_3\eqdef  D_2\less{\Delta_2}$:\\[1mm]
\centerline{
$
D_2\diamond \Delta_2 \lts{\alogin@\ppc\bowtie\pps} 
     D_3\diamond \Delta, \pps\colon \ppa?\auth(\bool).T_\pps, \pp c\colon T'_\ppc, 
     \ppa\colon T^*_\ppa\eqdef\Delta_3
$}\\[1mm]
The interaction over the label \cancel is obtained by applying
\textsc{E-Out}, \textsc{E-Sel-R},
  \textsc{E-In}, \textsc{E-Bra-R},
  \textsc{Se-Com}, \textsc{Se-Top},
  where 
  \textsc{E-Sel-R} and \textsc{E-Bra-R} are the right rules of 
  \textsc{E-Sel-L} and \textsc{E-Bra-L}, respectively:\\[1mm]
\centerline{
$
    D_2\diamond \Delta_2 \lts{\acancel@\ppc\bowtie\pps} 
    D_3\diamond    \Delta, \pps\colon \End, \pp c\colon \ppa!\quit(\unit).\End, 
    \ppa\colon T^*_\ppa\eqdef\Delta'_3
    $
    }\\[1mm]
We conclude by noting that the transition system 
$D\diamond \Delta \lts\alpha D\diamond \Delta$
is indeed non-deterministic (Definition~\ref{def:lts})
by 
$\sync \alogin \ppc \pps\ne  \sync \acancel \ppc \pps$
and 
$\Delta_3\ne\Delta'_3$.
\qed
\end{example}

\subsection{Deterministic Session Environment Transitions}\label{sec:dlts}
In this section, we define a deterministic LTS for environments that is the basis for the definition 
of \emph{closure} in~\S~\ref{sec:closure}, and in turn for 
the mechanisation of \emph{compliance} (cf.~\S~\ref{sec:compliance}) in deductive tools of the OCaml ecosystem (cf.~\S~\ref{sec:implementation}).

The transition system  $D\diamond \Delta \lts\alpha D'\diamond \Delta'$ is non-deterministic,
for two reasons: 
(1) threads can reduce or interact in any order; 
(2) label synchronisation among two participants can occur on multiple labels and in any order.

To make the LTS deterministic (cf.~Definition~\ref{def:lts}), we need four ingredients:

\noindent
(i) To partition the environment into minimal environments, and 
invoke the LTS on each minimal environment;

\noindent(ii) To collect information about discarded branches and selections in synchronisations;

\noindent(iii) To pass an oracle $\Omega$ that given an 
environment~$\Delta$ returns the next two engaging participants, 
or the next participant firing a $\tau$ action, 
or nothing;

\noindent(iv) To define a scheduling policy for labels of communicating participants.

We discuss (i) and (iii), and provide the signature of the deterministic LTS.
Feature (i) relies on following definition;
see App.~\ref{sec:maximal} for all details.

Let $\keyword{parties}(\ppp?l(S).T)\eqdef\keyword{parties}(T)\cup$ 
$\{\ppp\}=\keyword{parties}(\ppp!l(S).T)$,
$\keyword{parties}(\mu X.T)\eqdef\keyword{parties}(T)$,
$\keyword{parties}(T_1+T_2)\eqdef\keyword{parties}(T_1)$ $\cup\keyword{parties}(T_2)$,
and $\keyword{parties}(T)\eqdef\emptyset$ otherwise.
Define $\keyword{parties}(\emptyset)\eqdef\emptyset$,
 $\keyword{parties}(\Delta,\ppp\colon T)$ 
 $\eqdef\{\ppp\}\cup
 \keyword{parties}(T)\cup 
 \keyword{parties}(\Delta)$.
Let $\Delta\less\eNd$ project all non-ended participants of $\Delta$.
 \begin{definition}[Minimal partition and environments]\label{def:partition}
  A set $\{\Delta_1,\dots,\Delta_n\}\ne\emptyset$ is a partition of $\Delta_1\cup \cdots\cup \Delta_n$
  whenever $\Delta_i\ne\emptyset$
  and 
  $\keyword{parties}(\Delta_i)\cap \keyword{parties}(\Delta_j)=\emptyset$ 
  for all $\{i,j\}\subseteq\{1,\dots, n\}$, $i\ne j$.
  Let ${\cal P}_{\cal R}(\Delta)$ be the set of all partitions of~$\Delta$.
  We say that $\Delta$ is minimal if  
  there not exists ${\cal P}_{\cal R}(\Delta\less\eNd)\ni S\ne\{\Delta\less\eNd\}$ s.t. 
  $\Delta\less\eNd=\bigcup_{\Delta'\in S} \Delta'$.
  A~partition $\{\Delta_1,\dots,\Delta_n\}$ of $\Delta$ is minimal,
  denoted as $\partition\Delta(\Delta_1,\dots, \Delta_n)$, whenever 
  $\Delta_i$ is minimal, for all $i\in\{1,\dots,n\}$. 
 \end{definition}  
The aim of invoking the LTS on minimal environments is to avoid the non-determinism
coming from sub-systems executing unrelated behaviours.
The fixed point mechanism based on decreasing sets assumes that once we re-encounter
the same environment twice, we can stop since we already explored all possible computations.
This is no longer sound if the system contain unrelated sub-systems.
For instance, if an environment contains 
two participants \ppp\ and \ppq\ communicating with each other and reaching a fixed point after few steps, 
and  \emph{also} two participants \ppr\ and~\pps\ communicating with each other,
then, depending on the oracle (see (iii)), it might be the case that
the computation finishes without analysing~\ppr\ and~\pps\ (cf.~\cite{GlabbeekHH21}).
On contrast, if we consider a minimal environment, all parties are properly parsed, 
because the oracle is forced to analyse all sub-processes of the interacting participants. 
As we shall see in~\S~\ref{sec:typing}, the minimality assumption does not pose any limitation
because 
we perform the compliance analysis on all environments of a minimal partition.

Feature (iii) is implemented by adding a \emph{fair} oracle
returning participants 
willing to reduce or communicate when this option is available.
The top level participant of a well-formed type~$T$, denoted ${\tt top}(T)$, 
is a partial function indicating the unguarded participant of a
branching or of a selection:\\
\centerline{
$
{\tt top}(\ppp?(S).T)\eqdef\ppp
\qquad
{\tt top}(\ppp!(S).T)\eqdef\ppp
\qquad 
{\tt top}(T_1 + T_2)\eqdef{\tt top}(T_1)
$
}

\begin{definition}[Oracle fairness]\label{def:fair}
A oracle $\Omega$ is fair whenever:
\begin{enumerate}
  \item $\Omega(\Delta)=(\ppp, \ppq)$ implies ${\tt top}(\Delta(\ppp))=\ppq$ and 
   ${\tt top}(\Delta(\ppq))=\ppp$
   \item $\Omega(\Delta)=\ppp$ implies $\Delta(\ppp)=\mu X.T$
   \item $\Omega(\Delta)$ undefined implies 
   \begin{enumerate}
   \item forall $\ppp\in\dom(\Delta)$ we have $\Delta(\ppp)\ne\mu X.T$
   \item there not exists $\{\ppp, \ppq\}\subseteq\dom(\Delta)$ s.t.
     ${\tt top}(\Delta(\ppp))=\ppq$ and ${\tt top}(\Delta(\ppq))=\ppp$
    \end{enumerate}
\end{enumerate}
\end{definition}

Deterministic transitions of session environments have the following form:
$$
\Omega\triangleright D\diamond\Delta\dlts\alpha D\diamond\Delta\blacktriangleright\Delta
$$
where 
$\Delta$ is minimal (i),
$\Omega$ is a fair oracle (iii),
we assume a label scheduling policy (iv),
$\alpha$ is a synchronisation label $\sync l p q$  or a $\tau$ action decorated 
with the originating participant, denoted $\tau_\ppp$, and
$\Delta$ after the symbol $\blacktriangleright$ is called the \emph{sum continuation} 
and is a type environment or an environment placeholder, denoted~$\nabla^\circ$~(ii).
We note that, w.r.t. to Definition~\ref{def:lts}, we have that 
$\widetilde A=\Omega$, the set of states ${\cal S}_1$ contains 
$ D\diamond\Delta$, and the set of states ${\cal S}_2$ contains 
$D\diamond\Delta\blacktriangleright\Delta$.
Moreover, $\dlts{}$ is a partial function:
$\Omega\triangleright D\diamond\Delta\dlts{\alpha'} D'\diamond\Delta'_1\blacktriangleright\Delta'_2$
and 
$\Omega\triangleright D\diamond\Delta\dlts{\alpha''} D''\diamond\Delta''_1\blacktriangleright\Delta''_2$
imply $\alpha'=\alpha''$, 
and $D'= D''$, $\Delta'_i=\Delta''_i$, $i=1,2$.

\begin{example}\label{ex:auth-eredD}
Consider $D\diamond\Delta$ defined in Example~\ref{ex:auth-ered}.
We note that $\Delta$ is minimal.
Take a fair  oracle $\Omega$, and assume that  the scheduling of labels 
follows the \emph{lexicographic order}.
First, we note that $\Omega(\Delta)$ undefined gives rise to a contradiction, because e.g.
$\Delta(\pps)=\mu X.T$.
Depending on the oracle $\Omega$, we may have
$\Omega(\Delta)=\pps$ or $\Omega(\Delta)=\ppc$,
because any other combination would contradict  Definition~\ref{def:fair}.  

Assume $\Omega(\Delta)=\pps$. 
A first step let us infer the reduction of the service,
where $\Delta_1\eqdef \Delta, \pps\colon T^*_\pps, \ppc\colon
T_\ppc, \ppa\colon T^*_\ppa$, and $\keyword{minimal}(\Delta_1)$.\\[1mm]
\centerline{
$\Omega\triangleright D\diamond \Delta \dlts{\tau_\pps} D\less\Delta\diamond 
   \Delta_1\blacktriangleright\nabla^\circ
$}\\[1mm]  
Next, we assume that $\Omega(\Delta_1)=\ppc$,
where $\Delta_2\eqdef \Delta, \pps\colon T^*_\pps, \ppc\colon
T^*_\ppc, \ppa\colon T^*_\ppa$.\\[1mm]
\centerline{
$
\Omega\triangleright D\less\Delta\diamond \Delta_1 \dlts{\tau_\ppc} D\less{\Delta,\Delta_1}
  \diamond \Delta_2\blacktriangleright\nabla^\circ
  $}\\[1mm]  
  In the next round we have $\keyword{minimal}(\Delta_2)$ and  $\Omega(\Delta_2)=(\ppc,\pps)$,
  and the algorithm picks the first label in the 
  intersection of the labels of \ppc\ and \pps, 
  that is \cancel:\\[1mm]
\centerline{
$
\Omega\triangleright D_2\diamond \Delta_2
  \dlts{\acancel@\ppc\bowtie\pps}
  D_3 \diamond
  \Delta''\blacktriangleright\Delta'
$}\\[1mm]    
  where  $D_2, D_3$ are defined in Example~\ref{ex:auth-ered} and%
\\[1mm]
\centerline{
$\begin{array}{rll}
   \Delta''&\eqdef
  \pps\colon \End, \ppc\colon \ppa!\quit(\unit).\End,
  \ppa\colon T^*_\ppa
  \\
  \Delta'&\eqdef
  \pps\colon \ppc!\login(\unit).\ppa?\auth(\bool).T_\pps,
  \ppc\colon  
\pps?\login(\unit).T'_\ppc, \ppa\colon T^*_\ppa
\\
T'_\ppc &\eqdef \ppa!\pwd(\stringk).T_\ppc + \ppa!\ssh(\unit).T_\ppc 
\end{array}
$
}\\[1mm]
After this sequence of transitions, we have two 
minimal environments~$\Delta''$ and~$\Delta'$ corresponding to the redex of the interaction 
of the service \pps \ and the client \ppc\ over the label \cancel,
and to the environment prompt to let  \pps \ and  \ppc\
interact over the label \login, respectively.
The idea is to deterministically visit all the binary trees spawned by further transitions 
starting from $D_3\diamond \Delta''$ and 
from $D_3\diamond \Delta'$, respectively, as we discuss in the next section.
  \qed
  \end{example}  

\subsection{Closure}\label{sec:closure}
The aim of the deterministic LTS presented in~\S~\ref{sec:dlts} is to be used
by the function that computes the \emph{compliance} of session environments in the typing system 
(cf.~\S~\ref{sec:typing}).
Compliance analyses all final environments computed by the closure of the deterministic 
transitions originating from a type environment. 

More specifically, we consider the \emph{semireflexive-transitive closure} of the deterministic 
lts $\dlts{}$, denoted $\Lts{}$.
Semireflexivity means that a configuration is related with itself
only if is stuck, that is it cannot fire any transition.

We are interested in applying closure to environments preserving minimality.

\begin{definition}[Stuck environment]\label{def:stuck}
A minimal environment $\Delta$ is stuck w.r.t. an oracle $\Omega$ and a decreasing set 
$D$, denoted ${\tt stuck}_{\Omega, D}(\Delta)$, if there not exists $\alpha, \Delta_1$,
$\Delta_2$ such that
$\Omega\tr D\diamond \Delta\dlts\alpha D'\diamond \Delta_1\blacktriangleright \Delta_2$.
\end{definition}  

\begin{definition}[Closure]\label{def:closure}
Define:\\
\centerline{
\small
$
\begin{array}{c}
  \inference[\textsc{C-Rfl}]{
    {\tt stuck}_{\Omega, D}(\Delta)
  }{
  \Omega\tr D\diamond \Delta \Lts{} D\diamond \Delta  
  }
  \quad
  \inference[\textsc{C-Err}]{
    \neg{\tt minimal}(\Delta)
  }{
  \Omega\tr D\diamond \Delta \Lts{} {\tt err}  
  }
\\[2mm]  
\inference[\textsc{C-Tra}]{
    \Omega\tr D\diamond \Delta\dlts\alpha D'\diamond \Delta_1\blacktriangleright\Delta_2
\quad
    \Omega\tr D'\diamond \Delta_1\Lts{} \widetilde{E_1} 
\quad 
    \Omega\tr D'\diamond \Delta_2\Lts{} \widetilde{E_2} 
  }
  {
  \Omega \tr D\diamond \Delta\Lts{}\widetilde {E_1},\widetilde {E_2}
  }
\end{array}
$}\\[2mm]
The closure of a minimal environment $\Delta$ w.r.t. a decreasing set $D$ s.t. 
  $\Delta\in D$ and a fair oracle $\Omega$ is defined by the following rule:
  \begin{mathpar}
  \inference[\textsc{C-Top}]{
    \Omega \tr D\diamond \Delta\Lts{} D_1\diamond \Delta_1, \cdots, D_n\diamond \Delta_n
  }{
  {\tt closure}_{\Omega, D}(\Delta)=
  \Delta_1,\dots,\Delta_n 
  }
  \end{mathpar}

  \end{definition} 

Given a a fair oracle $\Omega$,  the relation $\Lts{}$ associates 
a configuration $C$ to a \emph{non-empty} tuple of \emph{e-configurations} $E_1,\dots, E_n$, denoted as 
$\widetilde E$, where each $E_i$ is a configuration $C$ or the \emph{failure} {\tt err}.
Given a configuration $C = D\diamond \Delta$, three cases may arise.
If $C$ is stuck, that is $C$ cannot fire any transition, then
we apply rule [\textsc{C-Rfl}] and relate $C$ with itself, else
if $C$ is not minimal, then we we apply rule [\textsc{C-Err}] and relate $C$ with {\tt err}.
Otherwise we have that $C$ fires an action and reaches the 
redex $\Delta'\diamond \Delta_1\blacktriangleright\Delta_2$:
we apply rule [\textsc{C-Tra}] and 
whenever $\Delta'\diamond \Delta_1$ is related by $\Lts{}$ to 
the e-configurations 
$\widetilde {E_1}$,
and $\Delta'\diamond \Delta_2$ is related by $\Lts{}$
to 
$\widetilde {E_2}$,
we let $C$ be related by $\Lts{}$ to 
$\widetilde {E_1},\widetilde {E_2}$.

The {\tt closure} of a session environment $\Delta$ is defined iff $\Lts$ does not relate~$\Delta$
with failures. If this is the case,  then $\Lts$ relates $\Delta$ with configurations $\widetilde C$:
the function strips off all decreasing sets and associates $\Delta$  to a set of  minimal stuck environments.
It is worth noting that {\tt closure} is a \emph{terminating function},
because it is deterministic and   
it has  $|D|$ as decreasing measure. 

\begin{example}\label{ex:auth-closure}
We continue the analysis started in Example~\ref{ex:auth-eredD} and find a subset of 
${\tt closure}_{\Omega,D}(\Delta)$, which is defined because~$\Delta$ and its redexes are minimal.
Remember $T^*_\ppa$ defined in~\S~\ref{ex:auth}:
  $T^*_\ppa\eqdef 
\ppc?\pwd(\stringk).\pps!\auth(\bool).T_\ppa +
\ppc?\ssh(\unit).\pps!\auth$ $(\bool).T_\ppa + 
\ppc?\quit(\unit).\End$. 
Consider $D_3\diamond\Delta''\blacktriangleright\Delta'$ 
defined in Example~\ref{ex:auth-eredD}:\\[1mm]
\centerline{
$
\Omega\triangleright D\diamond\Delta
\dlts{\tau_\pps}
\dlts{\tau_\ppc}
\dlts{\acancel@\ppc\bowtie\pps}
D_3\diamond
  \Delta''\blacktriangleright\Delta'
$
}\\[1mm]
To calculate the closure of $\Delta$ w.r.t. $D$, we need to analyse the closures of 
$\Delta''$ and $\Delta'$ w.r.t. $D_3$, respectively.
We have that $\Delta''$ and $\Delta'$ are minimal: we analyse the former closure, 
and note that $D_3\diamond \Delta''$ is not stuck, 
i.e. the client \ppc\ and the server \ppa\ can communicate on \quit.
Assume $\Omega(\Delta'')=(\ppa,\ppc)$. We have:
\begin{mathpar}
  \inference[\textsc{C-Tra}]{
    \Omega\triangleright D_3\diamond \Delta''\dlts{\aquit@\ppa \bowtie\ppc}
    D_3\less{\Delta''}\diamond\pps\colon \End, \ppc\colon\End,
    \ppa\colon \End  
    \blacktriangleright\nabla^\circ
    \quad 
    (A)
  }{
\Omega\triangleright D_3\diamond \Delta''\Lts{}D_3\less{\Delta''}\diamond 
\pps\colon \End, \ppc\colon\End,
  \ppa\colon \End
} 
\end{mathpar}
\centerline{
$
(A)\ \inference[\textsc{C-Rfl}]{}{
    \Omega\triangleright D_3\less{\Delta''}\diamond\pps\colon \End, \ppc\colon\End,
    \ppa\colon \End  \Lts{} 
    D_3\less{\Delta''}\diamond\pps\colon \End, \ppc\colon\End,
    \ppa\colon \End  
    }
$
}\\[1mm]

We can thus infer $(\pps\colon \End, \ppc\colon\End,
\ppa\colon \End)\in{\tt closure}_{\Omega,D}(\Delta)$.
\qed
\end{example}

\section{Iso-Recursive Multiparty Type System}
\label{sec:typing}
\begin{figure}[t]
  \emph{Sorting rules: \framebox{$\Gamma\vdash e\colon S$}}  
  \\
  \emph{Typing rules for processes: \framebox{$\Gamma\vdash P\colon T$}}  
  \begin{mathpar}
  \inference[\textsc{T-End}]{}{
  \Gamma\vdash 0\colon\End 
  }
  \and
  \inference[\textsc{T-Rec}]{
  \Gamma, \chi\colon \mu X.T \vdash P\colon T\{\mu X.T/X\}
  }{
  \Gamma\vdash \mu\chi.P\colon \mu X.T
  }    
  \and
  \inference[\textsc{T-Var}]{
  \Gamma(\chi)=\mu X.T
  }{
  \Gamma\vdash \chi\colon \mu X.T
  }
  \and
  \inference[\textsc{T-Inp}]{
  \Gamma,x\colon S\vdash P\colon T    
  }{
  \Gamma\vdash   \PIN \ppr l x P \colon
    \ppr?l(S).T
  }
  \and
  \inference[\textsc{T-Out}]{
  \Gamma\vdash e\colon S \qquad \Gamma\vdash P\colon T
  }{
  \Gamma \vdash \POUT \ppr l e P \colon 
  \ppr!l(S).T}
  \and
  \inference[\textsc{T-Sum}]{
  \Gamma \vdash  P\colon T_1\qquad
  \Gamma\vdash Q\colon T_2
  }{
  \Gamma \vdash  P + Q\colon T_1 + T_2}
  \and
  \inference[\textsc{T-Sum-L}]{
  \Gamma \vdash  P\colon T_1
  }{
  \Gamma \vdash  P\colon T_1 + T_2
  }
  \and
  \inference[\textsc{T-Sum-R}]{
  \Gamma \vdash  P\colon T_2
  }{
  \Gamma \vdash  P\colon T_1 + T_2
  }
  \and
  \inference[\textsc{T-If}]{
  \Gamma\vdash e\colon \bool\qquad
  \Gamma\vdash P\colon T \qquad \Gamma\vdash Q\colon T
  }{
  \Gamma\vdash  \ifk\; e \; \thenk\; P\; \elsek\; Q\colon T
  }
  \end{mathpar}
  \emph{Typing rules for sessions: \framebox{$\Gamma\Vdash {\cal M}\colon \Delta$}}  
  \begin{mathpar}
  \inference[\textsc{T-Thr}] {
  \Gamma\vdash P \colon T
  \qquad
  \wf(T)  
  }{
  \Gamma\Vdash {\pp p}\lhd P\colon {\pp p}\colon T
  }
  \and 
  \inference[\textsc{T-Ses}]{
  \Gamma\Vdash {\pp p}_1\lhd P_1\colon {\pp p}_1\colon T_1 \quad \cdots \quad 
  \Gamma\Vdash {\pp p}_n\lhd P_n\colon {\pp p}_n\colon T_n \quad
  \Delta = {\pp p}_1\colon T_1,\dots,{\pp p}_n\colon T_n\\
  \\
  \partition\Delta(\Delta_1,\dots,\Delta_k)
  \and
  \forall j\in \{1,\dots, k\}\,.\, 
  \keyword{comp}(\Delta_j)
  }{
  \Gamma \Vdash \PAR_{i\in \{1,..,n\}}\ {\pp p}_i\lhd P_i\colon \Delta
  }
  \end{mathpar} 
  \caption{Type system}\label{fig:typing} 
  \end{figure}
  
The typing rules for processes and sessions are defined in 
Figure~\ref{fig:typing};
the rules for expressions are defined in App.~\ref{sec:sorting_rules}.

Typing judgements for processes have the form $\Gamma\vdash P\colon T$, where $\Gamma$ 
maps variables to sorts and process variables to types:
$$
\Gamma := \emptyset  \mid \Gamma,x\colon S\mid \Gamma, \chi\colon T
$$

Typing judgements for sessions have the form $\Gamma\Vdash{\cal M}\colon \Delta$,
where $\Delta$ is the session environment introduced in~\S~\ref{sec:alg-lts}, 
that is a map from participants to types,  and invoke the type system $\vdash$.
The type system for sessions $\Vdash$ 
 only invokes the type system for processes $\vdash$ with 
 well-formed types (cf.~\S~\ref{sec:syntax}):
 for this reason, the typing rules for processes  
 involving type sums can be simplified
(cf.~rules \textsc{T-Sum},\textsc{T-Sum-L},\textsc{T-Sum-R}).

The rule depicting the essence of iso-recursive multiparty session types is 
\textsc{T-Rec}.
In order to allow $\Gamma$ to type a recursion process $\mu \chi.P$ with a type $\mu X.T$,
it must be the case that $\Gamma,\chi\colon\mu X.T$  types the continuation $P$ with 
the unfolded type $T\{\mu X.T/X\}$.
That is, in our iso-recursive setting the continuation must be typed by explicitly
unfolding the recursive type.
This is different from the equi-recursive approach, e.g.~\cite{GhilezanJPSY19},
where the type of $\mu \chi.P$ and the type of the continuation $P$
can be equal, because types $\mu X.T$ and  $T\{\mu X.T/X\}$ are equal.
For the same reason, in rule \textsc{T-Var}  an environment $\Gamma,\chi\colon\mu X.T$
assigns the type $\mu X.T$ to the process variable $\chi$:
note that it is not possible to assign a non-recursive type to process variables.

Rule \textsc{T-Inp} allows $\Gamma$ to type a input process $\PIN \ppr l x P$ with type 
$\ppr?l(S).T$ whenever $\Gamma,x\colon S$ assigns the type $T$ to the continuation $P$.
Dually, rule \textsc{T-Out} allows $\Gamma$ to type an output process $\POUT \ppr l e P$ with type 
$\ppr?l(S).T$ whenever the expression has sort $S$ and 
$\Gamma$ assigns the type $T$ to the continuation $P$.

Rule \textsc{T-Sum} is used for branching and selection, that are sums containing only
input types from the same participant and without duplicated labels,
or  output types from the same participant and without duplicated labels,
respectively 
(cf. Well-Formed Types in~\S~\ref{sec:syntax}, 
and Definition~\ref{def:notation}).
Note indeed that well-formed types do not contain types of the form e.g.
$T_1 + \mu X.T_2$, or $\End + T$.
The rule says that if $\Gamma$ can be used to type a process 
$P_1$ with type $T_1$, and a process $P_2$ with type $T_2$, then $\Gamma$ types $P_1 + P_2$ 
with type $T_1 + T_2$.

While rule \textsc{T-Sum} types exactly each input and output with their corresponding 
input and output type singletons, rule \textsc{T-Sum-L} 
allows for typing a process $P$ having type $T_1$ with the type  $T_1 + T_2$.
For instance, if 
$P$ is the branching process $\ppr?l_1(x).P_1 +\cdots +\ppr?l_n(x).P_n$ 
then we can use \textsc{T-Sum-L} to 
assign to $P$ the type $\&_{i\in \{1,\dots,n+1\}} \ppr?l_i(S_i).T_i$.
Rule \textsc{T-Sum-R} does the same thing, on the right:
if $P$ has type $T_2$ then we can use the rule to assign to $P$ the type  $T_1 + T_2$.

The increased flexibility offered by rules \textsc{T-Sum-L}, \textsc{T-Sum-R} is used in
the rule for if-then-else, that is \textsc{T-If}.
In order to type process $\ifk\; e \; \thenk\; P\; \elsek\; Q$ with type $T$ 
we require that $e$ has a boolean sort, and that both $P$ and $Q$ have type $T$.
To allow $P$ and $Q$ to use different labels to communicate in input/output with a participant,
we use rules \textsc{T-Sum-L} and  \textsc{T-Sum-R} in the premises of \textsc{T-If},
thus mimicking  a simple form of subtyping. The next example illustrates this idea.
\begin{figure}[t]
\begin{align*}
    P_\ppa &\eqdef  \mu \chi. (P_1 + P_2)
\\     
    P_1 &\eqdef \PIN \ppc\pwd x \textit{Check}_\ppa 
    \\
    P_2 &\eqdef
    \PIN \ppc\ssh x
    \POUT \pps\auth\truek \chi + 
    \PIN \ppc\quit x \INACT
    \\
    \textit{Check}_\ppa &\eqdef
    \ifk\, x= \text{``miau''}\,\thenk\,
    \POUT \pps\auth\truek \chi\,
    \,\elsek\,
    \POUT \pps\fail{} \INACT 
    \\
    T' &\eqdef \ppc?{\tt pwd}(\stringk).(\pps!\auth(\bool).X + 
  \pps!\fail(\unit).\End)\\
  T'' &\eqdef \ppc?\ssh(\unit).\pps!\auth(\bool).X 
  + \ppc?\quit(\unit).\End
\\ 
  T &\eqdef T' + T''  
  \end{align*} 
\caption{Variant  of authorisation server in~\S~\ref{ex:auth}}  
\label{fig:ex-auth-typing}
\end{figure}

\begin{example}\label{ex:auth-typing}
Consider the variant of Figure~\ref{fig:ex-auth-typing}
of the authorisation server \ppa\ in~\S~\ref{ex:auth} such 
that \ppa\ verifies the password sent by the client \pp c
while allowing only one attempt: if the password is wrong, \ppa \ sends 
\fail to the service~\pps \ and stops.
We discuss the typing of the authorisation server~$P_\ppa$, and omit the 
types of the other participants.
A formal derivation is included in App.~\ref{sec:auth-derivation}.

Let $\Gamma\eqdef \chi\colon \mu X.T,x\colon\stringk$, 
consider the two branches of $\textit{Check}_\ppa$, and let 
$T_{\textrm{if}}\eqdef \pps!\auth(\bool).\mu X.T + 
  \pps!\fail(\unit).\End$.
The left branch $\POUT \pps\auth\truek \chi$ can be  assigned to $T_{\textrm{if}}$ 
under $\Gamma$ 
by using \mbox{\textsc{T-Sum-L}}, \textsc{T-Out}, \textsc{T-Var}.
The right branch  $\POUT \pps\fail{} \INACT$ can be assigned to  
$T_{\textrm{if}}$ under $\Gamma$ 
by using \mbox{\textsc{T-Sum-R}}, \textsc{T-Out}, \textsc{T-End}.
By applying \textsc{T-If} we thus assign $T_{\textrm{if}}$ to $\textit{Check}_\ppa$ under $\Gamma$ ;
in turn, process $P_1$ is assigned to $T_1\eqdef\ppc?{\tt pwd}(\stringk).T_{\textrm{if}}$ under 
$\chi\colon \mu X.T$  by using
\textsc{T-Inp}.
We note that $T_1 = T'\{\mu X.T/X\}$.
Process $P_2$ is assigned to $T_2\eqdef \ppc?\ssh(\unit).\pps!\auth(\bool).\mu X.T
+ \ppc?{\tt quit}(\unit).\End$ under $\chi\colon \mu X.T$: we omit all details.
We note that $T_2 = T''\{\mu X.T/X\}$.

We use \textsc{T-Sum} to assign $T'\{\mu X.T/X\}+ T''\{\mu X.T/X\}=T\{\mu X.T/X\}$
to $P_1+P_2$ under $\chi\colon \mu X.T$.
We conclude by using \textsc{T-Rec} to assign $\mu X.T$ to~$P_\ppa$ under the empty environment,
thus typing the authorisation server.
\qed
\end{example}

\myparagraph{Type checking sessions}
The typing rules for sessions of Figure~\ref{fig:typing} have the form 
$\Gamma\Vdash {\cal M}\colon\Delta$ 
and use the rules for processes $\Gamma\vdash P\colon T$. 
The system relies on the notion of \emph{minimal partition} (cf.~Definition~\ref{def:partition}).

Rule \textsc{T-Thr} is used for single threads and  says that if the type system for processes $\vdash$ 
can be used to 
type a process $P$ with a well-formed  type $T$  (cf.~\S~\ref{sec:syntax}),
then the type system~$\Vdash$ assigns the typing $\pp p \colon T$ to  
the thread $\pp p \lhd P$.

Rule \textsc{T-Ses} is the top-level rule used to type-check the multiparty session.
In order to type-check a session composing the threads 
$\pp p_1\lhd P_1,\dots, \pp p_n\lhd P_n$ with 
the session environment $\Delta = \pp p_1\colon T_1,\dots, \pp p_n\colon T_n$, we require two things:
\begin{enumerate}
\item Each thread $\pp p_i\lhd P_i$ is typed with the environment $\pp p_i\colon T_i$, 
for $i=1,\dots n$;
\item Each environment $\Delta_j$ of the minimal partition 
$\{\Delta_1,\dots,\Delta_k\}$ of $\Delta$
satisfies \emph{compliance}, denoted $\keyword{comp}(\Delta_j)$.
\end{enumerate}
Compliance  resembles the approach based on safe 
contexts (e.g.~\cite[Definition~4.1]{ScalasY19}), although is fully computational.

\subsection{Compliance}\label{sec:compliance}

Intuitively, a session typed by a compliant environment never reaches an \emph{error},
that is a deadlocked system,
or a redex containing 
two participants $\pp p$ and $\pp q$ that are willing to communicate, 
e.g. $\pp p$ is sending an output to $\pp q$, and 
$\pp q$ is receiving an input from $\pp p$, or vice-versa,
but they mismatch the communication label and/or the type payload, 
or both $\pp p$ and $\pp q$ are sending (receiving) a value to each other: 
that is, there is a mismatch that makes the two participants stuck.

The formal definition of compliance relies on the \emph{closure} of $\dlts{}$ introduced 
in~\S~\ref{sec:alg-lts},
and of the formal definition of error below.
Let the \emph{tagged labels} of a type $T$, denoted ${\cal L}(T)$,
 be defined inductively as follows: 
 ${\cal L}(\ppr!l(S).T) \eqdef \{l@S\}$,
 ${\cal L}(\ppr?l(S).T)\eqdef \{l@S\}$,
 ${\cal L}(T_1 + T_2) \eqdef {\cal L}(T_1)\cup {\cal L}(T_2)$,
 ${\cal L}(T)\eqdef\emptyset$ otherwise.

\begin{definition}[Well-formed environment]\label{def:wfenv}
A session environment $\Delta$ is well-formed, denoted $\wf(\Delta)$,
whenever $\pp p\in\dom(\Delta)$ implies $\wf(\Delta(\pp p))$.
\end{definition}  

\begin{definition}[Communication mismatch]\label{def:mismatch}
 A well-formed session environment $\Delta$ is a communication mismatch whenever there exists 
$\{\pp p,\pp q\}\subseteq\dom(\Delta)$ such that one of the following cases arise:
\begin{align*}
\Delta(\pp p)&= \oplus_{i\in I} {\pp q}!l_i(S_i).T_i  &
\Delta(\pp q)&= \oplus_{j\in J} {\pp p}!l_j(S_j).T_j
\\
\Delta(\pp p)&= \&_{i\in I} {\pp q}?l_i(S_i).T_i  &
\Delta(\pp q)&= \&_{j\in J} {\pp p}?l_j(S_j).T_j
\\
\Delta(\pp p)&= \oplus_{i\in I} {\pp q}!l_i(S_i).T_i  &
\Delta(\pp q)&= \&_{j\in J} {\pp p}?l_j(S_j).T_j &
&{\cal L}(\Delta(\pp p))\cap {\cal L}(\Delta(\pp q))=\emptyset 
\\
\Delta(\pp p)&= \&_{i\in I} {\pp q}?l_i(S_i).T_i  &
\Delta(\pp q)&= \oplus_{j\in J} {\pp p}!l_j(S_j).T_j &
&{\cal L}(\Delta(\pp p))\cap {\cal L}(\Delta(\pp q))=\emptyset
\end{align*}
\end{definition}
The notion of deadlock is insensitive to decreasing sets and determinism, and is based on the 
non-deterministic transition system 
$\Delta\lts\alpha\Delta$ of Figure~\ref{fig:lts-types}.

\begin{definition}[Deadlock]\label{def:deadlock}
  Let $\keyword{consumed}(\Delta)\eqdef\forall \pp p \in\dom(\Delta)\,.\, \Delta(\pp p)=\End$.
  A session environment $\Delta$ is a deadlock when both
  (1) there not exists $\alpha, \Delta'$ such that 
  $\Delta\lts\alpha\Delta'$, 
  and (2) $\neg\keyword{consumed}(\Delta)$.
\end{definition}

\begin{definition}[Error]\label{def:error}
A well-formed  environment $\Delta$ is an error 
whenever $\Delta$ is a communication mismatch,
or $\Delta$ is a deadlock.
\end{definition}

\begin{definition}[Compliance]\label{def:compliance}
Let $\Delta$ be a minimal well-formed environment.
Define $\keyword{comp}(\Delta)$ whenever for all  fair oracles $\Omega$ and 
fixed points $D$ including $\Delta$,
 if ${\tt closure}_{\Omega,D}(\Delta)=\Delta_1, \dots, \Delta_n$
 then 
 $\Delta_i$ is not an error, for all $i\in \{1,\dots,n\}$.
\end{definition}

 \begin{example}\label{ex:auth-not-comp}
 Consider the minimal well-formed environment $\Delta''$ introduced  at the end of~\S~\ref{ex:auth}, 
  and the claim
 $\neg\keyword{comp}(\Delta'')$, which follows from  $\Delta''$ reaching the deadlocked environment 
 $\Delta_{\tt lock}=
 \pps\colon \ppa?\auth(\bool).T_\pps$,
 $\pp c\colon \ppa!{\tt pwd}(\stringk).T_\ppc + \ppa!\ssh(\unit).T_\ppc$,
 $\ppa\colon \End$.
 We prove the claim by  
using a Lemma mapping non-deterministic transitions to deterministic transitions. 
 We start by a sequence of 
 (non-deterministic) transitions  from $\Delta''$ that lead to  $\Delta_{\tt lock}$, and use  
 the result  to find a fair oracle $\Omega$ mimicking the sequence:
 \begin{align}
 D\diamond \Delta'' 
 &\lts{\tau_{\pps}}\lts{\tau_{\ppc}}
 \lts{\alogin@\ppc\bowtie\pps}
 \lts{\assh@\ppa\bowtie\ppc}
 \lts{\aauth@\pps\bowtie\ppa}
 \label{eq:eanc-1}
 \\
 D_1\diamond \pps\colon T_\pps,\ppc\colon T_\ppc, \ppa\colon T'_\ppa
 &
 \lts{\tau_{\pps}}\lts{\tau_{\ppc}} 
 \lts{\alogin@\ppc\bowtie\pps}
 \lts{\tau_{\ppa}}
 \lts{\assh@\ppa\bowtie\ppc}
 \lts{\aauth@\pps\bowtie\ppa}
 \label{eq:eanc-2}
 \\
 D_2\diamond \pps\colon T_\pps,\ppc\colon T_\ppc, \ppa\colon \End
 &\lts{\tau_{\pps}}\lts{\tau_{\ppc}} 
 \lts{\alogin@\ppc\bowtie\pps}D_3\diamond  \Delta_{\tt lock}
 \label{eq:eanc-3}
 \end{align} 
 The transitions in (\ref{eq:eanc-1}) correspond to a first round of the protocol,
 which leads the service \pps \ and the client \ppc\ to re-initialise, while the
 authorisation server \ppa \ reaches the type 
 $ T'_\ppa = \mu X. (\ppc?{\tt pwd}(\stringk).\pps!\auth(\bool).X +
 \ppc?\ssh(\unit).\pps!\auth(\bool).\End$ $+\ppc?{\tt quit}(\unit).\End)$.
 The transitions in (\ref{eq:eanc-2}) correspond to a second round of the protocol,
 which leads the service \pps \ and the client \ppc\ to re-initialise, while the
 authorisation server \ppa \ reaches the type $\End$.
 The transitions in (\ref{eq:eanc-3}) correspond to the starting of the protocol
 where the service \pps\  sends a \login request to the client~\ppc.
 After that, both the service and the client waits to interact with the server \ppa,
 which has ended. Note that ${\tt stuck}_{\Omega,D_3}(\Delta_{\tt lock})$.

 \noindent
 We apply Lemma~\ref{lem:multistep} in~App.~\ref{sec:multistep} 
 and infer
 $
  \Omega\triangleright D\diamond \Delta''\Lts{}\widetilde{C_1}, 
  D_3\diamond \Delta_{\tt lock},\widetilde{C_2}
 $.
 By Definition~\ref{def:closure},
 we have
 $\Delta_{\tt lock}\in \keyword{closure}_{\Omega,D}(\Delta'')$.
 To prove $\neg\keyword{comp}(\Delta'')$, we show that $\Delta_{\tt lock}$ 
 is an error. In fact, $\Delta_{\tt lock}$ is a deadlock (cf.~Definition~\ref{def:deadlock}), 
 because it cannot fire any action, 
 and 
because  there is a participant that has not finished, e.g. 
 $\Delta_{\tt lock}(\pps)\ne\End$.
 By Definition~\ref{def:error},  $\Delta_{\tt lock}$ is an error.
 \qed
 \end{example}

 \begin{example}\label{ex:auth-comp}
 Consider the minimal well-formed environment $\Delta$ of the authorisation protocol in~\S~\ref{ex:auth}.
 We claim that for any fair oracle $\Omega$ and fixed point $D\ni\Delta$,
 the closure of $\Delta$ returns two environments, where 
 $\Delta^{\tt end}\eqdef\pps\colon\End,\ppc\colon\End,\ppa\colon\End$:
 $\keyword{closure}_{\Omega,D}(\Delta)=\{\Delta,\Delta^{\tt end}\}$.
 Following this claim, we have $\keyword{comp}(\Delta)$.
 In fact, both $\Delta$ and $\Delta^{\tt end}$ are not errors.
 By definition, neither $\Delta$ nor $\Delta^{\tt end}$ is a mismatch:
 the latter case is clear;
 in the former case, the unique unguarded sum of prefixes is the branching 
 of the authorisation service \ppa\ below, while the type of \ppc\ is guarded:
 \begin{center}
 $
 \ppc?{\tt password}(\stringk).\pps!\auth(\bool).T_\ppa +
\ppc?\ssh(\unit).\pps!\auth(\bool).T_\ppa + 
\ppc?{\tt quit}(\unit).\End
$  
\end{center}
Moreover, neither $\Delta$ nor $\Delta^{\tt end}$ is a deadlock.
$\Delta$ can indeed take a step: the environment is in the closure because 
it is first contained in the initial decreasing set $D$ and then re-encountered
after a sequence of interactions.
The claim can be verified by using the certified implementation in~\S~\ref{sec:implementation}. 
 \qed  
 \end{example}  
 
 \begin{remark}
  In~\cite{ScalasY19} an environment is deadlock-free if for all redexes $\Gamma$ 
  reachable in multiple steps 
  we have that if $\Gamma$ does not move then its range contains only the type $\End$. 
  Conversely, Definition~\ref{def:deadlock} expresses a negative property, and in turn 
  we transform the implication $\keyword{stuck}(\Gamma) \rightarrow \keyword{consumed} 
  (\Gamma)$ of~\cite{ScalasY19} 
  into its negation: $\keyword{stuck}(\Gamma) \land \neg \keyword{consumed}(\Gamma)$.
  \qed
\end{remark}

\subsection{Subject Reduction and Progress}\label{sec:sr-sketch}
We conclude this section by showing that the typing system satisfies subject reduction and progress.
We outline the sketch of the proof of subject reduction, and refer to App.~\ref{sec:sr} for all details. 
The proof of progress is in App.~\ref{sec:progress}.

The purpose  of the subject reduction theorem is to establish that
 if a session $\cal M$ is well-typed and does a step $\alpha$ and reaches the session ${\cal M}'$,
then ${\cal M}'$ is well-typed.
Assume that $\Gamma\Vdash{\cal M}\colon\Delta$.
To assess subject reduction, we provide an environment $\Delta'$ s.t. 
$\Gamma\Vdash{\cal M}'\colon\Delta'$.
Since the step ${\cal M}\lts\alpha {\cal M}'$ is non-deterministic, we match this step
with a non-deterministic environment transition (cf.~Figure~\ref{fig:lts-types}).

A key result to establish subject reduction is that compliance (cf.~Definition~\ref{def:compliance})
is preserved by non-deterministic transitions of session environments.

\begin{lemma}\label{lem:lts-alts}
  Let $\Delta$ be minimal.
  If 
  $D\diamond \Delta \lts\alpha D'\diamond \Delta'$ then there exists a fair oracle~$\Omega$
  and environment $\Delta''$ s.t.
  $\Omega\triangleright D\diamond \Delta \dlts\alpha D'\diamond \Delta'\blacktriangleright\Delta''$.
  \end{lemma}
  
\begin{lemma}\label{sec:compliance-dlts}
  If $\keyword{comp}(\Delta)$ and
  $\Omega\triangleright D\diamond \Delta \dlts\alpha D'\diamond \Delta'\blacktriangleright\Delta''$
  then 
  $\keyword{comp}(\Delta')$.
\end{lemma}

\begin{corollary}[Compliance preservation]\label{cor:comp-preserve}
  If $\keyword{comp}(\Delta)$ and
  $D\diamond \Delta \lts\alpha D'\diamond \Delta'$
  then
  $\keyword{comp}(\Delta')$.
\end{corollary}

\begin{lemma}\label{lem:types-struct}
  If $\Gamma\Vdash {\cal M}_1\colon \Delta$ and ${\cal M}_1\pequiv {\cal M}_2$ then
  $\Gamma\Vdash {\cal M}_2\colon \Delta$.
  \end{lemma}
  \begin{proof}
  If follows from  the inversion of  $\Gamma\Vdash {\cal M}_1\colon \Delta$ and 
  the definition of
  $\pequiv$. The result is mechanised in Coq~(cf.~App.~\ref{sec:coq}).
  \qed
  \end{proof}

\begin{theorem}[Subject Reduction]\label{th:sr}
  Let $\cal M$ be a closed session, and let $D$ be a fixed point including $\Delta$.
  Assume (1) $\Gamma\Vdash {\cal M}\colon \Delta$ and 
  (2) ${\cal M}\lts\alpha{\cal M'}$.

\noindent
We have $\Gamma\Vdash {\cal M}'\colon \Delta$ or 
$D\diamond \Delta\lts\alpha D'\diamond \Delta'$ and 
$\Gamma\Vdash {\cal M'}\colon \Delta'$.
\end{theorem}
\begin{proof}
By induction on (2), using value and process substitution (mechanised in Coq, cf.~App~\ref{sec:coq}),
Lemma~\ref{lem:types-struct}, 
and  Corollary~\ref{cor:comp-preserve}.
\qed
\end{proof}

Let $\keyword{Ended}(\PARI{\ppp_i\lhd P_i})$  when for all $i\in I$ we have $P_i =\INACT$.
\begin{theorem}[Progress]\label{th:progress}
  Let $\cal M$ be a closed session.
  If $\Gamma\Vdash{\cal M}\colon\Delta$ and
  does not exist ${\cal M}'$ s.t. ${\cal M}\lts\tau {\cal M}'$
  or 
  ${\cal M}\lts{\ASYN l \ppp \ppq} {\cal M}'$,
  for all $l,\ppp, \ppq$,  
then $\keyword{Ended}(\cal M)$.
\end{theorem} 

\section{Automated Deductive Verification of Compliance}
\label{sec:implementation}
The typing system presented in~\S~\ref{sec:typing} relies on the notion of \emph{compliance}, which
is defined theoretically by relying on the novel definitions of 
\emph{deterministic session environment transitions}
and \emph{closure} introduced in~\S~\ref{sec:alg-lts}. 
In this section, we showcase how these theoretical notions can de deployed \emph{soundly} 
in mainstream programming languages and compilers 
by presenting a \emph{reference implementation of compliance} and by mechanising 
the properties of the implementation, which are that compliant environment are mismatch-free and deadlock-free.

Our goal is to define \emph{compliance} as a computable function 
that decides when a session environment has a ``good behaviour'', 
and in turn can be assigned by the typing system to a session.
We note that computability is an essential pre-requisite for decidable type checking
while assigning non-compliant environments to sessions is unsound because it invalidates 
progress, and must be avoided.

Towards this aim, we need to 
(1) deploy the function;
(2) provide a mechanised proof that the function terminates;
(3) provide a mechanised proof that the function decides freedom from mismatches and deadlocks.
This result is established once (by the type system designer): after that, 
the function can be used each time we invoke the type checker on a session process.

The proofs and their mechanisation in (2) and (3) are necessary because the designer 
can deploy a wrong implementation, e.g. it could have forgotten a case leading 
to an environment deadlock,  thus allowing to type check  sessions that deadlock at runtime.
By providing a computer-assisted proof that the implementation rules out errors and deadlocks in environments,
we can rely on Theorem~\ref{th:progress} to obtain that sessions typed by accepted environments do not deadlock.

In the remainder of the section, we tackle the requirements (1), (2) and (3) by  
defining function \emph{compliance} and its \emph{behavioural specification},
that is the contract of the function~\cite{Meyer92}.
We choose OCaml as target language and use tools of the OCaml ecosystem
relying on Why3~\cite{FilliatreP13} to 
enable automated deductive verification of behavioural specifications by using constraint solvers, 
e.g.~\cite{alt-ergo,MouraB08,cvc4},
while supporting imperative features, ghost code~\cite{FilliatreGP16},
and interactive proofs.

The verification has been done by using Cameleer~\cite{PereiraR20,Pereira24}, 
which in turn relies on~\cite{ChargueraudFLP19,FilliatreP13}.
Proofs of lemmas requiring induction are  done interactively in Why3.

\subsection{Structure of the Implementation}\label{sec:code}
To implement \emph{closure} (cf.~Definition~\ref{def:closure}) in OCaml, 
we use function \verb|cstep| receiving 
a fair oracle~$\Omega$, a decreasing set~$D$, a (ghost) fixed point~$\cal W$,
a session environment~$\Delta$,
and a (ghost) list of environments~$\cal H$ representing the \emph{history}
of the visited environments; the function returns an environment.
Function \verb|compliance| invokes \verb|cstep| in order to accept or reject the environment~$\Delta$:
\begin{ocamlsc}
  let[@ghost] rec cstep (o : oracle) (d : typEnv list) 
    ((w : typEnvRedexes)[@ghost]) (delta: typEnv)
    ((history : typEnv list)[@ghost]):typEnv = $\cdots$
  let[@ghost] compliance (o : oracle) (d : typEnv list)
    ((w : typEnvRedexes)[@ghost]) (delta : typEnv) :bool =
    try let m = cstep o d w delta [] in consumed m
    with | Fixpoint h -> (* h = h0, e *) let e = last h in let h0 = pre h 
    in mem_typEnv e h0 && sound e next | _ -> raise NotCompliant
  \end{ocamlsc}
The behavioural specification of the functions is described in~\S~\ref{sec:verification}.
Ghost parameters are used both to provide a semantics to the fixed point mechanism and to prove the 
soundness of the accepted environments, 
and do not 
have computational interest: all ghost code referring to such parameters should be 
erased from the regular code after providing the proof effort~\cite{FilliatreGP16,PereiraR20}.

In function \verb|cstep| we use exceptions to tackle different behaviours of  environments.
In all cases but for exception \verb|Fixpoint|,
termination by raising an exception determines failure of establishing compliance.

\begin{definition}[Positive exits]\label{def:exits}
  Positive exits of function \verb|cstep| are listed below.
  A positive exit implies that the parameter $\Delta$ of \verb|cstep| satisfies compliance.
  \begin{center}
    \begin{tabular}{|c|c|c|c|c|}
      \hline
    \ \textbf{\keyword{name}}\ &
    \ \textbf{\keyword{param}}\ &
    \ \textbf{\keyword{exit}}\ &   
    \ \textbf{\keyword{exception}}\ &  
    \ \textbf{\keyword{positive}}
    \\
    \hline
    \verb|cstep|
    &$\Omega,D,{\cal W},\Delta,{\cal H} $
    &\checkmark
    &&
    \checkmark
    \\
    \hline
    &
    &&\verb|Fixpoint| 
    &\checkmark
    \\
    \hline
    \end{tabular}
    \end{center}
  \end{definition} 
W.r.t. the signatures of \verb|cstep| and of \keyword{closure} in Definition~\ref{def:closure},
the non-ghost parameters are the same 
while  
the return type is different, because \keyword{closure} returns a set of environments.
Remember that the aim of the returned set of environments is to establish \emph{compliance}
by verifying that all the final environments are not a communication mismatch, 
or a deadlock (cf.~Definition~\ref{def:compliance}).
Function \verb|cstep| achieves the same result
by using exception handling and ghost parameters.

The body of \verb|cstep| is recursive, and  contains sub-calls of the form\\
${\tt cstep}(\Omega, D\less\Delta, {\cal W}, \Delta', ({\cal H},\Delta))$:
the first parameter $\Omega$ is the oracle and is the same in all calls;
the second parameter $D\less\Delta$ corresponds to the removal of $\Delta$ from the decreasing set~$D$; 
the third parameter $\cal W$ is the fixed point and is the same in all calls;
the fourth parameter $\Delta'$ is obtained by updating the type of one or of two participants returned by 
the oracle~$\Omega$ (cf.~Definition~\ref{def:fair});
the last parameter appends $\Delta$ to the history $\cal H$: in the remainder of the section,
the \emph{notation} ${\cal H},\Delta$ indicates that  $\Delta$ is the last environment visited in
${\cal H}\cup\{\Delta\}$.

\subsection{Verification}\label{sec:verification}
The verification of function \verb|compliance|
relies on the behavioural specification and verification 
of function \verb|cstep|, which in turn relies on auxiliary lemmas.
\begin{figure}[t]
\begin{tabular}{|c|c|c|c|l|l|l|}
  \hline
\ \textbf{\keyword{name}}\ &
\ \textbf{\keyword{param}}\ & 
\ \textbf{\keyword{result}}\ &  
\ \textbf{\keyword{variant}}\ &
\ \textbf{\keyword{requires}}\ &
\ \textbf{\keyword{raises}} \ &
\ \textbf{\keyword{ensures}} \
\\
\hline
\verb|cstep|
&$\Omega$ 
&$\Delta_o$
&$|D|$ 
&
$\keyword{fair}(\Omega)$ 
&
\verb|OracleNotFair|$ \Rightarrow\falsek$
&
$\keyword{cons}(\Delta_o)$
\\
\hline
&$D$&
&&
$\keyword{isFix}({\cal W},\Delta,|D|*2)$
&
\verb|Fixpoint|$({\cal H}',\Delta') \Rightarrow$ &
\\
&&&&& $\Delta'\in {\cal H}'\land \keyword{sound}_\Omega(\Delta')$ &
\\
\hline
&${\cal W}$
&&&
$D\cap {\cal H}=\emptyset$
&
\verb|Deadlock|$({\cal H}',\Delta') \Rightarrow$ &
\\
&&&&& $\Delta'\in {\cal H}'\land \keyword{mismatch}(\Delta')\vee$ &
\\
&&&&& $\Omega(\Delta')=\keyword{Ret}_0 \land \neg\keyword{cons}(\Delta')$
&
\\ 
\hline
&$\Delta$
&&&
$D\cup {\cal H}=\keyword{comb}(\cal W)$
&
\verb|WrongBranch|$({\cal H}',\Delta') \Rightarrow$ &
\\
&&&&&
$\exists \ppp,\ppq\,.\,\Omega(\Delta')=\keyword{Ret}_2(\ppp,\pp q)
\land$
&
\\
&&&&&
$\keyword{mismatch}_2(\Delta'(\ppp),\Delta'(\ppq))$
&
\\
\hline
&$\cal H$
&&&
$\Delta\in\keyword{comb}(\cal W)$
&
\verb|DecrNotFix|$({\cal H}',\Delta') \Rightarrow$ 
&
\\
&&&&&
$\Delta'\not\in {\cal H}'$
&
\\
\hline
&&&&&
\verb|NotMinimal|$({\cal H}',\Delta') \Rightarrow$ 
&
\\
&&&&&
$\neg \keyword{minimal}(\Delta')$
&
\\
\hline
\end{tabular}
\\[1em]
\begin{mathpar}
\small
{
\inference[\textsc{Sd-Rec}]{
\Omega(\Delta)=\keyword{Ret}_1(\pp p)\\ \neg\keyword{mismatch}(\Delta)}{
\keyword{sound}_\Omega(\Delta)  
}
\and
\inference[\textsc{Sd-Com}]{
  \Omega(\Delta)=\keyword{Ret}_2(\pp p, \pp q)\\
  \neg\keyword{mismatch}(\Delta)
  }{
\keyword{sound}_\Omega(\Delta)  
}}
\end{mathpar}
\begin{tabular}{|c|c|c|l|l|l|}
  \hline
\ \textbf{\keyword{name}}\ &
\ \textbf{\keyword{param}}\ & 
\ \textbf{\keyword{result}}\ &  
\ \textbf{\keyword{requires}}\ &
\ \textbf{\keyword{raises}} \ &
\ \textbf{\keyword{ensures}} \
\\
\hline
\verb|compliance|
&
$\Omega$ 
&$b$
&$\keyword{fair}(\Omega)$ 
&\verb|NotCompliant|$ \Rightarrow\truek$
&$b=\truek$
\\
\hline
&$D$
&&$\keyword{isFix}({\cal W},\Delta,|D|*2)$
&
&
\\
\hline
&$\cal W$
&&  $D=\keyword{comb}(\cal W)$&& 
\\
\hline
&$\Delta$
&&  $\Delta\in D$&&
\\
\hline
\end{tabular}
\caption{Behavioural specification of the implementation}\label{fig:contract}  
\end{figure}

Figure~\ref{fig:contract} presents the behavioural specification of the implementation.
The column \textbf{param} lists the input arguments of each function.
The column \textbf{result} lists the result returned by each function.
The column \textbf{variant} indicates the decreasing argument of \verb|cstep|;
note that \verb|compliance| is not recursive.
The column \textbf{requires} indicates the pre-conditions stated in terms of the 
{parameters}.
The column \textbf{raises} indicates the formula holding for the argument 
carried by  the exception; in the specification of \verb|cstep| we omit exceptions asserting true.
The column \textbf{ensures} indicates the post-condition stated in terms of the 
result.

There are two positive exits of function \verb|cstep| establishing the 
\emph{compliance} of the environment $\Delta$ received in input (cf.~Definition~\ref{def:exits}): 
termination, and   
raising \verb|Fixpoint|.
The conditions holding when \verb|cstep| raises an exception (cf. keyword \textbf{raises}) 
are discussed  below while illustrating the verification process.
Note that exceptions \verb|Fixpoint|, \verb|Deadlock|, \verb|Wrongbranch|, \verb|DecrNotFix| and
\verb|NotMinimal| carry the history ${\cal H}',\Delta'$,
where $\Delta'$ is the last visited environment.

The predicate $\keyword{sound}_\Omega(\Delta')$ relies on the result of 
the oracle and on Definition~\ref{def:mismatch}: if the oracle receives $\Delta'$ and returns one 
(cf.~rule \textsc{Sd-Rec}) or two (cf.~rule~\textsc{Sd-Com}) participants,
and $\Delta'$ is not a mismatch, then $\Delta'$ is sound.
The predicate $\keyword{cons}(\Delta')$ says that all participants in the 
environment~$\Delta'$ have type $\End$.

The post-condition (cf. keyword \textbf{ensures}) of \verb|cstep| says that the returned environment is consumed:
all participants  have type $\End$.
For what concerns the pre-conditions of \verb|cstep| (cf. keyword \textbf{requires}),
the predicate $\keyword{fair}(\Omega)$ implements Definition~\ref{def:fair}
by relying on constructors $\keyword{Ret}_2,\keyword{Ret}_1$ and
$\keyword{Ret}_0$:\\[1mm]
\centerline{
$\begin{array}{rll}
\keyword{fair}(\Omega) \eqdef
\forall \Delta\,.\, (&\forall\pp p\, \pp q\,.\,
\Omega(\Delta)= \keyword{Ret}_2(\pp p, \pp q)\Rightarrow
\keyword{top}(\Delta(\pp p))=\pp q \land 
\keyword{top}(\Delta(\pp q))=\pp p)\ \land
\\
&
(\forall \pp p \,.\,
\Omega(\Delta)= \keyword{Ret}_1(\pp p)\Rightarrow
\exists X\, T\,.\, \Delta(\pp p)= \mu X.T) \ \land 
\\
&
(\forall X\, T \, \pp r\ \pp p\ \pp q\,.\,
\Omega(\Delta)= \keyword{Ret}_0\Rightarrow\Delta(\pp r)\ne \mu X.T \ \land
\\
&\ \neg(\keyword{top}(\Delta(\pp p))=\pp q \land 
\keyword{top}(\Delta(\pp q)=\pp p)))
\end{array}
$}\\[1mm]
The predicate $\keyword{isFix}({\cal W},\Delta,n)$ says that 
${\cal W}$ is a fixed point of $\Delta$ (up-to depth $n$).
The core mechanism to analyse iso-recursive types and environments  is to rely on 
fixed points $\cal W$ of type \verb|typEnvRedexes|, that is a map from participants
to all type redexes up-to depth~$n$,
and on the projection of all combinations of these mappings into a set
of environments, denoted $\keyword{comb}(\cal W)$.
The depth~$n$ indicates how many type transitions 
$T\lts{\alpha_1}\cdots \lts{\alpha_n}T'$
 are considered (cf.~Figure~\ref{fig:lts-types});
these include the unfolding of iso-recursive types
$\mu X.T$ into $T\{\mu X.T/X\}$. 
Given a fixed point $\cal W$, we require that the decreasing set $D$ and the 
history~$\cal H$ partition the set $\keyword{comb}(\cal W)$ (cf. Figure~\ref{fig:contract}, 
function \verb|cstep|,
keyword \keyword{requires},
lines 3-4).

The pre-conditions of \verb|compliance| mirror those of \verb|cstep|, modulo the fact that
there is no history. 
The post-condition of \verb|compliance| ensures that the function returns true by exploiting (1) the
post-condition of \verb|cstep| and (2) the formula holding when \verb|cstep| raises \verb|Fixpoint|.
The exceptional exit of \verb|compliance| occurs when raising 
\verb|NonCompliant|, thus rejecting the input environment~$\Delta$.

\myparagraph{Termination}
The first result establishes that function \verb|cstep| terminates.
We instruct~\cite{PereiraR20} to use $|D|$ as decreasing measure, cf. the keyword 
\keyword{variant} in the function specification of Figure~\ref{fig:contract},
and obtain the desired result automatically. 

\myparagraph{Absence of communication mismatches}
In order  to show that environments accepted by \verb|compliance| are mismatch-free,
we ensure that positive exits of function \verb|cstep| (cf.~Definition~\ref{def:exits}) carry environments that
are not communication mismatches (cf.~Definition~\ref{def:mismatch}) by inspecting \verb|cstep|'s contract  in
Figure~\ref{fig:contract}.

The first positive exit is termination: \verb|cstep| returns $\Delta_o$.
The contract's clause with keyword \keyword{ensures} establishes 
that $\Delta_o$ is consumed:
by definition, $\Delta_o$ is not a mismatch.
The second positive exit corresponds to the exception \verb|Fixpoint|: 
the exceptions carries the history ${\cal H'},\Delta'$, where $\Delta'$ is the last visited environment.
The clause \keyword{raises} establishes that $\Delta'\in \cal H'$
and that $\Delta'$ is sound.
By inversion of rules \textsc{Sd-Rec}, \textsc{Sd-Com}, we obtain that
$\Delta'$ is not a mismatch.\qed

The automated verification is performed in~\cite{PereiraR20} and relies on 
the predicate $\keyword{mismatch}_2(T_1, T_2)$ (cf.~Figure~\ref{fig:contract}) 
to deal with wrong choices of sums: intuitively, the predicate follow
Definition~\ref{def:mismatch} by using types rather than participants.

\myparagraph{Absence of deadlocks}
Similarly, we show that positive exits of \verb|cstep| of Definition~\ref{def:exits} 
correspond to absence of deadlocks of Definition~\ref{def:deadlock}.

The first positive exit occurs when function \verb|cstep| returns $\Delta_o$.
The clause with keyword \keyword{ensures} establishes 
that $\Delta_o$ is consumed:
by definition, $\Delta_o$ is not a deadlock.
The second positive exit is raising exception \verb|Fixpoint|; 
the exception  carries the history~${\cal H}',\Delta'$, where $\Delta'$ is the last 
visited environment.
From the contract's clause with keyword \keyword{raises}, we infer that 
$\Delta'\in\cal H'$ and that $\Delta'$ is sound.

By inversion of rules \textsc{Sd-Rec}, \textsc{Sd-Com}, we obtain that
two cases arise: 
(1) there is a participant \pp p s.t. $\Omega(\Delta')=\keyword{Ret}_1(\pp p)$ and 
$\Delta'$ is not a mismatch; 
(2) there are participants \pp p, \pp q such that
$\Omega(\Delta')=\keyword{Ret}_2(\pp p, \pp q)$ and $\Delta'$ is not a mismatch. 
We show that in both cases (1) and (2) we have that $\Delta'$ can do a transition.
\begin{enumerate}
  \item By the fairness pre-condition of \verb|cstep|, we obtain\\[1mm] 
\centerline{
  $ 
  \Omega(\Delta')= \keyword{Ret}_1(\pp p)\Rightarrow
  \exists X\, T\,.\, \Delta'(\pp p)= \mu X.T
  $
}\\[1mm]
We apply \textsc{E-Rec}, \textsc{Se-Rec} of Figure~\ref{fig:lts-types} and find $\Delta''$ 
  s.t. 
  $\Delta'\lts\tau\Delta''$.
  \item 
  By the fairness pre-condition of \verb|cstep|, we obtain 
\\[1mm] 
\centerline{$ 
\Omega(\Delta')= \keyword{Ret}_2(\pp p, \pp q)\Rightarrow
\keyword{top}(T_{\pp p})=\pp q \land \keyword{top}(T_{\pp q})=\pp p
$}\\[1mm]
where $\Delta'(\pp p)\eqdef T_{\pp p}$ and 
$\Delta'(\pp q)\eqdef T_{\pp q}$.

By inversion of  $\keyword{top}(T_{\pp p})$ and  $\keyword{top}(T_{\pp q})$ 
(cf.~Definition~\ref{def:fair}), 
we obtain that 

$T_{\pp p}=\&_{i\in I} \pp q?l_i(S_i).T_i$ or 
$T_{\pp p}=\oplus_{i\in I} \pp q!l_i(S_i).T_i$, and 
$T_{\pp q}=\&_{j\in J} \pp p?l_j(S_j).T_j$ or
$T_{\pp q}=\oplus_{j\in J} \pp p!l_j(S_j).T_j$.

By hypothesis, $\Delta'$ is not a mismatch: Definition~\ref{def:mismatch} ensures that
two sub-cases arise:
$T_{\pp p}=\&_{i\in I} \pp q?l_i(S_i).T_i$ and 
$T_{\pp q}=\oplus_{j\in J} \pp p!l_j(S_j).T_j$ and 
${\cal L}(T_{\pp p})\cap {\cal L}(T_{\pp q})\ne \emptyset$,
or
$T_{\pp p}=\oplus_{i\in I} \pp q!l_i(S_i).T_i$, and 
$T_{\pp q}=\&_{j\in J} \pp p?l_j(S_j).T_j$ and 
${\cal L}(T_{\pp p})\cap {\cal L}(T_{\pp q})\ne \emptyset$.
In both cases we apply \textsc{Se-Com} of Figure~\ref{fig:lts-types} and find 
$\alpha, \Delta''$ s.t. $\Delta'\lts\alpha\Delta''$. \qed
\end{enumerate}

\section{Related Work}
\label{sec:discussion}
To the best of our knowledge, only few works follow an iso-recursive approach to session types.
\cite{HeuvelP22} proposes a decentralized analysis of multiparty protocols 
that is based on a typed  asynchronous \mbox{$\pi$-calculus} relying 
on the notion of router processes; deadlock-freedom is established by following
the priority-based approach of session types~\cite{DardhaG18}.
The rule to type check recursion types the continuation by 
 unfolding iso-recursive types and lifting priorities to 
a common greater highest priority.
Finally, type preservation holds up to unfolding 
(cf.~\cite[Theorem 2]{HeuvelP22}).

\cite{HorneP24} studies iso-recursive and equi-recursive subtyping for binary sessions.
Session types are interpreted as propositions of multiplicative/additive linear logic extended with least and
greatest ﬁxed points (cf.~\cite{CairesP10,Wadler14}).
The typing rules correspond to the proof rules in~\cite{BaeldeDKS22}, 
and include the unfolding of least and greatest fixed points. 
The authors compare the two subtyping relations,
 and note that the relations  preserve not only
the usual safety properties, but also termination.

Many recent papers~\cite{ZhouO25,ZhouWO24,ZhouZO23,ZhouZO22,Rossberg23,PatrignaniMD21,LigattiBN17} rely on 
iso-recursive types for variants of the $\lambda$-calculus, 
following the seminal work on Amber rules~\cite{theoryOfObjects96}.
While the setting is different from ours, these papers provide several insights on
the advantage of iso-recursive types and on their algorithmic implementation
 and mechanised verification.
Previous papers~\cite{BengtsonBFGM11,BackesHM14} studied iso-recursive types for a 
concurrent $\lambda$-calculus
that can be seen as the foundational theory of core $F^\sharp$.
 
As mentioned above, iso-recursive types have been first studied formally 
in the setting of Amber  rules~\cite{theoryOfObjects96}.
Pierce's book~\cite{Pierce02} further discusses the differences between iso-recursive and equi-recursive types.

\myparagraph{Future Work} 
Our plans go along two directions: completing the study in the paper and 
extending the language model and the type analysis.

Towards completion, we plan to conclude the mechanisation of subject reduction in Coq,
and to compare the performance of compliance checking in OCaml with the verification 
of deadlock freedom in bottom-up approaches (cf.~\cite{UdomsrirungruangY25}) relying 
on model-checking~\cite{ScalasY19}, eventually considering a realistic testing suite 
involving multiple participants and interactions~(cf.~\cite[Table 2]{ScalasY19}).

For what concerns extensions, there are two main features we are interested in: 
session delegation and asynchronous subtyping for multiparty session types.

Handling session delegation in session types is challenging and 
might require type constructors~\cite{GiuntiV10} or session channel 
decorations~\cite{GayH05,DCDMY09,Vasconcelos12} to preserve type soundness.  
Our plan is to enforce soundness at the type level, without affecting
the programmer's syntax.

Asynchronous subtyping (e.g.~\cite{ampst23}) is known to be undecidable 
for more than two participants.
We envision to overcome this obstacle to an algorithmic solution 
by considering a maximal depth of the search of the asynchronous outputs that 
can be anticipated, similarly to the bound on recursion in~\cite{CutnerYV22}.   

{\small
\myparagraph{Acknowledgements} 
We thank the reviewers for detailed and helpful comments.  
This work is partially supported by EPSRC EP/T006544/2,
EP/N027833/2, EP/T014709/2, EP/Y005244/1, EP/V000462/1,
Horizon EU TaRDIS 101093006, 
Advanced Research and Invention Agency (ARIA), and a grant from the Simons Foundation. 
}

\addcontentsline{toc}{section}{Bibliography}
\bibliography{st}
\bibliographystyle{splncs04}

\appendix

\section{Well-formedness}
\label{sec:wf}

We define well-formed types, noted $\wf(T)$, 
which are the types of our interest.

We require three properties: contractiveness, closure, and well-behaviour.
Contractiveness and closure are defined in~\S~\ref{sec:syntax}.
Well-behaviour 
is defined in \mbox{Figure~\ref{fig:wb}} and 
relies on the auxiliary definitions below.

Let $\multiset{l_1,\dots,l_n}$ denote \emph{multisets} of labels,
and let $\keyword{nodup}(\multiset{l_1,$ $\dots,l_n})$ hold 
whenever $l_i$ is unique, for all $i\in\{1,\dots,n\}$.

Let $\keyword{labels}$ be a function mapping types to multisets of labels,
and assume
$\keyword{labels}(\pp r?l(S).T) \eqdef \multiset l$, 
$\keyword{labels}(\pp r!l(S).T)\eqdef \multiset l$, and
$\keyword{labels}(T_1 + T_2) \eqdef \keyword{labels}(T_1) \uplus \keyword{labels}(T_2)$,
and $\keyword{labels}(T)\eqdef\emptyset$ otherwise.

Let $\keyword{polarity}$ be a partial function mapping types to polarities 
$p\in\{!, ?\}$, and let $\keyword{polarity}(T)\downarrow$ whenever 
$\keyword{polarity}$ is defined on type $T$.
Let
$\keyword{polarity}(\pp r?l(S).T) \eqdef ?$,
$\keyword{polarity}(\pp r!l(S).T)\eqdef !$, and
$\keyword{polarity}(T_1 + T_2) \eqdef p $ 
whenever $\keyword{polarity}(T_1) = p$ and $\keyword{polarity}(T_2)= p$,
and $\keyword{polarity}(T)$ be undefined otherwise.

Let $\keyword{participant}$ be a  partial function mapping types to participants,
and let $\keyword{participant}(T)\downarrow$ whenever $\keyword{participant}$ is defined on type $T$.
Let 
$\keyword{participant}(\pp r?l(S).T)$ $\eqdef \ppr$,
$\keyword{participant}(\pp r!l(S).T)\eqdef \ppr$, and
$\keyword{participant}(T_1 + T_2) \eqdef \ppr $
whenever $\keyword{participant}(T_1)=\ppr$ and 
$\keyword{participant}(T_2)=\ppr$,
and $\keyword{participant}(T)$ be undefined otherwise.

\begin{definition}[Uniform sums]\label{def:uniform}
A sum type $T_1+T_2$ is uniform, noted \\ $U(T_1+T_2)$, if all these conditions hold:
\begin{enumerate} 
\item $\keyword{nodup}(\keyword{labels}(T_1+T_2))$
\item $\keyword{polarity}(T_1 + T_2)\downarrow$
\item $\keyword{participant}(T_1 + T_2)\downarrow$ \ .
\end{enumerate}
\end{definition}  

\begin{figure}[t]
  \begin{mathpar}
    \inference[\textsc{Wb-Inp}]{
    \keyword{wb}(T)
    }{
    \keyword{wb}(\pp r?l(S).T)
    }
    \and
    \inference[\textsc{Wb-Out}]{
    \keyword{wb}(T)
    }{
    \keyword{wb}(\pp r!l(S).T)
    }
    \and
    \inference[\textsc{Wb-Sum}]{
      U(T_1 + T_2) 
      \quad \keyword{wb}(T_1)
      \quad \keyword{wb}(T_2)
      }{
      \keyword{wb}(T_1 + T_2)
      }
      \and
    \inference[\textsc{Wb-End}]{}{\keyword{wb}(\End)} 
    \and
    \inference[\textsc{Wb-Rec}]{
    \keyword{wb}(T)}{\keyword{wb}(\mu X.T)}  
    \and
    \inference[\textsc{Wb-Var}]{}{\keyword{wb}(X)}
  \end{mathpar}  
    \caption{Well-behaved types}\label{fig:wb}  
    \end{figure}  
    
\begin{definition}[Well-behaviour]\label{def:wb}
A type $T$ is well-behaved, noted $\keyword{wb}(T)$, when 
one of the rules of Figure~\ref{fig:wb} applies to $T$.
\end{definition}  

\begin{definition}[Well-formedness]
A type $T$ is well-formed, noted $\wf(T)$,
whenever is contractive, closed, and well-behaved. 
\end{definition}  

\begin{example}\label{ex:wf}
Examples of well-formed types include 
$\ppp ?l(\nat).\ppq !l(\bool).T_1$,  
and \\
$\ppp ?l_1(\stringk).T_1 +
\ppp ?l_2(\intk).T_2 +
\ppp ?l_3(\unit).T_3$,
and $\ppp !l_1(\stringk).T_1 +
\ppp !l_2(\intk).T_2 +$\\
 $
\ppp !l_3(\unit).T_3$,
and $\mu X.(\ppp ?l_1(\stringk).X +
\ppp ?l_2(\intk).T_2)$, 
whenever $l_1,l_2$ and $l_3$ are \mbox{distinct,}
and $T_1,T_2$ and $T_3$ are well-formed.

A first example of ill-formed type is $\ppp ?l(\stringk).T_1 +
\ppp ?l(\intk).T_2$, because the label~$l$ is used twice.
The type $\ppp ?l(\stringk).X$ is ill-formed, because it is open.

The type $T\eqdef\mu X.(\ppp ?l_1(\stringk).X) + \ppp ?l_2(\intk).T_1$ is ill-formed,
because is not uniform.
In fact, $\keyword{polarity}(T)$ is undefined, because 
$\keyword{polarity}(\mu X.(\ppp ?l_1(\stringk).X))$ is undefined. 

To recover well-formedness, whenever $\ppp ?l_2(\intk).T_1$ is well-formed, 
we can lift recursion up:
$\wf(\mu X.(\ppp ?l_1(\stringk).X + \ppp ?l_2(\intk).T_1))$.
Note that there is no risk to capture a free occurrence of~$X$ in $\ppp ?l_2(\intk).T_1$,
because the assumption $\ppp ?l_2(\intk).T_1$ well-formed
implies that $\ppp ?l_2(\intk).T_1$ is closed.

The type $\ppp !l(\intk).(\ppq ?l(\bool).\End + \End)$ is ill-formed,
because the continuation is not uniform:
$\keyword{polarity}(\ppq ?l(\bool).\End + \End)$ is undefined because
$\keyword{polarity}(\End)$ is undefined.
\qed
\end{example}  

Next, we show that well-formedness is preserved by type transitions of 
Figure~\ref{fig:lts-types}.
This property is mandatory to prove subject reduction (cf.~\S~\ref{sec:sr}).

Let $\keyword{sums}(T_1 + T_2)\eqdef\keyword{sums}(T_1)\cup \keyword{sums}(T_2)$,
and $\keyword{sums}(T)\eqdef\{T\}$ otherwise.

\begin{lemma}\label{lem:wb-rec}
If 
$\wb(T)$ then 
$\wb(T\{\mu X.T/X\})$.
\end{lemma}
\begin{proof} 
By structural induction on $T$.
\begin{description}
\item[Case $T=\ppr ?l(S).T_1$.] The I.H. is  $\wb(T_1\{\mu X.T/X\})$.
The result follows from $T\{\mu X.T/X\}=
\ppr ?l(S).(T_1\{\mu X.T/X\})$, and from rule \textsc{Wb-Inp} of Figure~\ref{def:wb}.
\item[Case $T=\ppr !l(S).T_1$.] Analogous.
\item[Case $T = T_1 + T_2$.]
The I.H. is  $\wb(T_1\{\mu X.T/X\})$ and $\wb(T_2\{\mu X.T/X\})$.
To conclude and infer 
$\wb(T\{\mu X.T/X\})=
\wb(T_1\{\mu X.T/X\} + T_2\{\mu X.T/X\})$
we verify $U(T_1\{\mu X.T/X\} + T_2\{\mu X.T/X\})$:
that is items (1), (2) and (3) of Definition~\ref{def:uniform} hold.
By inversion of $\wb(T)$ we obtain $U(T_1+T_2)$.
First, we note that 
$\keyword{polarity}(T')\downarrow$, forall $T'\in \keyword{sums}(T)$.
From this we infer $X\not\in \keyword{sums}(T)$
and in turn 
$\keyword{labels}(T_i\{\mu X.T/X\})=
\keyword{labels}(T_i)$, which trivially proves (1).
To see (2), we note that $\keyword{participant}(T)\downarrow$ and $X\not\in \keyword{sums}(T)$
imply $\keyword{participant}(T_1)=\keyword{participant}(T_2)$ and 
$\keyword{participant}(T_i\{\mu X.T/X\})=
\keyword{participant}(T_i)$, respectively.
In the same manner, from $\keyword{polarity}(T)\downarrow$ and $X\not\in \keyword{sums}(T)$
we obtain $\keyword{participant}(T_1)$ $=\keyword{participant}(T_2)$ and 
$\keyword{participant}(T_i\{\mu X.T/X\})=
\keyword{participant}(T_i)$, respectively, and in turn (3).
\item[Case $T = \End$.] Follows from $T\{\mu X.T/X\}=T$.
\item[Case $T = \mu Y. T_1$.] There are two sub-cases corresponding to 
(1) $X=Y$ and ($X\ne Y$).
In case (1) we have $T\{\mu X.T/X\}=T$ and we have done.
In case (2) the I.H. is $\wb(T_1\{\mu X.T/X\})$.
The result then follows from $T\{\mu X.T/X\}=\mu Y. (T_1\{\mu X.T/X\})$.
\item[Case $T = Y$.] There are two sub-cases corresponding to 
(1) $X=Y$ and (2) $X\ne Y$. In case (1) we have $T\{\mu X.T/X\}=\mu X.T$: we 
apply \textsc{Wb-Rec} to the hypothesis  $\wb(T)$ and conclude that $\wb(\mu X.T)$, as desired.
In case (2) we have $T\{\mu X.T/X\}=T$ and we have done.
\qed
\end{description} 
\end{proof}

\begin{lemma}\label{lem:wf-rec}
  If  
  $\wf(\mu X.T)$ then 
  $\wf(T\{\mu X.T/X\})$.
\end{lemma}
\begin{proof}
It is easy to see that the hypothesis $\wf(\mu X.T)$ implies that 
$T\{\mu X.T/X\}$ is both contractive and closed.
To see that $\wb(T\{\mu X.T/X\})$, we invert $\wb(\mu X.T)$ and obtain 
$\wb(T)$. We apply Lemma~\ref{lem:wb-rec} and infer the desired result.
\qed
\end{proof}

\begin{lemma}\label{lem:wf-preserve}
If $\wf(T)$ and $T\lts\alpha T'$ then $\wf(T')$.
\end{lemma}  
\begin{proof}
By induction on $T\lts\alpha T'$.
Case \textsc{E-Rec} relies on Lemma~\ref{lem:wf-rec}.
\qed
\end{proof}  

\section{Minimal Environments}
\label{sec:minimal}
 In this section, we formalise the projection 
 $\Delta\less\eNd$ used in Definition~\ref{def:partition}, and show 
 examples of both minimal and non-minimal environments.

Let $\emptyset\less\eNd\eqdef \emptyset$,
$\Delta,\ppp\colon \End\eqdef \Delta\less\eNd$, 
and $\Delta,\ppp\colon T\eqdef \Delta\less\eNd,\ppp\colon T$ otherwise.

\medskip\noindent
In order to illustrate minimality,
take the following example:
\begin{align*}
\Delta\eqdef
&\ppp\colon \ppq!l_1(S).\pp t?l_2(S).\End,
\ppq \colon \ppp?l_1(S).\pps !l_2(S).\End,
\\
&\pp r \colon\End,\pp s\colon\End,
\pp t \colon \ppp!l_2(S).\End,
\pps\colon \ppq?l_2(S).\End
\end{align*}
Note that 
\begin{align*}
\Delta\less\eNd = 
&\ppp\colon \ppq!l_1(S).\pp t?l_2(S).\End, 
\ppq \colon \ppp?l_1(S).\pps !l_2(S).\End,\\
&\pp t \colon \ppp!l_2(S).\End,\pps\colon \ppq?l_2(S).\End
\end{align*}
In fact, there is only a partition of $\Delta\less\eNd$, that is
$\{\Delta\less\eNd\}$.
In turn,  ${\cal P}_{\cal R}(\Delta\less\eNd)=\{\{\Delta\less\eNd\}\}$.

Since  
$S\in{\cal P}_{\cal R}(\Delta\less\eNd)$ implies 
$S=\{\Delta\less\eNd\}$, we cannot find the candidate that is 
required for breaking minimality,
hence $\Delta$ is minimal.

\medskip
An example of non-minimal environment is 
$$
\Delta\eqdef
\ppp\colon \ppq!l_1(S).\End,
\ppq\colon \ppp?l_1(S).\End, 
\ppr \colon \pps!l_2(S).\End,
\pps\colon \ppr?l_2(S).\End,
\pp t \colon\End,
\pp u\colon\End
$$
Note that $\Delta\less\eNd= \ppp\colon \ppq!l_1(S).\End,
\ppq\colon \ppp?l_1(S).\End, 
\ppr \colon \pps!l_2(S).\End,
\pps\colon \ppr?l_2(S).\End$.

The set of all partitions of $\Delta\less\eNd$ is 
${\cal P}_{\cal R}(\Delta\less\eNd)\eqdef\{\{\Delta\less\eNd\},\{\Delta_1, \Delta_2\}\}$ 
where
\begin{align*}
\Delta_1 &\eqdef \ppp\colon \ppq!l_1(S).\End, \ppq\colon \ppp?l_1(S).\End
\\
\Delta_2 &\eqdef \ppr \colon \pps!l_2(S).\End, \pps\colon \ppr?l_2(S).\End
\end{align*}

Take ${\cal P}_{\cal R}(\Delta\less\eNd)\ni S=\{\Delta_1, \Delta_2\}$, and note that
$S$ is a candidate for breaking minimality, because 
$\{\Delta_1, \Delta_2\}\ne\{\Delta\less\eNd\}$. 

From $\Delta\less\eNd = \bigcup_{\Delta'\in S}$ we obtain that $\Delta$ is not minimal.

\section{Deterministic Transition System}
\label{sec:alg-lts-rules} 
  \begin{figure}[t]
    \emph{Conditional sum}
    \begin{align*}
      R &:= \circ \mid T &
      \circ\oplus R &= R 
      &
      R \oplus \circ &= R &
      T_1 \oplus T_2 &= T_1 + T_2
    \end{align*}
      \emph{Transition rules for types: \framebox{$T\alts\alpha  T\blacktriangleright R$}}  
    \begin{mathpar}
    \inference[\textsc{E-Out-A}]{
    \emptyset\vdash v\colon S  
    }{
    {\pp r}!l(S).T \alts{\AOUT \ppr l v} 
    T
    \blacktriangleright \circ
    }
    \and
    \inference[\textsc{E-In-A}]{
      \emptyset\vdash v\colon S
    }{
      {\pp r}?l(S).T  \alts{\AIN \ppr l v} T
      \blacktriangleright \circ
    } 
    \and
    \inference[\textsc{E-Sel-A-L}]{
    T_1\alts{\AOUT \ppr l v} T'
    \blacktriangleright R_1
    }{
    T_1 + T_2\alts{\AOUT \ppr l v}
    T'
    \blacktriangleright R_1 \oplus  T_2
    }
    \and
    \inference[\textsc{E-Bra-A-L}]{
      T_1\alts{\AIN \ppr l v} T'
      \blacktriangleright R_1
    }{
      T_1 + T_2\alts{\AIN \ppr l v} T'
      \blacktriangleright R_1 \oplus T_2
    }
    \and 
    \inference[\textsc{E-Rec-A}]{
    }{
    \mu X.T \alts{\tau} T\{\mu X.T/X\}
    \blacktriangleright \circ
    }
    \end{mathpar}
    \emph{Transition rules for session environments: 
    \framebox{$D\diamond \Delta \alts\alpha D\diamond\Delta\blacktriangleright\Delta$}} 
    \begin{mathpar}
      \inference[\textsc{Se-Rec-A}]{
    \Delta\in D\and  
    T\alts\tau T'
    \blacktriangleright \circ
    }{
      D\diamond \Delta,{\pp p}:T\alts{\tau_{\pp p}}
      D\less\Delta\diamond\Delta,{\pp p}:T'\blacktriangleright\nabla^\circ
    }
    \and 
    \inference[\textsc{Se-Com-A}]{
      \Delta\in D\and  
      T_\ppp\alts{\AIN \ppq l v} T'_\ppp
      \blacktriangleright R_\ppp
      \and
      T_\ppq\alts{\AOUT \ppp l v}  T'_\ppq
      \blacktriangleright R_\ppq
      \\
      \Delta_1 = \text{ if } R_\ppp = T''_\ppp \text{ and }  R_\ppq= T''_\ppq \text{ then }
      \Delta, \ppp\colon T''_\ppp, \ppq\colon T''_\ppq \text{ else } \nabla^\circ
    }{
      D\diamond \Delta, \ppp\colon T_\ppp, \ppq\colon T_\ppq \alts{\ASYN l \ppp \ppq} 
      D\less\Delta\diamond\Delta, \ppp\colon T'_\ppp, \ppq\colon T'_\ppq
    \blacktriangleright   \Delta_1
    }
    \end{mathpar}  
    \caption{Semi-algorithmic transition system of session environments}
    \label{fig:alg-lts-types}
    \end{figure}  

  The semi-algorithmic rules for session environments and types are defined in 
  Figure~\ref{fig:alg-lts-types}; we omit right rules for sum types.
  As for the non-deterministic rules in Figure~\ref{fig:lts-types}, we assume that types and environments
  are well-formed (cf.~App.~\ref{sec:wf}).
   The main difference w.r.t. the rules in Figure~\ref{fig:lts-types} is the presence of a 
   \emph{separator}, noted $\blacktriangleright$, to introduce a new parameter, 
   which we call \emph{sum continuation}.

  Intuitively, a sum continuation contains the information that is discarded by the semi-algorithmic
  selection of a sum type, and that is used by the \emph{closure} function specified in Definition~\ref{def:closure}.
   To illustrate an example of sum continuation of types, consider $T_1+ T_2$ defined in 
   Example~\ref{ex:auth-typing} and further analysed in App.~\ref{sec:auth-derivation}.
  We have $T_1+ T_2\alts{\ppc?{\tt pwd}(``\text{1234}")} T_{\textrm{if}}\blacktriangleright T_2$.

  The sum continuation does not necessarily contain useful information.
  In order to discard the information provided in the sum continuation, we introduce two \emph{placeholders} for
  types and environments, noted $\circ$ and $\nabla^\circ$, respectively.
  Formally, type placeholders are embedded in the the semi-algorith{\-}mic transition rules for types, which  
  have the form $T\alts\alpha  T\blacktriangleright R$.
  The term $R$, which we call \emph{remnant}, is  obtained by the union of  $T$ and $\circ$.
  Since we want to discard the placeholder $\circ$ in sums, we define a \emph{sum operator on remnants},
  noted $\oplus$, s.t.  $\circ$ is neutral and $T_1\oplus T_2\eqdef T_1 + T_2$.

  Rules \textsc{E-Out-A},\textsc{E-In-A}, and \textsc{E-Rec-A} for types set  the 
  {placeholder} $\circ$  as sum continuation.
    Rules \textsc{E-Out-A} and \textsc{E-In-A} describe the transitions of single output and input types, 
    respectively, while 
    rule~\textsc{E-Rec-A} describes the silent transitions of recursive types.
    Rules \textsc{E-Sel-A-L} and \textsc{E-Bra-A-L} depict the transitions of selection and branching types, 
    respectively. Their continuations make use of the $\oplus$ operator on remnants defined on top of 
    Figure~\ref{fig:alg-lts-types}.

    To illustrate, consider the type $T_1 + T_2 + T_3$ and assume that we can infer the 
    transition:
    $T_2\alts{\AOUT \ppr l v} \End\blacktriangleright R_2$.
    By applying the (right) rule for selection \textsc{E-Sel-A-R}, we infer
    $T_1 + T_2\alts{\AOUT \ppr l v} \End\blacktriangleright T_1 \oplus R_2$.
    By applying \textsc{E-Sel-A-L}, we obtain 
    $T_1 + T_2 + T_3\alts{\AOUT \ppr l v} \End\blacktriangleright (T_1 \oplus R_2)\oplus T_3\eqdef R$.
    We may have (i) $R_2=\circ$ or (ii) $R_2 = T''$, for some $T''$.
    In case (i) we have $R = T_1 + T_3$ while in case (ii) 
    we have $R= T_1 + T'' + T_3$.

For what concerns the rules for environments, consider rule \textsc{Se-Rec-A}, 
which allows a participant \ppp\ to fire a \emph{decorated silent transition},
noted $\tau_\ppp$.  The sum continuation contains the placeholder $\nabla^\circ$, thus indicating that  
the environment \emph{will be ignored} by the \emph{closure} function specified in Definition~\ref{def:closure}.
Note that in the hypothesis we assume that a type $T$ firing a $\tau$ action has the placeholder~$\circ$ in
the sum continuation.
This holds because the internal action originated from the unfolding of a recursive type, 
which cannot occur in a well-formed sum (cf.~App.~\ref{sec:wf}).

    The remaining rule for environments is \textsc{SE-Com-A}, and describes a synchronisation.
    When the label is $\sync l p q$, that is a synchronisation 
    on the label~$l$ between the participant \ppp\ doing an input and 
    the participant \ppq\ doing an output, the sum  continuation can:
    (1) account for the remnant type of the branching type of $\pp p$,
    and for the remnant type of the selection type of $\pp q$, or
    (2) be the placeholder $\nabla^\circ$ to be ignored.
    Case (2) occurs when at least one of the participants has a singleton type, 
    and in turn one of the remnants is the placeholder~$\circ$:
    since no further synchronisation among the remnants is possible, we set the placeholder 
    $\nabla^\circ$ as sum continuation.
    
    \begin{figure}[t]
      \emph{Deterministic lts for session environments: 
      \framebox{$\Omega\tr D\diamond  \Delta \dlts{\alpha} D\diamond\Delta\blacktriangleright\Delta$}}       
    \begin{mathpar}
      \inference[\textsc{Se-Rec-D}]{
      \keyword{minimal}(\Delta)\and
      \Omega(\Delta)= {\pp p} \and
      D\diamond\Delta\alts{\tau_{\pp p}}D\less\Delta\diamond\Delta_1\blacktriangleright\nabla^\circ  
      }{
          \Omega\tr D\diamond 
          \Delta\dlts{\tau_{\pp p}}
          D\less\Delta\diamond\Delta_1\blacktriangleright\nabla^\circ 
      } 
      \and
      \inference[\textsc{Se-Com-D}]{
        \keyword{minimal}(\Delta)\and
      \Omega(\Delta)= ({\pp p}, {\pp q}) \and
      D\diamond\Delta\alts\alpha D\less\Delta\diamond\Delta_1\blacktriangleright\Delta_2  \\
      \alpha\in\{l@\pp p\bowtie \pp q, l@\pp q\bowtie \pp p \}\and
      l = {\tt first}({\cal L}(\Delta(\pp p))\cap {\cal L}(\Delta(\pp q)))
      }{
          \Omega\tr D\diamond 
          \Delta\dlts\alpha 
          D\less\Delta\diamond\Delta_1\blacktriangleright\Delta_2
      } 
      \end{mathpar}  
      \caption{Deterministic session environment reduction}\label{fig:lts-types-det}
      \end{figure}
    
    \medskip
    The deterministic transition system is defined in 
    Figure~\ref{fig:lts-types-det}
    on top of the semi-algorithmic system defined in~Figure~\ref{fig:alg-lts-types}. 
    The system contains two rules, \mbox{\textsc{Se-Rec-D}} and 
    \textsc{Se-Com-D},
     which are the deterministic counterpart
    of rules \textsc{Se-Rec} and \textsc{Se-Com} of Figure~\ref{fig:lts-types}, respectively.
    Both rules are only defined for \emph{minimal environments} (cf.~Definition~\ref{def:partition}).
    
Consider a fair oracle $\Omega$ (cf.~Definition~\ref{def:fair}).
Rule \textsc{Se-Rec-D} applies when the oracle selects a \emph{single participant} \pp p of the environment~$\Delta$,
 that is the type of \pp p in $\Delta$ is $\mu X.T$.
Given a decreasing set $D$, the rule allows for firing a $\tau_{\pp p}$ transition that unfolds
the type of \pp p to $T\{\mu X.T/X\}$ by means of rule \mbox{\textsc{Se-Rec-A}} of Figure~\ref{fig:alg-lts-types}:
if the configuration $D\diamond\Delta$ can do a semi-algorithmic $\tau_{\pp p}$ step and reach the configuration 
$C\eqdef D\less\Delta\diamond\Delta_1$ with sum continuation $\nabla^\circ$, 
then $D\diamond\Delta$  under $\Omega$ does a deterministic $\tau_{\pp p}$ step and reaches 
the configuration~$C$ with sum continuation~$\nabla^\circ$.

Rule \textsc{Se-Rec-Com-D} holds when the oracle selects a \emph{pair of participants} \pp p 
and \pp q of the environment $\Delta$, and allows a synchronisation on a label~$l$
that updates the types of \pp p and \pp q with their continuations.
Determinism is achieved by identifying~$l$: note that 
there could be multiple possible synchronisations among \pp p and \pp q (cf.~Example~\ref{ex:auth-ered}
and Example~\ref{ex:auth-eredD}).
We exploit assumption (iv) of ~\S~\ref{sec:dlts} and use the scheduling policy of labels to select the 
\emph{first label} in the intersection of the tagged labels of \pp p and \pp q.

Let $D$ be the decreasing set and $l$ be the selected label. 
By the oracle's fairness in~Definition~\ref{def:fair} 
and the intersection hypothesis,  we have that if $\Delta\in D$ then 
the configuration $D\diamond\Delta$ can fire a semi-algorithmic transition $\alpha$, 
\mbox{$\alpha=\sync l p q$} or $\alpha=\sync l q p$, and reach the configuration 
$C\eqdef D\less\Delta\diamond\Delta'$ with sum continuation 
$\Delta''$, where both $\Delta'$ and $\Delta''$ are defined in rule \textsc{Se-Rec-Com-A} in the same Figure.
By applying rule \textsc{Se-Rec-Com-D}, we infer that $D\diamond\Delta$ under $\Omega$ fires 
the deterministic transition $\alpha$ and reaches the configuration $C$ with sum continuation~$\Delta''$.

We conclude by noting that a third case arises for a fair oracle~$\Omega$: \emph{no participants} in $\Delta$ 
are selected.
In this case, Definition~\ref{def:fair} ensures that (i) there are not participants in $\Delta$ 
having a recursive type, and (ii) there are not a pair of participants in $\Delta$ that can synchronise among each 
other. That is, $\Delta$ is stuck (cf.~Definition~\ref{def:stuck}). 

\section{Multistep Transitions}
\label{sec:multistep}
We show that the relation $\Lts{}$ of Definition~\ref{def:closure} can be used to
describe a multistep reduction system w.r.t. 
the non-deterministic transition system \mbox{$D\diamond\Delta\lts\alpha D\diamond\Delta$} of 
Figure~\ref{fig:lts-types}.

\begin{lemma}[Multistep reduction]\label{lem:multistep}
  If $D_1\diamond \Delta_1\lts{\alpha_1}\cdots$
  $\lts{\alpha_{n-1}} D_n\diamond\Delta_n$ and $\Delta_i$ is minimal, 
  $i\in\{1,\dots,n\}$, and 
  $D_n\diamond\Delta_n\nlts{\alpha} D'\diamond \Delta'$
  then there is a fair oracle~$\Omega$  such that 
  $\Omega\triangleright D_1\diamond \Delta_1\Lts{}\widetilde{C_1},D_n\diamond\Delta_n,\widetilde{C_2}$.  
  \end{lemma} 
\begin{proof}
 We build the result by using Lemma~\ref{lem:lts-alts} and by matching 
 non-deterministic transitions $\lts\alpha$ with deterministic transitions $\dlts\alpha$
 scheduled by a fair oracle. 

Consider an oracle $\Omega_{\tt ND}$ selecting processes according to the order 
 of transitions in
$D_1\diamond \Delta_1\lts{\alpha_1}\cdots\lts{\alpha_{n-1}} D_n\diamond\Delta_n$.
It is easy to see that $\Omega_{\tt ND}$ is fair (cf.~Definition~\ref{def:fair}).
Because $\Omega_{\tt ND}$ mimics the behaviour of the non-deterministic scheduler, we can 
repeatedly use Lemma~\ref{lem:lts-alts} and have an exact match by ignoring the sum continuation after each step:
\begin{center}
$
\Omega_{\tt ND}\triangleright D_1\diamond \Delta_1\dlts{\alpha_1}
\cdots\dlts{\alpha_{n-1}}D_n\diamond\Delta_n\blacktriangleright \_
$
\end{center}
Since by hypothesis $D_n\diamond\Delta_n\nlts{\alpha}D'\diamond \Delta'$, we have 
${\tt stuck}_{\Omega_{\tt ND}, D_n}(\Delta_n)$.
By repeated applications of rule \textsc{C-Tra}, ending with the application of rule 
\textsc{C-Rfl} to $D_n\diamond \Delta_n$, we obtain the desired result,
$\Omega_{\tt ND}\triangleright D\diamond \Delta\Lts{}\widetilde{C_1},D_n\diamond\Delta_n,\widetilde{C_2}$.  
\qed
\end{proof}

\section{Sorting Rules}
\label{sec:sorting_rules}
\begin{figure}
\begin{mathpar}
\inference[\textsc{S-Nat}]{n\in\mathbb N}{\Gamma\vdash n\colon \nat}
\and
\inference[\textsc{S-Int}]{z\in\mathbb Z}{\Gamma\vdash z\colon \intk}
\and
\inference[\textsc{S-Str-E}]{}{\Gamma\vdash \textrm{""}\colon \stringk}
\and
\inference[\textsc{S-Str}]{c\in\textrm{Char} \quad \Gamma\vdash s\colon\stringk}{
    \Gamma\vdash c \ \widehat{} \ s\colon \stringk}
\and
\inference[\textsc{S-Bl}]{b\in\{\truek,\falsek\}}{\Gamma\vdash b\colon \bool}
\and
\inference[\textsc{S-Uni}]{}{\Gamma\vdash ()\colon \unit}
\and
\inference[\textsc{S-Var}]{}{\Gamma,x\colon S\vdash x \colon S}
\and
\inference[\textsc{S-Lt}]{    
\Gamma\vdash e_1\colon S\quad
\Gamma\vdash e_2\colon S\quad
S\in\{\nat,\intk\}    
}{
\Gamma\vdash e_1 < e_2\colon \bool    
}
\and
\inference[\textsc{S-Eq}]{    
\Gamma\vdash e_1\colon S\quad
\Gamma\vdash e_2\colon S
}{
\Gamma\vdash e_1 = e_2\colon \bool    
}
\and
\inference[\textsc{S-Neg}]{    
\Gamma\vdash e\colon \bool
}{
\Gamma\vdash \neg e\colon \bool    
}
\and
\inference[\textsc{S-And}]{    
\Gamma\vdash e_1\colon \bool\quad 
\Gamma\vdash e_2\colon \bool
}{
\Gamma\vdash e_1\land e_2 \colon \bool    
}
\and
\inference[\textsc{S-Or}]{    
\Gamma\vdash e_1\colon \bool\quad 
\Gamma\vdash e_2\colon \bool
}{
\Gamma\vdash e_1\vee e_2\colon \bool    
}
\end{mathpar}
\caption{Sorting rules}\label{fig:sorting}
\end{figure}   

The sorting rules used in the typing system (cf.~\S~\ref{sec:typing})
are presented in Figure~\ref{fig:sorting}.

We first present the rules for values~$v$.
Rule \textsc{S-Nat} allows for typing a number~$n$ in the set of naturals $\mathbb N$ with type $\nat$.
Rule \textsc{S-Int} is used to type a number~$n$ in the set of integers $\mathbb Z$ with type $\intk$.
Rule \textsc{S-Str-E} type-checks the empty string.
Non-empty strings are typed by concatenating characters $c$ with strings $s$, written 
$c \ \widehat{} \ s$,
by means of   rule \textsc{S-Str}.
Boolean constants are typed with rule \textsc{S-Bl}.
The only value of type $\unit$ is the empty value, noted $()$, and is typed by using rule 
\textsc{S-Uni}. 

Variables~$x$ are typed with rule \textsc{S-Var}.

Expressions~$e$ are typed with the remaining rules.
Rule \textsc{S-Lt} accepts a ``less than'' test among expressions $e_1$ and $e_2$
whenever $e_1$ and $e_2$ have the same numerical sort.
Rule \textsc{S-Eq} does an equality test among expressions $e_1$ and $e_2$ of the same sort.
Rule \textsc{S-Neq} is used to negate a boolean expression $e$.
Rule \textsc{S-And} type-checks a conjunction of two boolean expressions $e_1$ and $e_2$.
Rule \textsc{S-Or} accepts a disjunction of two boolean expressions $e_1$ and~$e_2$.

\section{Type Derivation of Example~\ref{ex:auth-typing}
\label{sec:auth-derivation}}
Remember the process definitions:
\begin{align*}
P_\ppa &\eqdef  \mu \chi. (P_1 + P_2)
\\     
    P_1 &\eqdef \PIN \ppc\pwd x \textit{Check}_\ppa 
    \\
    P_2 &\eqdef
    \PIN \ppc\ssh x
    \POUT \pps\auth\truek \chi + 
    \PIN \ppc\quit x \INACT
    \\
    \textit{Check}_\ppa &\eqdef
    \ifk\, x= \text{``miau''}\,\thenk\,
    \POUT \pps\auth\truek \chi\,
    \,\elsek\,
    \POUT \pps\fail{} \INACT 
  \end{align*}

Define the types 
\begin{align*}
T_{\textrm{if}}&\eqdef \pps!\auth(\bool).\mu X.T + 
\pps!\fail(\unit).\End\\
T_1 &\eqdef \ppc?{\tt pwd}(\stringk).T_{\textrm{if}}
\\
T_2&\eqdef \ppc?\ssh(\unit).\pps!\auth(\bool).\mu X.T
+ \ppc?{\tt quit}(\unit).\End
\end{align*}

We note that $T_1 = T'\{\mu X.T/X\}$ and
$T_2 = T''\{\mu X.T/X\}$, and in turn
$T_1 + T_2 = T'\{\mu X.T/X\} + T''\{\mu X.T/X\} = T\{\mu X.T/X\}$. 

The type derivation for the authorisation server $P_\ppa$ is:
\begin{mathpar}
\inference[\textsc{T-Rec}]{
  \inference[\textsc{T-Sum}]{
  (A)\ \chi\colon \mu X.T\vdash P_1\colon T_1
  \\
  (B)\ \chi\colon \mu X.T\vdash P_2 \colon T_2
  }  
  {
    \chi\colon \mu X.T\vdash P_1 + P_2\colon T_1 + T_2
  }
}{
\emptyset   \vdash P_\ppa\colon \mu X.T
}
\end{mathpar} 
We  show the details of the derivations (A) and (B) that allow for typing $P_1 $ and $P_2$,
respectively.

Let $\Gamma_1\eqdef \chi\colon \mu X.T,x\colon\stringk$.
The derivation (A) is:
\begin{mathpar}
\inference[\textsc{T-Inp}]{
 \inference[\textsc{T-If}]{
  (C)\  \Gamma_1\vdash \POUT \pps \auth \truek \chi\colon T_{\textrm{if}}
  \\
  (D)\ \Gamma_1\vdash\POUT \pps \fail {} \INACT\colon T_{\textrm{if}}\hspace{2em}
 }{
  \Gamma_1\vdash \textit{Check}_\ppa\colon T_{\textrm{if}}
 } 
}{
  (A)\ \chi\colon \mu X.T\vdash P_1\colon T_1
}
\end{mathpar}  
where 
\begin{mathpar}
 \inference[\textsc{T-Sum-L}]{
  \inference[\textsc{T-Out}]{
    \inference[\textsc{T-Var}]{}{
      \Gamma_1\vdash \chi \colon \mu X.T
    }
  }{
    \Gamma_1\vdash \POUT \pps\auth\truek\chi\colon 
    \pps!\auth(\bool).\mu X.T
  }
  }{
    (C) \ \Gamma_1\vdash \POUT \pps \auth \truek \chi\colon T_{\textrm{if}}
    }
    \and
    \inference[\textsc{T-Sum-R}]{
      \inference[\textsc{T-Out}]{
        \inference[\textsc{T-End}]{}{
          \Gamma_1\vdash \INACT \colon \End
        }  
      }{
        \Gamma_1\vdash \POUT \pps\fail{} \INACT\colon \pps!\fail(\unit).\End
      }  
    }{
      (D) \ \Gamma_1\vdash\POUT \pps \fail {} \INACT\colon T_{\textrm{if}}
    }  
\end{mathpar}    
Next, we show the derivation (B) that type checks $P_2$.

Let $\Gamma_2\eqdef \chi\colon \mu X.T,x\colon\unit$.
The derivation (B) is:

\begin{mathpar}
\inference[\textsc{T-Sum}]{
  (E)\ \chi\colon \mu X.T\vdash \PIN \ppc\ssh x \POUT \pps\auth\truek \chi\colon 
  \ppc?\ssh(\unit).\pps!\auth(\bool).\mu X.T
  \\ 
  (F)\ \chi\colon \mu X.T\vdash\PIN \ppc\quit x \INACT\colon \ppc?\quit(\unit).\End
  \hspace{11em} 
  }{
  (B) \ \chi\colon \mu X.T\vdash P_2\colon T_2}
\end{mathpar}  
where 
\begin{mathpar}
  \inference[\textsc{T-Inp}]{
    \inference[\textsc{T-Out}]{
      \inference[\textsc{T-Var}]{}{
        \Gamma_2\vdash \chi \colon \mu X.T
      } 
    }{
      \Gamma_2\vdash \POUT \pps\auth\truek \chi\colon\pps!\auth(\bool).\mu X.T
    }
  }
  {(E)\ \chi\colon \mu X.T\vdash \PIN \ppc\ssh x \POUT \pps\auth\truek \chi\colon 
  \ppc?\ssh(\unit).\pps!\auth(\bool).\mu X.T
  }
  \and
  \inference[\textsc{T-Inp}]{
    \inference[\textsc{T-End}]{}{
      \Gamma_2\vdash\INACT\colon\End
    }
  }{
  (F)\ \chi\colon \mu X.T\vdash\PIN \ppc\quit x \INACT\colon \ppc?\quit(\unit).\End 
  }
\end{mathpar}

\section{Proof of the Subject Reduction Theorem}
\label{sec:sr}
\newcommand{\W}{{\mathcal W}}

Consider the transition system of sessions of Figure~\ref{fig:session-semantics}, 
the non-deterministic transition system of session environments
of Figure~\ref{fig:lts-types},
and the type system of Figure~\ref{fig:typing}.
We prove that subject reduction holds (cf.~Theorem~\ref{th:sr}).

We start with substitution results.

\begin{lemma}[Value Substitution]\label{lem:substv}
Assume that the following hold:
\begin{enumerate}
\item $\Gamma,x : S\vdash P\colon T$
\item $\Gamma\vdash v\colon S$
\end{enumerate}  
We have that $\Gamma\vdash P\{v/x\}\colon T$.
\end{lemma}
\begin{proof}
By induction on (1).
  \qed  
\end{proof} 

\begin{lemma}[Process Substitution]\label{lem:substP}
  Assume that the following hold:
  \begin{enumerate}
  \item $\Gamma,\chi : T\vdash P\colon U$
  \item $\Gamma\vdash Q\colon T$
  \end{enumerate}  
  We have that $\Gamma\vdash P\{Q/\chi\}\colon U$.
  \end{lemma}
  \begin{proof}
    By structural induction on $P$. See~App.~\ref{sec:coq}.  
    \qed  
  \end{proof}

We then introduce three simple results that will be used in the proof of Theorem~\ref{th:sr}.

Remember the definition of tagged labels of types ${\cal L}(T)$ (cf.~\S~\ref{sec:compliance}), 
and
let the \emph{labels of actions}, noted ${\cal L}(\alpha)$, be defined as:
${\cal L}(\AIN  \ppr l v)\eqdef\{l\}$,
${\cal L}(\AOUT \ppr l v )\eqdef\{l\}$,
${\cal L}(\tau)\eqdef\emptyset$.

Define the \emph{labels of processes}, noted ${\cal L}(P)$, 
as ${\cal L}(\ppr?l(x).P)\eqdef\{l\}$,
${\cal L}(\pp r!l\langle e \rangle.P)$ $\eqdef \{l\}$,
${\cal L}(P_1 + P_2)\eqdef{\cal L}(P_1)\cup{\cal L}(P_2)$,
${\cal L}(\ifk\; e \; \thenk\; P_1\; \elsek\; P_2)\eqdef 
{\cal L}(P_1)\cup{\cal L}(P_2)$,
and ${\cal L}(P)\eqdef\emptyset$ otherwise.

\begin{lemma}\label{lem:typ-inv}
If $\Gamma\vdash P\colon T$ and 
$l\in {\cal L}(P)$ then there exists $S$ s.t.  $l@S\in{\cal L}(T)$.
\end{lemma}
\begin{proof}
By induction on $\Gamma\vdash P\colon T$.
\qed
\end{proof}  

\begin{lemma}\label{lem:lts-inv}
If $\pp p\triangleright P\lts\alpha \pp p\triangleright P'$ and ${\cal L}(\alpha)=\{l\}$ 
then 
$l\in{\cal L}(P)$.
\end{lemma}
\begin{proof}
By induction on $\pp p\triangleright P\lts\alpha \pp p\triangleright P'$.
\qed
\end{proof}

\begin{lemma}\label{lem:tlts-inv}
If $T\lts\alpha T'$ and ${\cal L}(\alpha)=\{l\}$ then there exists $S$ s.t.  
$l@S\in{\cal L}(T)$.
\end{lemma}
\begin{proof}
  By induction on $T\lts\alpha T'$.
  \qed
\end{proof}

The following lemmas are needed to establish Corollary~\ref{cor:comp-preserve}.

\begin{lemma}[Semi-algorithmic preservation]\label{lem:semi-alg-preserve}
  The following hold.
  \begin{enumerate}
  \item If $T\lts\alpha T'$ then
  there exists $R''$ s.t. 
  $T\alts\alpha T'\blacktriangleright R''$;
  \item If $D\diamond \Delta \lts\alpha D'\diamond \Delta'$ then
  there exists $\Delta''$ s.t. 
  $D\diamond \Delta \alts\alpha D'\diamond \Delta'\blacktriangleright\Delta''$.
  \end{enumerate}  
  \end{lemma}
  \begin{proof}
    The semi-algorithmic rules $\alts{}$ for types and environments are in App.~\ref{sec:alg-lts-rules}.
  
    (1) is proved by induction on  $T\lts\alpha T'$.
  
  (2) is proved by induction on $D\diamond \Delta \lts\alpha D'\diamond \Delta'$ while relying on (1). 
  \qed
  \end{proof}

  \myparagraph{Proof of Lemma~\ref{lem:lts-alts}}
  By inversion  of $D\diamond \Delta \lts\alpha D'\diamond \Delta'$, we have 
  (i) $\alpha =\tau$ or
  (ii) $\alpha=\sync l p q$. 
  Consider fairness in Definition~\ref{def:fair}.
  
  In case (i),  it is clear that  there is $\pp p$ and $T$ such that $\Delta(\pp p)=\mu X.T$ and 
  $\Delta'(\pp p)=T\{\mu X.T /X\}$.
  We let $\Omega$ choose \pp p, and infer the desired result by using 
  rule \textsc{Se-Rec-D} of Figure~\ref{fig:lts-types-det}.

  In case (ii), we let $\Omega$ choose \pp p and \pp q, 
  and assume a scheduling policy  picking~$l$ as first label.
  We apply Lemma~\ref{lem:semi-alg-preserve} and infer that there is $\Delta''$ s.t.:
  $$
  D\diamond \Delta \alts\alpha D'\diamond \Delta'\blacktriangleright\Delta''
  $$ 
  The result then follows from an application of \textsc{Se-Com-D} in 
  Figure~\ref{fig:lts-types-det}.
  \qed

  \myparagraph{Proof of Lemma~\ref{sec:compliance-dlts}}
  First note that  compliance requires well-formedness (cf.~Definition~\ref{def:compliance}), 
  and in turn $\wf(\Delta)$.
  To obtain $\wf(\Delta')$, we use preservation of well-formedness by $T\lts\alpha T'$ 
  (cf. Lemma~\ref{lem:wf-preserve}).

  By inversion of  $\Omega\triangleright D\diamond \Delta \dlts\alpha D'\diamond \Delta'\blacktriangleright\Delta''$,
  we infer 
  (i) $\alpha=\tau$ and $\Delta=\Delta_0,\ppp\colon T$ and $\Delta'=\Delta_0,\ppp\colon T'$ or
  (ii) $\alpha=\sync l p q $ and $\Delta=\Delta_0,\ppp\colon T_1,\ppq\colon T_2$ and 
  $\Delta'=\Delta_0,\ppp\colon T',\ppq \colon T''$.
  In case (i), by inversion, we obtain $T\lts\tau T'$:  
  we apply~Lemma~\ref{lem:wf-preserve} and infer $\wf(T')$.
  In case (ii), by inversion, we obtain $T_1\lts{\AOUT \ppq l v} T'$ and 
  $T_2\lts  {\AIN \ppp l v} T''$: 
  we apply~Lemma~\ref{lem:wf-preserve} on both transitions and  infer $\wf(T')$ 
  and $\wf(T'')$. 
  
  Compliance then follows from Definition~\ref{def:closure} and from the following result:
  ${\tt closure}_{\Omega,D'}(\Delta')\subseteq {\tt closure}_{\Omega,D}(\Delta)$.
  In fact, compliance is defined on $\Delta$ iff its closure  does not contain failures $\tt err$:
  \begin{mathpar}
\inference{
  \Omega\triangleright D'\diamond \Delta'\Lts{}\widetilde{C_1}\and
  \Omega\triangleright D'\diamond \Delta''\Lts{}\widetilde{C_2}
}{
  \Omega\triangleright D\diamond\Delta\Lts{}\widetilde{C_1},\widetilde{C_2}
}
  \end{mathpar}  
Remember that  ${\tt closure}_{\Omega,D}(\Delta)$ is obtained by stripping the decreasing sets from  
$\widetilde{C_1},\widetilde{C_2}$, and that 
${\tt closure}_{\Omega,D'}(\Delta')$ is obtained by stripping the decreasing sets from 
$\widetilde{C_1}$.
Thus if $\widetilde{C_1}=D^1_1\diamond \Delta^1_1,\dots,D^1_n\diamond \Delta^1_n$ and 
$\widetilde{C_2}=D^2_1\diamond \Delta^2_1,\dots,D^2_m\diamond \Delta^2_m$, then
${\tt closure}_{\Omega,D}(\Delta)=\{\Delta^1_1,\dots,\Delta^1_n,\Delta^2_1,\dots,\Delta^2_m\}$
and
${\tt closure}_{\Omega,D'}(\Delta')=\{\Delta^1_1,\dots,\Delta^1_n\}$, as required.
  \qed

\medskip
Next, we establish subject reduction for single threads firing I/O actions.

\begin{lemma}[I/O subject reduction]\label{lem:sr}
    Consider a single thread $\pp p\triangleright P$ and a well-formed type $T$,
    and assume that the following hold:
    \begin{enumerate}
    \item $\Gamma\vdash P\colon T$
    \item $\pp p\triangleright P \lts\alpha \pp p\triangleright P'$ with 
    $\alpha\in\{\AIN \ppq l v,\AOUT \ppq l v\}$
    \item $T\lts\alpha T'$ 
    \end{enumerate}
    We have $\Gamma\vdash P'\colon T'$.
  \end{lemma}
\begin{proof}
By induction on (2), eventually relying on Lemma~\ref{lem:substv}. 
Three cases arise, corresponding to rules \textsc{R-Inp}, \textsc{R-Out}, and \textsc{R-Sum-L}.
We outline \textsc{R-Inp} and \textsc{R-Sum-L}; the case of rule \textsc{R-Out} is analogous to \textsc{R-Inp}.

In case  \textsc{R-Inp}, we have $P=\pp q?l(x).Q$, $\alpha=\AIN \ppq l v$,  
and $T=\pp q?l(S).T'$,  for some process $Q$ and  sort  $S$.
Moreover, $P' = Q\{v/x\}$.
From the inversion of (3), we obtain  $\emptyset\vdash v\colon S$. 
We use weakening and obtain  $\Gamma\vdash v\colon S$ (cf.~App.~\ref{sec:sorting_rules}).

By inversion of (1), we have $\Gamma,x\colon S\vdash Q\colon T'$.
We apply substitution (Lemma~\ref{lem:substv}) and infer the desired result,
$\Gamma\vdash P'\colon T'$.

In case \textsc{R-Sum-L}, we have $P=P_1 + P_2$ and 
$\pp p\triangleright P_1 \lts\alpha \pp p\triangleright P'$.
From the inversion of (1), we infer $T= T_1 + T_2$ and three sub-cases corresponding to rules \textsc{T-Sum}, 
rule \textsc{T-Sum-L}, and rule \textsc{T-Sum-R}.
From the inversion of (3), we infer that 
$$T_1\lts\alpha T'$$

In sub-case \textsc{T-Sum}, the hypothesis is $\Gamma\vdash P_1\colon T_1$ and 
$\Gamma\vdash P_2\colon T_2$. We apply the I.H. on $\pp p\triangleright P_1 \lts\alpha \pp p\triangleright P'$
and infer the desired result:
$\Gamma\vdash P'\colon T'$.

In sub-case \textsc{T-Sum-L}, the hypothesis is $\Gamma\vdash P\colon T_1$. 
We apply the I.H. on $\pp p\triangleright P \lts\alpha \pp p\triangleright P'$ 
and infer the desired result:
$\Gamma\vdash P'\colon T'$.

In sub-case \textsc{T-Sum-R}, the hypothesis is $\Gamma\vdash P\colon T_2$. 
This gives rise to a contradiction.

To see this, consider  the definition of labels of process, noted $\cal L$, at the 
beginning of the section. 
From Lemma~\ref{lem:typ-inv}, Lemma~\ref{lem:lts-inv} and 
Lemma~\ref{lem:tlts-inv} we obtain, respectively:

(1) 
$\exists S\,.\, l@S\in{\cal L}(T_2)$;

(2) 
$l\in{\cal L}(P_1)$;

(3) 
$\exists S\,.\, l@S\in{\cal L}(T_1)$.

\medskip
In particular, from (2) we infer $l\in{\cal L}(P)$.
Hence applying (1) to this conclusion, we have that (4) there is $S$ such that 
$l@S\in{\cal L}(T_2)$.

By gluing (3) and (4), we obtain that $l$ occurs both in $T_1$ and $T_2$, which contradicts 
well-formedness (cf.~App.~\ref{sec:wf}, Definition~\ref{def:uniform}.1).
\qed
\end{proof}    

In the remainder of the section, 
we use the \emph{notation} $(\Delta,\ppp\colon T)\less\ppp$ to indicate the environment~$\Delta$.

We use a Lemma to tackle the minimal partition and compliance assumptions in the proof of subject reduction.

\begin{lemma}\label{lem:sr-minimal}
Let $D\diamond\Delta$ be a configuration.
Assume that there are $\Delta_1,\dots,\Delta_k$ s.t. $\partition\Delta(\Delta_1,\dots,\Delta_k)$ and 
$\keyword{comp}(\Delta_j)$, for all $j\in \{1,\dots, k\}$.

If $D\diamond\Delta\lts\alpha D\less\Delta\diamond\Delta'$ then there is $h\in \{1,\dots,k\}$ and
$\Delta'_h$ s.t.:
\begin{enumerate}
\item $\Delta'=\Delta_1\cup\cdots\cup\Delta_{h-1}\cup\Delta'_h\cup \Delta_{h+1}\cup\cdots\Delta_k$;
\item $\{\Delta_1,\dots,\Delta_{h-1},\Delta'_h,\Delta_{h+1},\dots,\Delta_k\}$ is a partition of
$\Delta'$;
\item $\keyword{comp}(\Delta'_j)$, for all $j\in \{1,\dots, k\}$.
\end{enumerate}
\end{lemma}  
\begin{proof}
By inversion of $D\diamond\Delta\lts\alpha D\less\Delta\diamond\Delta'$, we infer that two cases arise:
\begin{enumerate}[label=(\roman*)]
\item There is $\ppp\in\dom(\Delta)$ s.t. $\Delta(\ppp)=\mu X.T$ and 
$\Delta'=\Delta\less\ppp,\ppp\colon T\{\mu X.T/X\}$
\item There is $\{\ppp,\ppq\}\subseteq\dom(\Delta)$ s.t.
$\Delta(\ppp)\lts{\AIN \ppq l v} T'_\ppp$
and
$\Delta(\ppq)\lts{\AOUT \ppp l v}  T'_\ppq$ and \\
$\Delta'=\Delta\less\ppp\less \ppq,\ppp\colon T'_\ppp,\ppq\colon T'_\ppq$
\end{enumerate}

\medskip\noindent\emph{Subcase (i)}.
First note that by the partition assumption, there is exactly one $j\in\{1,\dots,k\}$ s.t.
$\ppp\in\dom(\Delta_j)$. 
Let $h\in\{1,\dots,k\}$ and $\ppp\in\dom(\Delta_h)$.
Define $\Delta'_h \eqdef \Delta_h\less\ppp,\ppp\colon T\{\mu X.T/X\}$.
Then (1) and (2) follow by the hypothesis  $\partition\Delta(\Delta_1,\dots,\Delta_k)$, and by definition 
of $\Delta'_h$.

To prove (3), we need to show that $\keyword{comp}(\Delta'_h )$, since all other minimal environments
are compliant by hypothesis.

Take a decreasing set $D'$ such that $\Delta_h\in D'$.
By applying \textsc{Se-Rec}, \textsc{Se-Top} of Figure~\ref{fig:lts-types}
we infer:
\begin{align}
  D'\diamond \Delta_h\lts\tau D'\less{\Delta_h}\diamond \Delta'_h 
\end{align}
By Corollary~\ref{cor:comp-preserve}, we infer $\keyword{comp}(\Delta'_h )$.

\medskip\noindent\emph{Subcase (ii)}.
To see (1) and (2), note that,
by the partition assumption, there is exactly one $j\in\{1,\dots,k\}$ s.t.
$\{\ppp,\ppq\}\subseteq\dom(\Delta_j)$. 
To see that this holds, assume $\ppp\in\dom(\Delta_i)$, 
$\ppq\in\dom(\Delta_j)$, $\{i,j\}\subseteq\{1,\dots,k\}$, $j\ne i$: we show that this leads to a contradiction.

In fact, from  inversion of $\Delta_i(\ppp)\lts{\AIN \ppq l v} T'_\ppp$, we deduce that
$\ppq\in\keyword{parties}(\Delta_i(\ppp))$, and in turn 
$\ppq\in\keyword{parties}(\Delta_i)$. From $\ppq\in\keyword{parties}(\Delta_j)$ we deduce 
that $\keyword{parties}(\Delta_i)$ and $\keyword{parties}(\Delta_j)$ are not disjoint: 
this contradicts the hypothesis $\partition\Delta(\Delta_1,$ $\dots,\Delta_k)$.

We proceed as in case (i) and close the result.
Let $h\in\{1,\dots,k\}$ and $\{\ppp,\ppq\}\subseteq\dom(\Delta_h)$.
Define $\Delta'_h \eqdef \Delta_h\less\ppp\less\ppq,\ppp\colon T'_\ppp,\ppq\colon T'_\ppq$.
Items (1) and (2) follow by the hypothesis  $\partition\Delta(\Delta_1,\dots,\Delta_k)$, and by definition 
of $\Delta'_h$.

To show (3), we need to prove that  
$\keyword{comp}(\Delta'_h )$ holds, since all other minimal environments in the partition are compliant by hypothesis.
Take a decreasing set $D'$ such that $\Delta_h\in D'$.
By applying \textsc{Se-Com}, \textsc{Se-Top} of Figure~\ref{fig:lts-types}
we infer:
\begin{align}
  D'\diamond \Delta_h\lts\tau D'\less{\Delta_h}\diamond \Delta'_h 
\end{align}
By Corollary~\ref{cor:comp-preserve}, we infer $\keyword{comp}(\Delta'_h )$.
\qed
\end{proof}

\begin{corollary}\label{cor:sr-minimal}
  Let $D\diamond\Delta$ be a configuration.
  Assume that there are $\Delta_1,\dots,\Delta_k$ s.t. $\partition\Delta(\Delta_1,\dots,\Delta_k)$ and 
  $\keyword{comp}(\Delta_j)$, for all $j\in \{1,\dots, k\}$.
  
  If $D\diamond\Delta\lts\alpha D\less\Delta\diamond\Delta'$ then there is $h\in \{1,\dots,k\}$ 
  and $\Delta'_h$ s.t.
  \begin{enumerate}
  \item $\partition{\Delta'}(\Delta_1,\dots,\Delta_{h-1},\Delta'_h,\Delta_{h+1},\dots,\Delta_k)$
  \item $\keyword{comp}(\Delta'_j)$, for all $j\in \{1,\dots, k\}$.
  \end{enumerate}
  \end{corollary}
  \begin{proof}
  We only need to show item (1); item (2) is contained in Lemma~\ref{lem:sr-minimal}, and is re-stated
  here for convenience. 
   
  We need to show that 
  (\textbf{A})
  $\{i,j\}\subseteq\{1,\dots,k\}$ and $i\ne j$ imply 
  $\keyword{parties}(\Delta'_i)\cap\keyword{parties}(\Delta'_j)=\emptyset$, and
  (\textbf{B}) $i\in\{1,\dots,k\}$ implies $\keyword{minimal}(\Delta'_i)$.
  
  Item (A) follows from Lemma~\ref{lem:sr-minimal}.2, while item (B) follows from 
  Lemma~\ref{lem:sr-minimal}.3: if $\keyword{comp}(\Delta'_i)$
  then $\keyword{mininal}(\Delta'_i)$.
  \qed
  \end{proof}

We can now tackle the main result of this section: well-typed sessions reduce to well-typed sessions.

\myparagraph{Proof of Theorem~\ref{th:sr}}
By induction on (2).
Case \textsc{R-Rec} relies on Lemma~\ref{lem:substP} and Corollary~\ref{cor:sr-minimal}.
Case \textsc{R-Com} makes use of Lemma~\ref{lem:substv}, Lemma~\ref{lem:sr}, and 
 Corollary~\ref{cor:sr-minimal}.
Case \textsc{R-IfT} follows by definition (cf.~\textsc{T-If} of Figure~\ref{fig:typing}).
Case \textsc{R-Struct} uses the I.H. and Lemma~\ref{lem:types-struct}.

\begin{description}
\item[\textsc{R-Rec}]
We have
\begin{align}
&\Gamma\Vdash   {\pp p}\lhd  \mu\chi. P \PARI {\pp r}_i\lhd R_i\colon\Delta
\label{eq:sr1-fold-typ}
\\
&
{\pp p}\lhd \mu\chi.P \PARI {\pp r}_i\lhd R_i
\lts\tau
{\pp p}\lhd P\{\mu\chi. P/\chi\}\PARI {\pp r}_i\lhd R_i 
\end{align}
and find 
$$
D\diamond \Delta\lts\tau D\less\Delta\diamond \Delta'
$$ 
where
\begin{align}
\Delta &= {\pp p}\colon \mu X. T, 
{\pp r}_1\colon T_1,\dots, {\pp r}_n\colon T_n
\\
\Delta&\in D
\\
\Delta' &= {\pp p}\colon T\{\mu X.T/X\}, 
{\pp r}_1\colon T_1,\dots, {\pp r}_n\colon T_n
\end{align}   
We show that
$$
\textit{GOAL: }
\Gamma\Vdash {\pp p}\lhd P\{\mu\chi. P/\chi\}\PARI {\pp r}_i\lhd R_i \colon \Delta'
$$
By inversion of (\ref{eq:sr1-fold-typ}), we find $k$ s.t. 
\begin{align}
&\partition{\Delta}(\Delta_1,\dots,,\Delta_k)
\\
&\forall j\in \{1,\dots, k\}\,.\, \keyword{comp}(\Delta_j)
\end{align}

By Corollary~\ref{cor:sr-minimal}, we infer that there are $h\in\{1,\dots,k\}$ and $\Delta'_h$ s.t. 
\begin{align}
&\partition{\Delta'}(\Delta_1,\dots,\Delta_{h-1},\Delta'_h,\Delta_{h+1},\dots,\Delta_k)
\label{eq:sr1-fold-comp0}
\\
&\forall j\in \{1,\dots, k\}\,.\, \keyword{comp}(\Delta'_j)  
\label{eq:sr1-fold-comp1}
\end{align}

To prove the goal by means of \textsc{T-Ses}, we require
(\ref{eq:sr1-fold-comp0}) and (\ref{eq:sr1-fold-comp1}) and :
\begin{align}
&\Gamma\vdash P\{\mu \chi. P/\chi\}\colon T\{\mu X.T/X\}
\label{eq:sr1-fold-fold}
\\
&\Gamma_i \vdash R_i\colon T_i \quad (i=1, \dots, n)
\label{eq:sr1-fold-threads}
\end{align}    
To see (\ref{eq:sr1-fold-fold}), we invert (\ref{eq:sr1-fold-typ}) and obtain:
\begin{align}
\Gamma\vdash \mu \chi.P\colon\mu X.T
\label{eq:sr1-fold-typ-2}
\end{align}
By inverting (\ref{eq:sr1-fold-typ-2}), we have 
\begin{align}
  \Gamma,\chi\colon\mu X.T\vdash P\colon T\{\mu X.T/X\}
  \label{eq:sr1-fold-typ-3}
\end{align}
We apply Lemma~\ref{lem:substP} to (\ref{eq:sr1-fold-typ-2}) and 
(\ref{eq:sr1-fold-typ-3}),
and 
infer (\ref{eq:sr1-fold-fold}).

Item (\ref{eq:sr1-fold-threads})  follows straightforwardly from the inversion of 
(\ref{eq:sr1-fold-typ}).

\medskip
\item[\textsc{R-Comm}]
We have
\begin{align}
&\Gamma\Vdash  {\pp p}\lhd P\PAR {\pp q}\lhd Q\PARI {\pp r}_i\lhd R_i\colon\Delta
\label{eq:sr1-comm-typ}
\\
&
{\pp p}\lhd P\PAR {\pp q}\lhd Q\PARI {\pp r}_i\lhd R_i
    \lts{l@{\pp p}\bowtie {\pp q}}
    {\pp p}\lhd P'\PAR {\pp q}\lhd Q'\PARI {\pp r}_i\lhd R_i 
    \label{eq:sr-comm-Rec}
\end{align}
and find
$$
D\diamond \Delta\lts{l@{\pp p}\bowtie {\pp q}} D\less\Delta\diamond \Delta'
$$
where
\begin{align}
    \Delta &= {\pp p}\colon T_{\pp p}, {\pp q}\colon T_{\pp q},
    {\pp r}_1\colon T_1,\dots, {\pp r}_n\colon T_n
    \\
    \Delta&\in D
    \\
    \exists l, v\,&.\, T_{\pp p}\lts{{\pp q}?l\langle v\rangle} T'_{\pp p}
    \wedge
    T_{\pp q}\lts{{\pp p}!l(v)}  T'_{\pp q}
    \label{eq:sr-comm-prefix}
    \\
    \Delta' &= {\pp p}\colon T'_{\pp p}, {\pp q}\colon T'_{\pp q},
    {\pp r}_1\colon T_1,\dots, {\pp r}_n\colon T_n
    \end{align}   
    We show that
    $$
    \textit{GOAL: }
    \Gamma\Vdash {\pp p}\lhd P' \PAR {\pp q}\lhd Q' \PARI {\pp r}_i\lhd R_i \colon \Delta'
    $$

    By inversion of (\ref{eq:sr1-comm-typ}), we find $k$ s.t. 
\begin{align}
&\partition{\Delta}(\Delta_1,\dots,,\Delta_k)
\\
&\forall j\in \{1,\dots, k\}\,.\, \keyword{comp}(\Delta_j)
\end{align}

By Corollary~\ref{cor:sr-minimal}, we infer that there are $h\in\{1,\dots,k\}$ and $\Delta'_h$ s.t. 
\begin{align}
&\partition{\Delta'}(\Delta_1,\dots,\Delta_{h-1},\Delta'_h,\Delta_{h+1},\dots,\Delta_k)
\label{eq:sr1-comm-comp0}
\\
&\forall j\in \{1,\dots, k\}\,.\, \keyword{comp}(\Delta'_j)  
\label{eq:sr1-comm-comp1}
\end{align}

To prove the goal by means of \textsc{T-Ses}, we require
(\ref{eq:sr1-comm-comp0}) and (\ref{eq:sr1-comm-comp1}) and :
\begin{align}
&\Gamma\vdash P'\colon T'_{\pp p} \qquad \Gamma\vdash Q'\colon T'_{\pp q}
\label{eq:sr1-comm-processes}
\\
&\Gamma_i \vdash R_i\colon T_i \quad (i=1, \dots, n)
\label{eq:sr1-comm-threads}
\end{align}    

To see (\ref{eq:sr1-comm-processes}), we note that from 
the inversion of (\ref{eq:sr-comm-Rec}) and 
(\ref{eq:sr1-comm-typ}) we infer
\begin{align}
    &{\pp p}\lhd P\lts{{\pp q}?l(v)} {\pp p}\lhd P' &
    &{\pp q}\lhd Q\lts{{\pp p}!l\langle v\rangle} {\pp q}\lhd Q' 
    \label{eq:sr-comm-prefixes}
    \\
    &\Gamma\vdash P\colon T_{\pp p} &
    &\Gamma\vdash Q \colon T_{\pp q}
    \label{eq:sr-comm-typ_of_prefixes}
\end{align}    
From (\ref{eq:sr-comm-prefix}), (\ref{eq:sr-comm-prefixes}), (\ref{eq:sr-comm-typ_of_prefixes}),
and applications of Lemma~\ref{lem:sr}, we infer the desired result, that is 
$\Gamma\vdash P'\colon T'_{\pp p}$ and $\Gamma\vdash Q'\colon T'_{\pp q}$ (cf. (\ref{eq:sr1-comm-processes})).

Item (\ref{eq:sr1-comm-threads})  follows straightforwardly from the inversion of 
(\ref{eq:sr1-comm-typ}).

\medskip
\item[\textsc{R-IfT}]
We have
\begin{align}
&\Gamma\Vdash  {\pp p}\lhd \ifk\; e \; \thenk\; P\; \elsek\; Q\PARI {\pp r}_i\lhd R_i
\colon\Delta
\label{eq:sr-if-1}
\\
&
{\pp p}\lhd \ifk\; e \; \thenk\; P\; \elsek\; Q\PARI {\pp r}_i\lhd R_i
  \lts\tau
  {\pp r}\lhd P\PARI {\pp r}_i\lhd R_i
\label{eq:sr-if-2}
\end{align}
where
\begin{align}
\Delta &= {\pp p}\colon T, 
{\pp r}_1\colon T_1,\dots, {\pp r}_n\colon T_n
\end{align}   
We show that
$$
\textit{GOAL: }
\Gamma\Vdash {\pp p}\lhd P\PARI {\pp r}_i\lhd R_i \colon \Delta
$$
Since the environment is the same of the one in (\ref{eq:sr-if-1}),
to prove the goal by means of \textsc{T-Ses} we only require:
\begin{align}
&\Gamma\vdash P\colon T
\label{eq:sr-if-3}
\\
&\Gamma_i \vdash R_i\colon T_i \quad (i=1, \dots, n)
\label{eq:sr-if-4}
\end{align}
Both results follow from the inversion of (\ref{eq:sr-if-1}). 
In particular, the inversion produces the typing 
\begin{align}
\Gamma\vdash \ifk\; e \; \thenk\; P\; \elsek\; Q\colon T
\label{eq:sr1-if-6}
\end{align} 
The inversion of the typing  in (\ref{eq:sr1-if-6}) produces the proof of 
(\ref{eq:sr-if-3}).

\medskip
\item[\textsc{R-Struct}]
We have
\begin{align}
&\Gamma\Vdash   {\cal M}_1\colon\Delta
\label{eq:sr-struct-1}
\\
&
{\cal M}_1\lts\alpha{\cal M}_2
\label{eq:sr-struct-2}
\end{align}

The goal is:
\begin{align*}
\textit{GOAL: }
&\Gamma\Vdash {\cal M}_2 \colon \Delta \text{ or }
\\
&\exists \Delta' \text{ s.t. }
D \diamond \Delta\lts\alpha D\less\Delta\diamond \Delta' \text{ and } 
\Gamma\Vdash {\cal M}_2\colon \Delta'
\end{align*}

By inversion of (\ref{eq:sr-struct-2}), we obtain:
\begin{align}
  &{\cal M}_1\pequiv {\cal M}'_1 
  \label{eq:sr-struct-3}
  \\
  &{\cal M}'_1 \lts\alpha {\cal M}'_2
  \label{eq:sr-struct-4}\\
  &{\cal M}_2 \pequiv {\cal M}'_2
  \label{eq:sr-struct-4b}
\end{align}

We apply Lemma~\ref{lem:types-struct} to (\ref{eq:sr-struct-1}) and 
(\ref{eq:sr-struct-3}) and obtain:
\begin{align}
\Gamma\Vdash {\cal M}'_1 \colon \Delta
\label{eq:sr-struct-5}
\end{align}  

The I.H. specialised with (\ref{eq:sr-struct-4}) and (\ref{eq:sr-struct-5}) is: 
\begin{align}
&\Gamma\Vdash {\cal M}'_2 \colon \Delta\text{ or }
\label{eq:sr-struct-6}
\\
&\exists\Delta' \,.\, 
D\diamond \Delta\lts\alpha D\less\Delta\diamond \Delta' \text{ and } 
\Gamma\Vdash {\cal M}'_2\colon \Delta'
\label{eq:sr-struct-7}
\end{align} 

In case (\ref{eq:sr-struct-6}) we apply Lemma~\ref{lem:types-struct} 
to (\ref{eq:sr-struct-6}) and (\ref{eq:sr-struct-4b}),
and infer 
$\Gamma\Vdash {\cal M}_2\colon \Delta$, which fulfils the first option of the goal.

In case (\ref{eq:sr-struct-7}) we apply Lemma~\ref{lem:types-struct} 
to $\Gamma\Vdash {\cal M}'_2\colon \Delta'$ in (\ref{eq:sr-struct-7})
and (\ref{eq:sr-struct-4b}),
and infer 
$\Gamma\Vdash {\cal M}_2\colon \Delta'$, 
which fulfils the second option of the goal.
\qed
\end{description}

\section{Proof of the Progress Theorem}
\label{sec:progress}
\begin{table}
  \begin{center}
  \begin{tabular}{||c l ||} 
   \hline
   \textbf{Step} & \textbf{Action} \\ [0.5ex] 
   \hline\hline
   \hline
   \textbf 1 & Define unblocked well-formed sessions $\cal M$: \\
   &if ${\cal M}$ does not reduce then 
   ${\cal M}$ only contains $\INACT$ processes 
   \\ 
   \hline
   \textbf 2 & Characterise environments of sessions of step \textbf{1}: \\
   &if $\Gamma\Vdash {\cal M}\colon \Delta$ with $\cal M$ closed and 
   ${\cal M}$ does not reduce then forall $\alpha$ and $\Delta'$:  
   \\
   &$\Delta\nlts\alpha\Delta'$
   \\
   \hline
   \textbf 3 & Progress Theorem: \\
   &if $\Gamma\Vdash {\cal M}\colon \Delta$ with $\cal M$ closed and ${\cal M}$ does not reduce then ${\cal M}$ 
   only contains 
   \\
   &$\INACT$ processes (cf. (1))
   \\
   \hline
   \textbf 4 &Proof of \textbf{3}
   \\
   i&From (\textbf{2}) and (\textbf{3}) we have $\partition\Delta(\Delta_1,\dots,\Delta_k)$ with
   $\Delta_j\nlts\alpha\Delta'_j$, for all 
   \\
   &$j\in\{1,\dots,k\}$. We deduce that 
  ${\tt stuck}_{\Omega, D_j}(\Delta_j)$, for all $\Omega$ and $D_j$.
   \\
   &By Def.~\ref{def:closure}: ${\tt closure}_{\Omega, D_j}(\Delta_j)=\Delta_j$.
   \\
   ii&By Def.~\ref{def:compliance} and (\textbf i) we have that 
   all $\Delta_j$ in the minimal 
   \\
   &partition of $\Delta$
   are not mismatches 
   and are not deadlocks (Def.~\ref{def:error}).
   \\
   iii &That is: (\textbf A) $\Delta_j$ does move or 
   (\textbf{B}) $\Delta_j(\ppp)=\End$, for all $\ppp\in\dom(\Delta_j)$ 
   \\
   iv &Since (\textbf A) contradicts (\textbf{i}), we have (\textbf B)
   \\
   v & We show that $\Gamma\Vdash {\cal M}\colon \ppp_1\colon\End,\dots,\ppp_n\colon\End$ 
   and $\ppp_i\lhd P_i$ occurs in ${\cal M}$, 
   \\
   &$0\leq i \leq n$, imply  (\textbf C) $P_i=\INACT$ or 
   (\textbf{D}) $P_i=\ifk\, e \, \thenk\, Q_1\, \elsek\,  Q_2$
   \\
   vi &Since (\textbf{D}) contradicts the hypothesis in (\textbf{3}) as $\cal M$ moves, 
   we have (\textbf C).
   \\
   \hline
  \end{tabular}
  \end{center}
  \caption{Outline of the proof of progress}
  \label{tab:progress-outline}
\end{table}
In this section we show that well-typed closed sessions do not get stuck (cf.~Theorem~\ref{th:progress}).
The outline of the proof is in Table~\ref{tab:progress-outline}.

We start by defining a notion of well-formedness of sessions that is specular to the one 
for types in App.~\ref{sec:wf}.

These are the closed sessions accepted by the typing system of Figure~\ref{fig:typing}, as we will show.

\subsection{Well-formed sessions}
We overload the functions for types defined  in App.~\ref{sec:wf}.

Let $\keyword{labels}$ map processes to \emph{multisets} of labels:
$\keyword{labels}(\PIN \ppr l x P) \eqdef \multiset l$, 
$\keyword{labels}(\POUT \ppr l e P)\eqdef \multiset l$, 
$\keyword{labels}(\ifk\, e \, \thenk\, P_1 \elsek\,  P_2) \eqdef$ 
$\keyword{labels}(P_1) \uplus \keyword{labels}(P_2)$,
and
$\keyword{labels}(P_1 + P_2) $ 
$\eqdef \keyword{labels}(P_1) \uplus \keyword{labels}(P_2)$,
and $\keyword{labels}(P)\eqdef\emptyset$ otherwise.

Let $\keyword{polarity}$ be a partial function mapping processes to polarities 
$p\in\{!, ?\}$, and let $\keyword{polarity}(P)\downarrow$ whenever 
$\keyword{polarity}$ is defined on process $P$.
Let
$\keyword{polarity}(\PIN \ppr l x P)$ $ \eqdef ?$,
$\keyword{polarity}(\POUT \ppr l e P)\eqdef !$, 
$\keyword{polarity}(\ifk\, e \, \thenk\, P_1\, \elsek\,  P_2) \eqdef p $ 
whenever $\keyword{polarity}(P_1)$ $= p$ and $\keyword{polarity}(P_2)= p$,
$\keyword{polarity}(P_1 + P_2) \eqdef p $ 
whenever $\keyword{polarity}(P_1) = p$ and $\keyword{polarity}(P_2)= p$,
and $\keyword{polarity}(P)$ be undefined otherwise.

Let $\keyword{participant}$ be a  partial function mapping processes to participants,
and let $\keyword{participant}(P)\downarrow$ when $\keyword{participant}$ is defined on process $P$.
Let 
$\keyword{participant}(\PIN \ppr l x P)$ $\eqdef \ppr$,
$\keyword{participant}(\POUT \ppr l e P)\eqdef \ppr$, 
$\keyword{participant}(\ifk\, e \,$ $ \thenk\, P_1\, \elsek\,  P_2
) \eqdef \ppr $
whenever $\keyword{participant}(P_1)=\ppr$ and 
$\keyword{participant}(P_2)=\ppr$,
$\keyword{participant}(P_1 + P_2) \eqdef \ppr $
whenever $\keyword{participant}(P_1)=\ppr$ and 
$\keyword{participant}(P_2)=\ppr$,
and $\keyword{participant}(P)$ be undefined otherwise.

\begin{definition}[Uniform process]\label{def:uniform-p}
A sum process $P_1+P_2$ is uniform, noted $U(P_1+P_2)$, if all these conditions hold:
\begin{enumerate} 
\item $\keyword{nodup}(\keyword{labels}(P_1+P_2))$
\item $\keyword{polarity}(P_1 + P_2)\downarrow$
\item $\keyword{participant}(P_1 + P_2)\downarrow$.
\end{enumerate}
\end{definition}

\begin{figure}[t]
    \begin{mathpar}
      \inference[\textsc{WbP-Inp}]{
      \keyword{wb}(P)
      }{
      \keyword{wb}( \PIN \ppr l x P)
      }
      \and
      \inference[\textsc{WbP-Out}]{
      \keyword{wb}(P)
      }{
      \keyword{wb}(\POUT \ppr l e P)
      }
      \and
      \inference[\textsc{WbP-Sum}]{
        U(P_1 + P_2) 
        \quad \keyword{wb}(P_1)
        \quad \keyword{wb}(P_2)
        }{
        \keyword{wb}(P_1 + P_2)
        }
        \and 
        \inference[\textsc{WbP-If}]{
        \keyword{wb}(P)\quad
        \keyword{wb}(Q)
        }{
        \keyword{wb}(\ifk\; e \; \thenk\; P\; \elsek\; Q)
        }
        \and
      \inference[\textsc{WbP-End}]{}{\keyword{wb}(\INACT)} 
      \and
      \inference[\textsc{WbP-Rec}]{
      \keyword{wb}(P)}{\keyword{wb}(\mu \chi.P)}  
      \and
      \inference[\textsc{WbP-Var}]{}{\keyword{wb}(\chi)}
    \end{mathpar}  
      \caption{Well-behaved processes}\label{fig:wb-p}  
      \end{figure} 

\begin{definition}[Well-behaved processes]
A process is well-behaved, noted $\wb(P)$, when one  
of the rules of Figure~\ref{fig:wb-p} applies to $P$.
\end{definition}    

\begin{definition}[Well-formed processes and sessions]
  A process $P$ is well-formed, noted $\wf(P)$,
  when is closed, and well-behaved. 
  A session $\PARI{\ppp_i\lhd P_i}$ is well-formed, noted $\wf(\PARI{\ppp_i\lhd P_i})$,
   whenever for all $i\in I$ we have $\wf(P_i)$. 
  \end{definition} 

  \begin{example}\label{ex:wfp}
  Consider process $P_0 = P + \ifk\, \falsek \, \thenk\, Q\, \elsek\, R$, where  
  $P\eqdef \PIN \ppr {l_1} x \INACT$, $Q \eqdef\PIN \ppr {l_2} x \INACT$, and 
$R \eqdef \PIN \ppr {l_3} x \INACT + \PIN \ppr {l_4} x \INACT$. We have 
$\wf(P_0)$.   
  Process $P_1\eqdef \mu\chi.R_1 + R_2$ is not well-formed, because is not uniform and in turn is not 
  well-behaved.  
  \qed
  \end{example}  

  The following lemmas are used to establish that type checking enforces well-formedness of closed sessions 
  (cf.~Corollary~\ref{cor:type-wfp}).

  \begin{lemma}\label{lem:uniform-corr}
  If $\Gamma\vdash P_1 + P_2: T_1 + T_2$ and $U(T_1 + T_2)$ then
  $U(P_1 + P_2)$.
  \end{lemma}

  \begin{lemma}
  If $\Gamma\vdash P\colon T$ and $\wb(T)$ then $\wb(P)$.
  \end{lemma}
  \begin{proof}
  We proceed by induction on $\Gamma\vdash P\colon T$.

  \begin{description}
    \item[\textsc{T-End}] We apply \textsc{WbP-End} and obtain $\wb(\INACT)$.
    \item[\textsc{T-Rec}] We apply Lemma~\ref{lem:wf-preserve} and obtain the I.H.: $\wb(P)$. 
    We need to prove that:
    \begin{align}
    \textbf{GOAL:} \qquad \wb(P\{\mu \chi.P/\chi\}) \label{eq:eq-wb-pres-1}
    \end{align}
  The result $\wb(P)$ implies $\wb(P\{\mu \chi.P/\chi\})$ can be easily shown 
  by structural induction on $P$.
  In case  $P=\chi_1$ we need to show that 
  $\wb(\chi_1\{\mu \chi.\chi_1/\chi\})$. There are two cases corresponding to 
  (1) $\chi_1=\chi$ and (2) $\chi_1\ne\chi$.
  In case (1) we infer $\wb(\mu \chi_1.\chi_1)$ by applying \textsc{WbP-Var},\textsc{WbP-Rec}.
  In case (2) we have $\wb(\chi_1)$ by applying \textsc{WbP-Var}.
  \item[\textsc{T-Var}] We apply \textsc{WbP-Var} and obtain $\wb(\chi)$.
  \item[\textsc{T-InP}] By $\wb(\pp r ?l(S).T)$ and inversion of \textsc{Wb-Inp} in Figure~\ref{fig:wb} 
  we infer $\wb(T)$.
  The I.H. is $\wb(P)$. We apply  \textsc{WbP-Inp}  to the I.H. and obtain  
  $\wb(\pp r ?l(x).P)$.
  \item[\textsc{T-Out}] Analogous to \textsc{T-InP}.
  \item[\textsc{T-Sum}] By $\wb(T_1+T_2)$ and inversion of \textsc{Wb-Sum} in Figure~\ref{fig:wb} 
  we infer $U(T_1+T_2)$ and $\wb(T_1)$ and $\wb(T_2)$.
  The I.H. is $\wb(P_1)$ \emph{and} $\wb(P_2)$. In order to infer $\wb(P_1 + P_2)$ by means of
  \textsc{WbP-Sum}, we need to establish  $U(P_1 + P_2)$.
  We use Lemma~\ref{lem:uniform-corr} and we conclude.
  \item[\textsc{T-Sum-L}] By $\wb(T_1+T_2)$ and inversion of \textsc{Wb-Sum} in Figure~\ref{fig:wb} 
  we infer $\wb(T_1)$.
  By I.H. we have $\wb(P)$.
  \item[\textsc{T-Sum-R}] By $\wb(T_1+T_2)$ and inversion of \textsc{Wb-Sum} in Figure~\ref{fig:wb} 
  we infer $\wb(T_2)$.
  By I.H. we have $\wb(P)$.
  \item[\textsc{T-If}] By I.H. we have $\wb(P)$ and $\wb(Q)$.
  We apply \textsc{WbP-If} and infer $\wb(\ifk\; e$  $\thenk\; P\; \elsek\; Q )$.\qed
\end{description} 
  \end{proof}

  \begin{corollary}[Typing well-formed sessions]\label{cor:type-wfp}
  Let $\cal M$ be closed. 
  If $\Gamma\Vdash{\cal M}\colon \Delta$ then $\wf(\cal M)$.
  \end{corollary}
 
\subsection{Progress}

\begin{definition}[Ended session]
A well-formed session $\PARI{\ppp_i\lhd P_i}$ is ended, noted $\keyword{Ended}(\PARI{\ppp_i\lhd P_i})$,
 when for all $i\in I$ we have $P_i =\INACT$.
\end{definition}    

\begin{lemma}[Endness typing]\label{lem:end-typing}
Let $\Gamma\Vdash{\cal M}\colon \Delta$ with $\cal M$ closed. If $\keyword{consumed}(\Delta)$ 
then $\keyword{Ended}(\cal M)$ or there is ${\cal M}'$ s.t. ${\cal M} \lts\tau {\cal M'}$. 
\end{lemma}
\begin{proof}
Let ${\cal M}=\PARI{\ppp_i\lhd P_i}$ for some $I={1,\dots, n}$, $n\geq 1$.
First note that $\cal M$ is well-formed, by Corollary~\ref{cor:type-wfp}.

We need to show that for all $i\in I$ 
we have $P_i=\INACT$, 
or there is ${\cal M}'$ s.t. $\PARI{\ppp_i\lhd P_i}\lts\tau{\cal M}'$.  

We proceed by induction on $n$.

\begin{description}
\item[(n=1)] 
By inversion of $\Gamma\Vdash{\ppp_1\lhd P_1}\colon \Delta$ we obtain
\begin{align}
  &\Gamma\Vdash {\pp p}_1\lhd P_1\colon {\pp p}_1\colon \End
  &&\Delta = {\pp p}_1\colon \End 
  \label{eq:end-1}
  \end{align}

  By inversion of (\ref{eq:end-1}), we obtain 
  $\Gamma\vdash P_1\colon \End$. 
  A case analysis on the rules in Figure~\ref{fig:typing} show that (i) $P_1=\INACT$ or 
  (ii) $P_1 = \ifk\, e \, \thenk\, Q_1\, \elsek\, Q_2$, for some $e$, $Q_1$ and $Q_2$.
  
  Case (i) is clear.
  
  Case (ii) allow us to deduce that there is ${\cal M}'$ s.t. ${\ppp_1\lhd P_1} \lts\tau {\cal M'}$.
  There are two sub-cases corresponding to $e\downarrow \truek$ and $e\downarrow \falsek$.
  In the former case, we have ${\cal M}'={\pp p}_1\lhd Q_1$;
  in the latter case, we have ${\cal M}'={\pp p}_1\lhd Q_2$.

\item[(n$>$1)]

  The I.H. is 
  \begin{align}    
  &\forall 1\leq k\leq n-1\,.\, P_k =\INACT &\text{ or }
  \label{eq:end-2}
  \\
  &\exists {\cal M}'\,.\,\parallel_{i\in I\less n}\hspace{.125em}{\ppp_i\lhd P_i}\lts\tau{\cal M}'
  \label{eq:end-3}
\end{align}
  
  Consider $\Gamma\vdash P_n\colon \End$. 
  A case analysis on the rules in Figure~\ref{fig:typing} show that (i) $P_n=\INACT$ or 
  (ii) $P_n = \ifk\, e \, \thenk\, Q_1\, \elsek\, Q_2$, for some $e$, $Q_1$ and $Q_2$.

  \medskip
  In case (i) we have two sub-cases corresponding to (\ref{eq:end-2}) and (\ref{eq:end-3}).
  In sub-case (\ref{eq:end-2}) we have $\keyword{Ended}(\PARI{\ppp_i\lhd \INACT})$, and we have done.
  In sub-case (\ref{eq:end-3}) we invert the reduction and find that a rule 
  $\textsc r\in \{ 
  \textsc{R-Rec}, \textsc{R-IfT},\textsc{R-IfF}\}$ has been applied. 
  We then build a derivation ending with rule \textsc r:  
  $$
  \parallel_{i\in I\less n}\hspace{.125em}{\ppp_i\lhd P_i}\parallel \ppp_n\lhd \INACT\lts\tau
  {\cal M}'\parallel \ppp_n\lhd \INACT
  $$ 

  \medskip
  In case (ii) we find ${\cal M}'$ s.t. 
  $$
  \parallel_{i\in I\less n}\hspace{.125em}{\ppp_i\lhd P_i}\parallel \ppp_n\lhd P_n\lts\tau
  {\cal M}'
  $$

  There are two sub-cases corresponding to $e\downarrow \truek$ and $e\downarrow \falsek$.
  In the former case, we have ${\cal M}'= 
  \parallel_{i\in I\less n}\hspace{.125em}{\ppp_i\lhd P_i}\parallel{\pp p}_n\lhd Q_1$.
  In the latter case, we have ${\cal M}'= 
  \parallel_{i\in I\less n}\hspace{.125em}{\ppp_i\lhd P_i}\parallel{\pp p}_n\lhd Q_2$.\qed
\end{description}
\end{proof}

\begin{definition}[Unblocked session]
A well-formed session $\cal M$ is unblocked, noted $\keyword{Unblocked}(\cal M)$, whenever:
$$
(\forall {\cal M}'\,.\, {\cal M}\nlts\tau {\cal M}') \land
(\forall l,\ppp,\ppq, 
{\cal M}'\,.\, {\cal M}\nlts{\ASYN l \ppp \ppq} {\cal M}')    
\Longrightarrow \keyword{Ended}(\cal M) 
$$
\end{definition}    

\begin{definition}[Silent session]
A well formed process $P$ is silent if 
there is $Q$ s.t. $P= \mu \chi.Q$,
or there are $e, P_1, P_2$ s.t. $P = \ifk\, e \, \thenk\, P_1\, \elsek\, P_2$,
or there are $P_1$, $P_2$ s.t. $P = P_1 + P_2$ and exists $i\in{1,2}$ s.t. 
$P_i$ is silent.
A well formed session $\PARI{\ppp_i\lhd P_i}$ is silent   if there exists $i\in I$ s.t. 
$P_i$ is silent.
\end{definition}  

\begin{example}
Consider $P + \ifk\, \falsek \, \thenk\, Q\, \elsek\, R$:
if the process is well-formed, then  it is silent.
For instance take $P_0$ in Example~\ref{ex:wfp}:
$$
\PIN \ppr {l_1} x \INACT + \ifk\, \falsek \, \thenk\, \PIN \ppr {l_2} x \INACT\, \elsek\, 
(\PIN \ppr {l_3} x \INACT + \PIN \ppr {l_4} x \INACT)
$$
Process $P_0$ is well-formed and in turn it is silent. We note that, by applying rules 
\textsc{R-Sum-R,R-IfF}, we have
$$
\ppp\lhd P_0\lts\tau 
\ppp\lhd \PIN \ppr {l_3} x \INACT + \PIN \ppr {l_4} x \INACT
$$
\qed
\end{example}  

\begin{definition}[Prompt session]
  A well formed session $\PARI{\ppp_i\lhd P_i}$ is prompt if 
  there is $\{j,k\}\subseteq I$, $l,v, P', P''$ s.t. 
  $$
  \ppp_j\lhd P_j\lts{\AIN {\ppp_k} l v} \ppp_j\lhd P' \text{ and } 
  \ppp_k\lhd P_k\lts{\AOUT {\ppp_j} l v} \ppp_k\lhd P''
  $$
  \end{definition}

\begin{lemma}[Stuckness inversion]\label{lem:stuck-inv}
Let $\wf(\PARI{\ppp_i\lhd P_i})$.
The following hold.
\begin{enumerate}
    \item $\forall {\cal M}'\,.\, {\PARI{\ppp_i\lhd P_i}}\nlts\tau {\cal M}'$ implies that 
    $\PARI{\ppp_i\lhd P_i}$ is not silent
    \item $\forall l,\ppp,\ppq, 
{\cal M}'\,.\, \PARI{\ppp_i\lhd P_i}\nlts{\ASYN l \ppp \ppq} {\cal M}'$ implies that 
$\PARI{\ppp_i\lhd P_i}$ is not prompt
\end{enumerate}
\end{lemma}   

\begin{lemma}[Stuckness typing]\label{lem:stuck-typing}
  Let $\Gamma\Vdash{\cal M}\colon \Delta$ with $\cal M$ closed. If ${\cal M}$ is both not silent and not prompt 
  then for all  $\alpha$ and $\Delta'$ we have $\Delta\nlts\alpha\Delta'$.
  \end{lemma} 
\begin{proof}
We show that the hypothesis $\exists \alpha,\Delta_1$ s.t. 
\begin{align} 
  \Delta\lts\alpha \Delta_1
  \label{eq:stuck-typ-1}
\end{align} gives rise
to a contradiction.

Let ${\cal M}=\PARI{\ppp_i\lhd P_i}$ for some $I={1,\dots, n}$, $n\geq 1$. 

By inversion of (\ref{eq:stuck-typ-1}), we have two sub-cases corresponding to \textsc{Se-Rec} (\ref{eq:stuck-typ-2})
and 
\textsc{Se-Com} (\ref{eq:stuck-typ-3}), respectively (cf.~Figure~\ref{fig:lts-types}):
\begin{align}
  &\Delta\lts{\tau}\Delta_1  
  \label{eq:stuck-typ-2}
  \\
      &\Delta\lts\alpha \Delta_1  
      &\alpha\in\{\sync l \ppr \ppq, \sync l \ppq \ppr \}
      \label{eq:stuck-typ-3}
\end{align}

\begin{description}

\item[(\ref{eq:stuck-typ-2})]
We invert the transition and obtain that there are $T,T'$ s.t.
$\Delta(\ppr)=T$ and $T\lts\tau T'$.
 By inverting $T\lts\tau T'$, we obtain that 
$T =\mu X.R$, for some $R$. 

By inversion of $\Gamma\Vdash\PARI{\ppp_i\lhd P_i}\colon \Delta$, we obtain that 
there is there is $j\in I$ s.t. $\ppr = \ppp_j$, and in turn 
$\Gamma\vdash P_j\colon T$. 
A case analysis on the typing rules in Figure~\ref{fig:typing} shows that $P_j = \mu \chi.P'$, for some $P'$,
or $P_j=\ifk\, e \, \thenk\, P\, \elsek\, Q$, for some $P, Q$.
Note indeed that the case $P_j=\chi$ is excluded, because 
$\cal M$ is well-formed.
Thus $\PARI{\ppp_i\lhd P_i}$ is silent, contradiction.

\item[(\ref{eq:stuck-typ-3})]
We show the case $\alpha=\sync l \ppr \ppq$; the case $\alpha=\sync l \ppq \ppr$ is analogous.
We invert the transition and obtain that there are participants $\ppr,\ppq$
and types $T_{\ppr}, T_{\ppq},T'_{\ppr}, T'_{\ppq}$, and  a value $v$
such that $\Delta(\ppr)=T_{\ppr}$ and $\Delta(\ppq)=T_{\ppq}$ and 
\begin{align}
T_{\ppr}\lts{{\ppq}?l(v)} T'_{\ppr}
      \label{eq:stuck-typ-4}
      \\
T_{\ppq}\lts{{\ppr}!l\langle v\rangle}  T'_{\ppq}
      \label{eq:stuck-typ-5}
\end{align}     
By inverting (\ref{eq:stuck-typ-4}), we obtain that there is $S$ s.t. 
$\emptyset\vdash v\colon S$ and 
\\
$T_{\ppr} = \&_{i\in I} \ppq?l_i(S_i).T'_{\ppr}$
(cf.~Definition~\ref{def:notation})
and 
$i\in I$ s.t. $l_i = l$ and $S_i = S'$. 

Similarly, 
by inverting (\ref{eq:stuck-typ-5}) we obtain that there is $S'$ s.t. 
$\emptyset\vdash v\colon S'$ and 
$T_{\ppq}=\oplus_{j\in J} \ppr!l_j(S_j).T'_{\ppq}$ 
and $j\in J$ s.t. $l_j = l$ and $S_j = S'$. 

We note that $\emptyset\vdash v\colon S$ and $\emptyset\vdash v\colon S'$ imply $S = S'$ 
(cf.~Figure~\ref{fig:sorting}).

\medskip
By inversion of $\Gamma\Vdash\PARI{\ppp_i\lhd P_i}\colon \Delta$, we obtain that 
there is there is $\{j,k\}\subseteq I$ s.t. $\ppr = \ppp_j$,$\ppq =\ppp_k$, and in turn 
\begin{align}
\Gamma\vdash P_j\colon T_{\ppr} 
\label{eq:stuck-typ-inp}
\\
\Gamma\vdash P_k\colon T_{\ppq}
\label{eq:stuck-typ-out}
\end{align}
 
 Consider (\textbf{\ref{eq:stuck-typ-inp}}).
 A case analysis on the typing rules in Figure~\ref{fig:typing} shows that 
 there is a continuation $P'$ s.t.
 $P_j = \PIN \ppq l x P'$,  or   
 $P_j = Q_1 + \cdots + \PIN \ppq l x P' + \cdots + Q_n$, for some $Q_1,\dots, Q_n$, or 
 there are $P'', Q''$ s.t. 
 $P_j = \ifk\, e \, \thenk\, P''\, \elsek\, Q''$
 or 
 $P_j = Q_1 + \cdots + \ifk\, e \, \thenk\, P''\, \elsek\, Q'' + \cdots + Q_n$,  
 for some $Q_1,\dots, Q_n$.

\medskip
Consider (\textbf{\ref{eq:stuck-typ-out}}).
Similarly, we infer that 
 there are an expression $e$ and a continuation $P'$ s.t. 
 $P_k = \POUT \ppr l e P'$,  or   
 $P_j = Q_1 + \cdots + \POUT \ppr l e P' + \cdots + Q_n$, for some $Q_1,\dots, Q_n$, or
 there are $P'', Q''$ s.t. 
 $P_k = \ifk\, e \, \thenk\, P''\, \elsek\, Q''$
 or 
 $P_k = Q_1 + \cdots + \ifk\, e \, \thenk\, P''\, \elsek\, Q'' + \cdots + Q_n$,  
 for some $Q_1,\dots, Q_n$.
 
\medskip
The analysis of all possible combinations of the sub-cases highlighted above 
show that two possibilities arise.

Case (1) below corresponds to the sub-cases whether the label is not guarded by 
$\ifk \text{-} \thenk \text{-} \elsek$, both in input and output.

Case (2) describes all the remaining sub-cases: 

\begin{enumerate}
\item $\exists {\cal M'}\,.\, {\cal M}\lts{\ASYN l \ppr \ppq}{\cal M}'$
\item $\exists {\cal M'}\,.\, {\cal M}\lts{\tau}{\cal M}'$
\end{enumerate}  

Thus one of the following holds: $\cal M$ is prompt, or $\cal M$ is silent, contradiction.
\qed
\end{description}
\end{proof}  

\paragraph{Proof of Theorem~\ref{th:progress}.}
Let ${\cal M}=\PARI{\ppp_i\lhd P_i}$ for some $I={1,\dots, n}$, $n\geq 1$.
From ${\cal M}$ closed, $\Gamma\Vdash {\cal M}\colon\Delta$, and Corollary~\ref{cor:type-wfp} we infer
$\wf({\cal M})$.

The structure of the proof is the following.
We use Lemma~\ref{lem:end-typing}, Lemma~\ref{lem:stuck-inv} and  Lemma~\ref{lem:stuck-typing} and show that 
$\Gamma\Vdash{\cal M}\colon\Delta$ and the hypotheses
\begin{align}
&\forall {\cal M}'\,.\, \PARI{\ppp_i\lhd P_i}\nlts\tau {\cal M}' 
\label{eq:prog-1}
\\
&\forall l,\ppp,\ppq, 
{\cal M}'\,.\, \PARI{\ppp_i\lhd P_i}\nlts{\ASYN l \ppp \ppq} {\cal M}'
\label{eq:prog-2}
\end{align}
imply:
$$
\textbf{GOAL}:\quad 
\keyword{Ended}(\PARI{\ppp_i\lhd P_i})
$$

To show that the implication holds, we start by applying Lemma~\ref{lem:stuck-inv} 
to the hypothesis (\ref{eq:prog-1})
and obtain that $\PARI{\ppp_i\lhd P_i}$ is not silent.
By applying Lemma~\ref{lem:stuck-inv} to the hypothesis (\ref{eq:prog-2})
we obtain that $\PARI{\ppp_i\lhd P_i}$ is not prompt.

We apply Lemma~\ref{lem:stuck-typing} to these results and infer that 
forall $\alpha$ and $\Delta'$ 
\begin{align} 
\Delta\nlts\alpha\Delta'\label{eq:progr-stuck}
\end{align}

By inversion of $\Gamma\Vdash{\PARI{\ppp_i\lhd P_i}}\colon\Delta$ we have:
\begin{align}
&\forall i\in I\,.\,\Gamma\Vdash {\pp p}_i\lhd P_i\colon {\pp p}_i\colon T_i
&&\Delta = \bigcup_{i\in I}{\pp p}_i\colon T_i 
\label{eq:progr-delta}
  \\
&\partition\Delta(\Delta_1,\dots,\Delta_k)
&&\forall j\in \{1,\dots, k\}\,.\, \keyword{comp}(\Delta_j) \label{eq:progr-4}
\end{align}

Note that (\ref{eq:progr-stuck}) implies that forall $j\in \{1,\dots, k\}$, $\alpha$, and $\Delta'_j$ we have:
\begin{align} 
  \Delta_j\nlts\alpha\Delta'_j\label{eq:progr-stuck2}
\end{align}

By (\ref{eq:progr-stuck2}) and Definition~\ref{def:closure}, we obtain that for all $j\in \{1,\dots, k\}$, 
$\Omega$, and $D_j$:
\begin{align}
  {\tt stuck}_{\Omega, D_j}(\Delta_j)\label{eq:progr-stuck3}
\end{align}

From this we easily infer that for all $j\in \{1,\dots, k\}$, 
$\Omega$, and $D_j$: 
\begin{align}
  {\tt closure}_{\Omega, D_j}(\Delta_j)=\Delta_j\label{eq:progr-3}
  \end{align}

It turns out that, by Definition~\ref{def:compliance} and by (\ref{eq:progr-3}),
we can rewrite (\ref{eq:progr-4}) as 
\begin{align}
\forall j\in \{1,\dots, k\}\,.\, \Delta_j \text{ is not an error}
\label{eq:progr-5}
\end{align}  

By Definition~\ref{def:error}, $\Delta$ is \textbf{an error} whenever $\Delta$ 
is a communication mismatch \textbf{or} $\Delta$ is a deadlock.
Thus $\Delta$ is \textbf{not an error} when $\Delta$ is not a communication mismatch 
\textbf{and} $\Delta$ is not a deadlock.

By Definition~\ref{def:deadlock}, $\Delta$ is \textbf{a deadlock} whenever 
there not exists $\alpha, \Delta'$ such that 
$\Delta\lts\alpha\Delta'$ 
\textbf{and} $\neg\keyword{consumed}(\Delta)$.
Thus $\Delta$ is \textbf{not a deadlock} when there exists $\alpha, \Delta'$ such that 
$\Delta\lts\alpha\Delta'$  
\textbf{or} $\keyword{consumed}(\Delta)$.

We can weaken (\ref{eq:progr-5}) as follows, where we ignore the information that
$\Delta_k$ is not a communication mismatch, since we do not need it.

\begin{align}
  \forall j\in \{1,\dots, k\}\,.\, 
  (\exists \alpha, \Delta'\,.\, \Delta_j\lts\alpha\Delta') \vee  
  \keyword{consumed}(\Delta_j)
  \label{eq:progr-6}
  \end{align}  
  
Next, we show that  $\forall j\in\{1,\dots, k\}\,.\,  
\neg(\exists \alpha, \Delta'\,.\, \Delta_j\lts\alpha\Delta')$.
In fact, for all $\alpha,\Delta_j'$, if 
$\Delta_j\lts\alpha\Delta_j'$ then 
there is $\Delta'$ such that $\Delta\lts\alpha\Delta'$:
this contradicts (\ref{eq:progr-stuck}).

We can thus further simplify (\ref{eq:progr-6}):
\begin{align}
  \forall j\in \{1,\dots, k\}\,.\, \keyword{consumed}(\Delta_j)
  \label{eq:progr-7}
  \end{align} 
From  (\ref{eq:progr-7}) and (\ref{eq:progr-delta}) we deduce that 
\begin{align}
  \Delta = \bigcup_{i\in I}{\pp p}_i\colon \End\label{eq:progr-8}
\end{align}

An application of Lemma~\ref{lem:end-typing} to (\ref{eq:progr-8}) let us infer:

$$
\keyword{Ended}(\PARI{\ppp_i\lhd P_i}) \vee
\exists {\cal M}'\,.\, \PARI{\ppp_i\lhd P_i} \lts\tau{\cal M}'
$$
Since $\exists {\cal M}'\,.\, \PARI{\ppp_i\lhd P_i} \lts\tau{\cal M}'$ contradicts 
(\ref{eq:prog-1}), we deduce $\keyword{Ended}(\PARI{\ppp_i\lhd P_i})$, as required.
\qed

\section{Implementation}
\label{sec:why3}
This section contains material 
used in the deductive verification of the implementation of 
\emph{compliance} in OCaml/Cameleer\cite{PereiraR20} discussed in~\S~\ref{sec:implementation}.
In addition, in \S~\ref{sec:decidability} we show that \emph{type checking is decidable}.

For what concerns \emph{compliance}, in particular, 
we first present meta-code of the tasks implementing function \verb|cstep|,
or \emph{closure} of environments, and then the source code of fixed points. 
\begin{ocamlsc}
  let[@ghost] rec cstep (o : oracle) (d : typEnv list) 
    ((w : typEnvRedexes)[@ghost]) (delta: typEnv)
    ((history : typEnv list)[@ghost]):typEnv = $\cdots$
  \end{ocamlsc}

In the remainder of this section, we use the notation ${\cal H},\Delta$ to indicate that $\Delta$ 
is the last environment visited in the history ${\cal H}\cup\{\Delta\}$.
This is consistent with the implementation, where the history $\cal H$ is deployed
as a list of environments: that is, ${\cal H},\Delta$ is implemented as ${\cal H}@ [\Delta]$.

\subsection{Meta-code of the Implementation of Closure in~Definition~\ref{def:closure}}
\label{sec:cstep}
Function \verb|cstep| performs the \emph{tasks} of Figure~\ref{fig:tasks} by relying on the exception 
\verb|WrongBranch|, which  is raised when two branches have different labels.
Let $*\in\{!, ?\}$, $\dual !=?$, and $\dual ?=!$;
let $\Delta[\pp p\mapsto T]= \Delta\less{{\pp p}\colon\Delta({\pp p})}, \pp p\colon T$.

  Task 0 checks that the environment received in input is minimal, thus implementing 
  the minimal partition requirements in Figure~\ref{fig:typing} (Rule \textsc{T-Ses}).
  Tasks (1) and (2), when they succeed, implement the first two pre-conditions of 
rules \textsc{Se-Com-D}  and \textsc{Se-Rec-D} of Figure~\ref{fig:lts-types-det}.
In particular, rule \textsc{Se-Com-D} is matched when the oracle selects the communicating participants 
$\pp p, \pp q$, and we move to task (3), 
and rule \textsc{Se-Rec-D} is matched when the oracle selects a single participant 
doing an internal transition, and we move to task (4).
The cases whether (0), or (1), or (2) \emph{fail} are tackled in the next paragraph.

Task (3) does pattern matching on the types $T_{\pp p}$ and $T_{\pp q}$ of \ppp\ and  \ppq\ in $\Delta$, 
respectively, and 
implement the remaining pre-conditions of rule \textsc{Se-Com-D}, 
assuming a syntactic order of labels (i.e. the order in which they occur in types).

Task (3.a), when succeeds, matches the step  in the premises of rule \textsc{Se-Com-D},
that is:
\begin{align*}
   D\diamond\Delta, {\pp p}: T_{\pp p},  {\pp q}: T_{\pp q} &\alts{l@{\pp p}\bowtie{\pp q}} 
   D\less\Delta\diamond\Delta, {\pp p}: T'_{\pp p},  {\pp q}: T'_{\pp q}
   \blacktriangleright   \nabla^\circ \text{ or }
   \\
   D\diamond\Delta, {\pp p}: T_{\pp p},  {\pp q}: T_{\pp q} &\alts{l@{\pp q}\bowtie{\pp p}} 
D\less\Delta\diamond\Delta, {\pp p}: T'_{\pp p},  {\pp q}: T'_{\pp q}
\blacktriangleright   \nabla^\circ
\end{align*}   
The step is inferred by means of rule \textsc{Se-Com-A} of Figure~\ref{fig:alg-lts-types}:
the sum continuation is the environment placeholder $\nabla^\circ$ because
both types are singletons.

Task (3.b) first verifies that the types of \ppp\ and \ppq\ in $\Delta$ do not mismatch
and second uses exception handling to implement the sum continuation mechanism of  
Figure~\ref{fig:alg-lts-types}.
The predicate $\keyword{mismatch}_2(T_1, T_2)$ is asserted when 
the instantiation of the existential parameters in Definition~\ref{def:mismatch} 
with~\ppp\ and \ppq\ s.t. $\Delta(\ppp)=T_1$ and $\Delta(\ppq)=T_2$ is a tautology.
When the test fails, we raise \verb|WrongBranch| by passing the history ${\cal H},\Delta$:
the exception can be caught by the caller, as we explain below.

\begin{figure}[t]
  \begin{enumerate}
   \setcounter{enumi}{-1}
  \item If $\Delta$ is minimal then go to 1 else \emph{exit 0}
  \item If $\Delta\in D$  then go to 2 else \emph{exit 1}
  \item If $\Omega(\Delta)=(\ppp, \ppq)$ \\
  then go to 3 \\
  else if $\Omega(\Delta)=\ppp$ then go to 4 else \emph{exit 2}
  \item Match $\Delta(\ppp), \Delta(\ppq)$ with
  
  \begin{enumerate}
  \item $\ppq{*}l_1(S_1).T_1$, $\ppp{\dual *} l_2(S_2).T_2$:\\
  if 
  $l_1= l_2$ and $S_1 = S_2$ 
  \\
  then 
  ${\tt cstep}(\Omega,D\less\Delta, {\cal W}, \Delta[\ppq\mapsto T_1,\ppp\mapsto T_2], ({\cal H},\Delta))$ \\
  else 
  raise \verb|WrongBranch|$(({\cal H},\Delta))$
  
  \item $T_1 + T_2$, $T_3$: \\
  if not $\keyword {mismatch}_{2}(\Delta(\ppp), \Delta(\ppq))$
  \\
  then \emph{handle sum}\\
  else raise \verb|WrongBranch|$(({\cal H},\Delta))$
  
  \item $T_1$, $T_2 + T_3$: specular to (b)
  
    \item $\ppq{*}l_1(S_1).T_1$, $\ppp{*} l_2(S_2).T_2$: raise \verb|IncompatiblePair|
    \item $T_1, T_2$: raise \verb|OracleNotFair|
  \end{enumerate}
  \item Match $\Delta(\pp p)$ with
  \begin{enumerate}
  \item $\mu X.T$: 
  ${\tt cstep}(\Omega, D\less\Delta, {\cal W}, \Delta[\pp p\mapsto T\{\mu X.T/X\}], ({\cal H},\Delta))$
  \item $T$: raise \verb|OracleNotFair|
  \end{enumerate} 
  \end{enumerate}
  \caption{Implementation of closure}\label{fig:tasks}  
  \end{figure}

\begin{figure}[!h]
  try
     ${\tt cstep}(\Omega, D\less\Delta, {\cal W}, \Delta[\pp p\mapsto T_1], ({\cal H},\Delta))$
    with
    \\
\verb|WrongBranch|$(\_)$, \verb|Fixpoint|$(\_)$: 
      ${\tt cstep}(\Omega, D\less\Delta, {\cal W}, \Delta[\pp p\mapsto T_2], 
 ({\cal H},\Delta))$
  \caption{Exception handling}\label{fig:exceptions}  
  \end{figure}
The code of \emph{handle sum} is in Figure~\ref{fig:exceptions}.
To illustrate the exception handling mechanism, let 
$$
\Delta\eqdef {\ppp}: \ppq!l_1(S).T_1 + (\ppq!l_2(S).T_2+\ppq!l_3(S).T_3),  
{\ppq}: \ppp?l_2(S).T'_2
$$ 
and consider the synchronisation 
$$
\Omega\triangleright D\diamond \Delta\dlts{l_2@{\pp p}\bowtie{\pp q}} 
D\less\Delta\diamond{\pp p}: T_2,  {\pp q}: T'_2
\blacktriangleright   \Delta_1
$$
inferred by using rules \textsc{Se-Com-D} of Figure~\ref{fig:lts-types-det} and
\textsc{Se-Com-A} of Figure~\ref{fig:alg-lts-types}.

Invoking \verb|cstep| on the environment~$\Delta$ leads to  Task (3.b) 
of Figure~\ref{fig:tasks}, and in turn to the check  
$\keyword {mismatch}_{2}(\Delta(\ppp), \Delta(\ppq))$, which fails:
the code in Figure~\ref{fig:exceptions} try to execute \verb|cstep| by passing the environment
${\ppp}: \ppq!l_1(S).T_1,  
{\ppq}: \ppp?l_2(S).T'_2$.

This time, the check  
$\keyword {mismatch}_{2}(\ppq!l_1(S).T_1, \ppp?l_2(S).T'_2)$ goes through:
the code raises the exception \verb|Wrongbranch|, which is catched 
(cf.~Figure~\ref{fig:exceptions}) and leads to the call
of \verb|cstep| on the environment~$\ppp:\ppq!l_2(S).T_2+\ppq!l_3(S).T_3, \ppq:\ppp?l_2(S).T'_2$.

After verifying $\neg(\keyword {mismatch}_{2}(\ppp:\ppq!l_2(S).T_2+\ppq!l_3(S).T_3, \ppp?l_2(S).T'_2))$,
\verb|cstep| is invoked on $\ppp:\ppq!l_2(S).T_2, \ppq:\ppp?l_2(S).T'_2$ 
(cf. Figure~\ref{fig:exceptions}):
in this case Task (3.a) in Figure~\ref{fig:tasks} is performed and in turn 
we call \verb|cstep| by 
passing the environment $\ppp:T_2, \ppq:T'_2$, thus mimicking the labeled transition above.

Task (3.d) rejects immediately the environment by raising an uncaught exception,
because the participants selected by the oracle
have types with the same polarity.

Task (3.e)
is listed for completeness, but is never matched by 
the types of two participants returned by the oracle received in input, because 
in the function specification in Figure~\ref{fig:contract} we require that the oracle
 is fair (cf.~Definition~\ref{def:fair}).

Task (4) verifies that the selected participant \pp p has a recursive type:
if this is the case, then it recursively calls \verb|cstep| by unfolding the 
type of \pp p, otherwise it exits raising an uncaught exception (cf. Task (3.e)).

\myparagraph{Ghost parameters and exceptions}
Definition~\ref{def:exits} lists the positive exits 
corresponding to (1) termination and (2) raising the \emph{exception} \verb|Fixpoint|:
(1) returns a \emph{consumed} environment (cf.~Definition~\ref{def:deadlock}); 
(2) is raised when a fixed point is encountered and carries the history ending with the environment
occurring twice in the history.

The ghost parameter $\cal W$ represents the \emph{fixpoint} and is used to verify that the environments passed 
to the recursive calls of the function are contained in the fixed point of the initial 
environment: that is, if 
${\tt cstep}(\Omega, D, {\cal W}, \Delta, {\cal H})$ spawns a call
${\tt cstep}(\Omega, D', {\cal W}, \Delta', {\cal H}')$ then
$\Omega\triangleright D\diamond\Delta\lts{\alpha_1}\cdots \lts{\alpha_n}D'\diamond \Delta'$.
The ghost parameter $\cal H$ represents the \emph{history} and is used in conjunction with the parameter $D$ 
to partition the set of the combinations spawned from $\cal W$ (cf.~specification of \verb|cstep| in Figure~\ref{fig:contract}).

We can now analyse the exits of function \verb|cstep| referred in Figure~\ref{fig:tasks}:

\begin{description}
   \item[Exit  0]
   raise \verb|NotMinimal|$(({\cal H},\Delta))$
   
   \item[Exit  1]
If  $\Delta\in\cal H$ then raise \verb|Fixpoint|$(({\cal H},\Delta))$
 else raise \verb|DecrNotFixpoint|$(({\cal H},\Delta))$

 \item[Exit 2]
  If not $\keyword{consumed}(\Delta)$
  then raise \verb|DeadLock|
  else return $\Delta$
\end{description}

\emph{Exit 0} deals with the notion of minimal environments (Definition~\ref{def:partition}) used 
in rule \textsc{T-Ses} of the type system in Figure~\ref{fig:typing}:
in fact, the inclusion of the predicate in the body of the function allows for simplifying
the presentation of  \textsc{T-Ses} by removing the minimal partition assumption.
We will present the simplified system in a forthcoming journal publication.

\emph{Exit 1} deals with the fixed point mechanism and is used to verify 
that the environment $\Delta$ was actually previously visited in the history $\cal H$.
This case corresponds to a \emph{positive exit} (cf.~Definition~\ref{def:exits}), meaning 
that compliance holds. 
The else branch should never be taken when the parameter~$\cal W$ is a fixed point containing~$\Delta$.

\emph{Exit 2} occurs when 
the oracle receiving $\Delta$ does not find one or two participants: 
in this case the computation stops.
Still, it could stop (1) \emph{positively} or (2) \emph{negatively}:
\begin{enumerate}
\item $\Delta$ is consumed, that is all participants of $\Delta$ have type $\End$, and we return $\Delta$
and reach a \emph{positive exit} of Definition~\ref{def:exits} leading to asserting compliance
\item $\Delta$ is not consumed, and we raise the uncaught exception \verb|Deadlock|$(({\cal H},\Delta))$
\end{enumerate}

\subsection{Fixed Points in OCaml}\label{sec:fixpoints}
\begin{figure}[!h]
\ocamlfile{mfixpoint.ml} 
\caption{Fixed points in OCaml}\label{fig:fixpoint}
\end{figure}
   
The fixed point mechanism is deployed in Figure~\ref{fig:fixpoint};
the code is accepted by Cameleer~\cite{PereiraR20}, 
which in turn relies on~\cite{ChargueraudFLP19,FilliatreP13}.
The syntax includes  a logical part and a computational part.
We use the convention to use the {\color{dkblue}blue} color  to indicate the parts of the code
that concerns the {\color{dkblue}logical} specification of the program's behaviour.
The tool automatically verifies {\color{dkblue}behavioural specifications} using constraint solvers\cite{alt-ergo,MouraB08,cvc4},
while supporting imperative features, ghost parameters~\cite{FilliatreGP16},
and interactive proofs.

Specifically, the code in Figure~\ref{fig:fixpoint} implements the predicate \keyword{isFix}
used in the pre-condition 
$\keyword{isFix}({\cal W},\Delta,|D|*2)$ 
of Figure~\ref{fig:contract}.
We proceed by analysing the presented code.

The fixed point $\cal W$ received by \verb|cstep| (cf.~\S~\ref{sec:cstep}) has type \verb|typEnvRedexes|, which maps participants to type redexes, 
which in turn are collections of types.

Function \verb|subsT| find the occurrences of $X$ in $T$ and substitutes them with $\mu X.T$:
that is, \verb|subsT|$(t, x, t)$ is $T\{\mu X.T/X\}$.
We use function $\mathsterling$ to perform the key-operation $\mathsterling \ x\ T = T\{\mu X.T/X\}$.

Function \verb|produceRedexes| returns a fixed point of a type $T$ up-to a depth~$n$.
When $n$ is non-positive, the fixed point is empty.
Otherwise, when $n>0$, we include $T$ and the result of the call of 
\verb|produceRedexes| by passing the continuation of $T$ and $n-1$.
In particular, the continuation mechanism launches calls for both components $T_1$, $T_2$
  of a sum type $T_1 + T_2$,
and launches a call for $T\{\mu X.T/X\}$ when it encounters the type $\mu X.T$.

Predicate \verb|fixpoint| holds for redexes ${\cal T}\eqdef \{T_1,\dots, T_m\}$  and a type $T$ and an integer 
$n$ when \verb|produceRedexes|$(T, n)$ contains closed types included in $\cal T$.

Predicate \verb|mfixpoint| corresponds to \keyword{isFix} 
of Figure~\ref{fig:contract}: 
\verb|mfixpoint|$({\cal W},\Delta,n)$ holds whenever $\dom(\Delta)\subseteq\dom(\cal W)$
and for all $\ppp\in\dom(\Delta)$
s.t. $\Delta(\ppp)= T$ and ${\cal W}(\ppp)={\cal T}$ we have \verb|fixpoint|$({\cal T}, T, n)$.

\subsection{Decidability of Type Checking}\label{sec:decidability}
We verified that type checking is decidable by using the fixed point technique 
introduced in~\S~\ref{sec:fixpoints}.
\begin{ocamlsc}
let[@ghost] rec typesP (decr : typRedexes) ((w_t:  typRedexes)[@ghost])
  ((history : typRedexes)[@ghost]) (g : gEnv) (k : pEnv) (p : proc) 
  (t : typ) : bool = $\cdots$
(*@ m = typesP decr w history g k p t
    requires nodup decr
    requires List.mem t w
    requires fixpoint w t (List.length decr + 1)
    requires tdisjoint decr history
    requires tunion decr history w
    variant List.length decr
    raises DecrNotFixpointT -> true
    raises FixpointT t1 -> true
    raises WrongBranchT -> true *)
 \end{ocamlsc}
 Function \verb|typesP| implements the type system $\Gamma\vdash P\colon T$ of Figure~\ref{fig:typing}
 in the tool of~\cite{PereiraR20}.
 The verification of the behavioural specification in~\cite{PereiraR20} ensures 
 that the function terminates by means of the \keyword{variant} clause.
 The parameters $\cal G$ (denoted as \verb|g|) and 
 $\cal K$ (denoted as \verb|k|) are maps from variables to sorts, and from 
 process variables to types, respectively:
 their union corresponds to the type environment~$\Gamma$.

 The ghost parameters ${\cal W}_T$ (denoted as \verb|w_t|) and $\cal H$ 
 (denoted as \verb|history|)
 do not appear in Figure~\ref{fig:typing}
 and are used to reason on the fixed point mechanism;
 all code referring to ghost parameters should be erased after the proof effort~\cite{FilliatreGP16}.
 While the fixed point mechanism is similar to the one used in the function \verb|cstep|,
 we note that the function  \verb|typesP| considers fixed points and histories that are \emph{type redexes},
while the  function \verb|cstep|  consider environment redexes 
(cf.~\S~\ref{sec:code} and Figure~\ref{fig:fixpoint}).

 Type checking uses the set of types $D$ (denoted as \verb|decr|) as decreasing measure.
 Intuitively, at the beginning of the execution of $\Gamma\vdash P\colon T$ we have $D={\cal W}_T$, where
 ${\cal W}_T$ is a fixed point including all redexes of type $T$ up-to a depth $n$,
 and ${\cal H}=\emptyset$.
 
 For instance, to infer the judgement $\emptyset\vdash P^*_\ppa\colon T^*_\ppa$ discussed 
 in~\S~\ref{ex:auth}, we invoke 
 \begin{center}
\verb|typesP|$({\cal W}_{T^*_\ppa},{\cal W}_{T^*_\ppa},\emptyset,\emptyset,\emptyset,P^*_\ppa, T^*_\ppa)$
 \end{center}
 While we require $T^*_\ppa\in {\cal W}_{T^*_\ppa}$,
 depending on~$n$ we may have 
 $T^{**}_\ppa\in {\cal W}_{T^*_\ppa}$, $\cdots$, $T^{**,\cdots,*}_\ppa\in {\cal W}_{T^*_\ppa}$:
 that is, the natural number~$n$ indicates how many type redexes corresponding to the steps 
 $T_\ppa\lts{\alpha_1}\cdots \lts{\alpha_n}T'$ 
 of Figure~\ref{fig:lts-types}
 are included in 
 the fixed point.
 Among these, there are the unfolding of iso-recursive types $\mu X.T$ into $T\{\mu X.T/X\}$,
 which corresponds to the step $\mu X.T\lts\tau T\{\mu X.T/X\}$.

 If a type redex cannot be found in $D$, nor in the history $\cal H$, then the function raises 
 \verb|DecrNotFixpointT|, which indicates that the depth of the fixed point is too small.
 Otherwise, if  the type redex is not in $D$ but has already visited in the history $\cal H$,
 then we have encountered a fixed point: the function raises \verb|FixpointT| 
 providing a positive answer to type checking.

 The branching mechanism is similar to the one in~\S~\ref{sec:cstep}:
 in order to establish $\Gamma\vdash P\colon T_1+ T_2$  
  (cf. rules~\textsc{T-Sum-L}, \textsc{T-Sum-R}
 of Figure~\ref{fig:typing}), we try the typing $\Gamma\vdash P\colon T_1$ and  
 catch the exception \verb|WrongbranchT| by invoking $\Gamma\vdash P\colon T_2$. 

 \paragraph{Type Checking Sessions}
 The type system $\Gamma\Vdash\PARI{\pp p}_i\lhd 
 P_i\colon {\pp p}_1\colon T_1,\cdots, {\pp p}_n\colon T_n$ 
 of Figure~\ref{fig:typing} is deployed by the call of the boolean function
 \verb|types|:
 \begin{center}
 \verb|types|
 $ (\widetilde\Delta, {\cal W}, \Gamma, \PARI{\pp p}_i\lhd 
 P_i,
 ({\pp p}_1\colon T_1,\cdots, {\pp p}_n\colon T_n))$
 \end{center}
Assume ${\cal W}=\ppp_1\colon \widetilde T_1, \cdots, \ppp_n\colon \widetilde T_n$,
and let $\Gamma_x,\Gamma_\chi$ be the projections of $\Gamma$ to variables and process variables,
respectively.
Let $\widetilde \Delta$ be the set of environments containing all combinations spawned from $\cal W$ 
(cf. function \keyword{comb} in ~Figure~\ref{fig:contract}).

The function  \verb|types| is implemented as follows:
\begin{enumerate} 
\item use the call \verb|typesP|$(\widetilde{T}_i, \widetilde{T}_i,\Gamma_x,\Gamma_\chi,
P_i, T_i)$ to infer the judgments $\Gamma\vdash P_i\colon T_i$,
 for all $i\in\{1,\dots,n\}$;
 \item require $\wf(T_i)$, for all $i\in\{1,\dots,n\}$;
 \item calculate the minimal partition $\{\Delta_1,\dots,\Delta_k\}$ of 
 ${\pp p}_1\colon T_1,\cdots {\pp p}_n\colon T_n$;
\item use the call 
\verb|compliance|$(\Omega, \widetilde\Delta, {\cal W}, \Delta_j)$ 
of~\S~\ref{sec:code}
to
infer $\keyword{comp}(\Delta_j)$, forall $\Delta_j\in \{\Delta_1,\dots,$ $\Delta_k\}$.
\end{enumerate}
The calculation of the minimal partition  is obtained by 
computing
all partitions and by taking the one where all elements are minimal, 
according to Definition~\ref{def:partition}.
Termination is obtained by relying on the termination of steps 1--4.
Hence, function \verb|types| is decidable.
\qed
    
\section{Results Mechanised in Coq}
\label{sec:coq}
\begin{figure}[t]
  \begin{coq}
  Check typ_lts.
    typ_lts : typ -> action -> typ -> Prop
  
  Check typ_close.
    typ_close : typ -> typ -> Prop
  
  Definition contractive (t : typ) :=
    forall t',
      typ_close t t' ->
      (forall a t'', ~typ_lts t' a t'') ->
      t' = typ_end.
          
  Lemma contractive_mu  X:
    contractive (typ_mu X (typ_var X)) -> False.    
  
  Check types.
    types : typ_env -> proc_env -> process -> typ -> Prop   
    
  Lemma types_struct G H M1 M2 U:
    types_session G H M1 U ->
    proc_equiv M1 M2 ->
    types_session G H M2 U.  
      
  Lemma types_substP P Q T U G K X:
    ~In X (bvP P) ->
    closedP Q ->
    K X = Some T ->
    types G K P U ->
    types G (removeE X K) Q T -> 
    types G (removeE X K) (substP P Q X) U.
  \end{coq}
  \caption{Definitions and Lemmas mechanised in Coq}
  \label{fig:coq-results}
  \end{figure}

Figure~\ref{fig:coq-results} presents the mechanisation in Coq~\cite{coq}
of contractive types, of Lemma~\ref{lem:types-struct},
 and of Lemma~\ref{lem:substP}.
All elements are used in the proof of subject reduction in~App.~\ref{sec:sr}.

Lemma~\ref{lem:types-struct} is presented in~\S~\ref{sec:sr-sketch} and says that structural congruence is preserved by typing.

Lemma~\ref{lem:substP} is presented in~App.~\ref{sec:sr} and 
allows for substituting a process variable $\chi$ with a process $\mu \chi.P$, when the types agree.

We outline below a correspondence among the Coq types and the definitions in the paper.

\begin{itemize}
  \item The relation \verb|typ_lts| corresponds to the LTS of types of Figure~\ref{fig:lts-types}

  \item The relation \verb|typ_close| is the reflexive-transitive closure of \verb|typ_lts| 
  
  \item The predicate \verb|contractive| formalises  the one defined in~\S~\ref{sec:syntax}  

\item The predicate \verb|types G K P U| corresponds to the judgement 
$\Gamma,{\cal K}\vdash P\colon U$ in~\S~\ref{sec:typing},
where $\Gamma$ is used for variables and 
$\cal K$ is used for process variables

\item The relation \verb|proc_equiv| corresponds to $\equiv$ in~\S~\ref{sec:lts-process} 

\item Function \verb|substP : proc -> proc -> proc -> proc| corresponds to process substitution 
(cf.~\S~\ref{sec:syntax}) 

\item Function \verb|bvP| returns the bound process variables of process $P$ (cf.~\S~\ref{sec:syntax})

\item The predicate \verb|closedP| corresponds to notion of closed processes (cf.~\S~\ref{sec:syntax}).

\end{itemize}

\end{document}

\typeout{get arXiv to do 4 passes: Label(s) may have changed. Rerun}